\documentclass[11pt, letterpaper, onecolumn, oneside, final]{article}

\usepackage{url}
\usepackage{times}
\usepackage{xcolor}
\usepackage{xspace}
\usepackage{authblk}
\usepackage{amssymb}
\usepackage{amsmath}
\usepackage{enumitem}
\usepackage{graphicx}
\usepackage{makecell}
\usepackage{booktabs}
\usepackage{algorithm}
\usepackage[compact]{titlesec} 
\usepackage[noend]{algpseudocode}
\usepackage[sort,compress]{natbib}
\usepackage[bb=dsserif,bbscaled=1.1]{mathalpha}
\usepackage[font=footnotesize,labelfont=bf]{caption} 
\usepackage[left=1in,right=1in,top=1in,bottom=1in]{geometry} 
\usepackage[bookmarks=true,bookmarksnumbered=true]{hyperref}
\usepackage{amsthm}
\usepackage{cleveref}

\makeatletter
\newcounter{algorithmicH}
\let\oldalgorithmic\algorithmic
\renewcommand{\algorithmic}{%
  \stepcounter{algorithmicH}
  \oldalgorithmic}
\renewcommand{\theHALG@line}{ALG@line.\thealgorithmicH.\arabic{ALG@line}}
\makeatother

\makeatletter
\algnewcommand{\LineComment}[1]{\Statex\hskip\ALG@thistlm $\blacktriangleright$ #1}
\makeatother

\theoremstyle{plain}
\newtheorem{theorem}{Theorem}[section]

\AddToHook{env/definition/begin}{\crefalias{theorem}{definition}}

\AddToHook{env/fact/begin}{\crefalias{theorem}{fact}}

\AddToHook{env/claim/begin}{\crefalias{theorem}{claim}}
\newtheorem{claiminline}{Claim}[theorem]
\newtheorem{lemma}[theorem]{Lemma}
\AddToHook{env/lemma/begin}{\crefalias{theorem}{lemma}}
\newtheorem{corollary}[theorem]{Corollary}
\AddToHook{env/corollary/begin}{\crefalias{theorem}{corollary}}

\AddToHook{env/property/begin}{\crefalias{theorem}{property}}
\newtheorem*{claim*}{Claim}

\newtoggle{submission}

\toggletrue{submission}

\iftoggle{submission}
{
	
	\newcommand{\highlightblue}[1]{{\color{black} #1}}
	\newcommand{\highlight}[2][black]{{\color{black} #2}}
}{
	
	\newcommand{\highlightblue}[1]{{\color{blue} #1}}
	\newcommand{\highlight}[2][red]{{\color{#1} #2}}
}

\iftoggle{submission}
{
	\newcommand{\marginnoteviolet}[1]{\marginpar{}}
	\newcommand{\marginnote}[2][white]{\marginpar{}}
}{
	\newcommand{\marginnoteviolet}[1]{\marginpar{\tiny\color{violet} #1}}
	\newcommand{\marginnote}[2][teal]{\marginpar{\tiny\color{#1} #2}}
}

\newcommand{\id}{\texttt{ID}}
\newcommand{\nid}{\texttt{NewID}}

\newcommand{\vectorcon}{\texttt{VectorConsensus}}
\newcommand{\shareid}{\texttt{SharedRenaming}}
\newcommand{\sharecon}{\texttt{PolyByz}}

\newcommand{\committee}{\texttt{CommitteeElection}}
\newcommand{\clst}{\texttt{CommitteeBuild}}
\newcommand{\redo}{\texttt{RedoConsensus}}

\newcommand{\comcon}{\texttt{CommitteeConsensus}}
\newcommand{\largelead}{\texttt{LargeCoreLeader}}

\newcommand{\lead}{\texttt{LeaderElection}}
\newcommand{\oldid}{\texttt{Renaming}}
\newcommand{\echoone}{\texttt{OldNameGather}}
\newcommand{\echotwo}{\texttt{NewNameScatter}}
\newcommand{\echothree}{\texttt{NewNameValidate}}

\newcommand{\bounce}{\texttt{Bouncing}}
\newcommand{\sample}{\texttt{Sampling}}

\newcommand{\ee}{\mathbb{E}}

\begin{document}


\title{\textbf{Distributed Renaming with Subquadratic Bits\\
via Scalable Committee Election}}
\author[1]{Sirui Bai}
\author[1]{Xinyu Fu}
\author[2,3]{Yuyi Wang}
\author[1,4]{Chaodong Zheng}
\affil[1]{Nanjing University}
\affil[2]{CRRC Zhuzhou Institute}
\affil[3]{Tengen Intelligence Institute}
\affil[4]{University of Chinese Academy of Science, Nanjing}
\date{
}
\pagenumbering{roman}
\maketitle
\thispagestyle{empty}

\begin{abstract}
In distributed computing, the renaming problem requires $n$ nodes with unique identities from a large namespace $[N]$ to acquire new, distinct identities from a smaller target namespace $[M]$. A solution is strong if $M=n$, and is order-preserving if the relative order of identities is maintained. In the synchronous message-passing model, although many fault-tolerant renaming algorithms achieve logarithmic time complexity, they universally incur a high message complexity of $\Omega(n^2)$. Recent work breaks the quadratic barrier, but demands linear runtime and relies on shared randomness.

This paper addresses the challenge of designing renaming algorithms that are simultaneously time-efficient, message-efficient, and Byzantine fault-tolerant, assuming only message authentication. We present two randomized algorithms for strong and order-preserving renaming that tolerate up to $(1/3-\delta)n$ Byzantine failures for any constant $\delta>0$. Our first algorithm, which assumes shared randomness, terminates in \highlight[cyan]{$O(\text{poly-log}(n))$} rounds with $\tilde{O}(n)$ total communication cost. This matches known lower bounds within poly-logarithmic factor. Our second algorithm eliminates the shared randomness assumption and achieves \highlight[cyan]{$O(\text{poly-log}(n))$} runtime with $\tilde{O}(n+\min\{nf,T\})$ total communication cost, where $f$ is the actual number of faulty nodes and $T$ is the amount of messages faulty nodes sent. This gives the first Byzantine renaming algorithm that achieves both poly-logarithmic runtime and subquadratic communication cost \highlight[cyan]{for a wide range of parameter regimes,} without shared randomness. A key technical enabler is a novel and scalable committee election primitive that could be easily integrated into other algorithms to solve various distributed computing problems with low cost and strong fault-tolerance.
\end{abstract}

\iftoggle{submission}
{}
{
\vspace{3em}\begin{Huge}
\begin{center}
{\bf\highlight{
!!! DO NOT EDIT THIS FILE !!! \\
!!! This file is for final merge !!!
}}
\end{center}
\end{Huge}
}

\newpage
\thispagestyle{empty}
\tableofcontents

\newpage
\pagenumbering{arabic}
\pagestyle{plain}
\setcounter{page}{1}


\section{Introduction}\label{sec:intro}

\emph{Renaming} is a fundamental problem in distributed computing that involves symmetry breaking~\cite{attiya90}. It requires $n$ participating nodes to transform their original unique identities in a large namespace $[N]=\{1, \cdots, N\}$ to new unique identities in a smaller target namespace $[M]$, where $n\leq M<N$ \highlight{and $N\in n^{O(1)}$}\footnote{\highlight[cyan]{We assume $N\in n^{O(1)}$ solely because we employ the leader election algorithm from \cite{king06soda} as a subroutine. This constraint could be obviated by substituting an alternative leader election algorithm with similar complexity.}}.
This problem holds significant theoretical importance since the size of the namespace directly influences the efficiency of many distributed algorithms, including vertex coloring algorithms~\cite{barenboim13}, network decomposition algorithms~\cite{ghaffari2021improved}, and contemporary energy-aware minimum spanning tree construction algorithms~\cite{augustine2024awake}.
In real-world distributed systems such as cryptocurrency networks, renaming also provides substantial benefits by reducing communication overheads that arise when nodes from disparate network domains interact using lengthy original identifiers~\cite{bonneau15}.

The fundamental requirement for renaming algorithms is \emph{uniqueness}, which demands that all nodes acquire distinct new identities in the target namespace. Beyond uniqueness constraint, renaming algorithms may offer additional desirable properties. \emph{Strong renaming} ensures optimal namespace utilization by requiring $M = n$. The \emph{order preserving} property preserves the relative order of identities after renaming, which is necessary when identifiers encode metadata such as resource access priorities.

Our investigation concentrates on developing efficient and Byzantine resilient renaming algorithms in the classical synchronous message-passing model.
Specifically, we consider a fully-connected distributed system containing $n$ nodes. Each node $v$ is aware of its identity $\id(v)\in[N]$ and the network size $n$, but not the identity of any other node.
\highlight[cyan]{All nodes start computation simultaneously and communicate in
synchronized rounds with message size bounded by $O(\log N)$ bits.}\footnote{\highlight[cyan]{Messages are sent at the start of a round and delivered by the end of that round. This implies round numbers need not to be transmitted within messages.}}
This allows each message to contain the sender's identity and some additional information.
\highlightblue{Among the $n$ nodes, $0\leq f<(1/3-\delta)n$ nodes suffer Byzantine failure, where $\delta>0$ is an arbitrarily small constant.}
A Byzantine node can deviate from the specified protocol arbitrarily; that is, a Byzantine node can exhibit arbitrary behavior.
Correct nodes do not know the value of $f$ or which nodes are Byzantine.
We assume an adversary determines which nodes are Byzantine before execution starts.
\highlight{We assume messages are authenticated, this prevents Byzantine nodes from masquerading as honest nodes or forging messages that purportedly originate from correct nodes.}\footnote{
Many works on Byzantine fault-tolerance make similar assumption, such as~\cite{dolev83,fitzi09}. In practice, message authentication is often achieved by using digital signature, which in turn can be achieved by a public key infrastructure (PKI). PKI has been widely used in practice nowadays. We also stress that, the availability of PKI does \emph{not} mean every node in a distributed system knows the public keys of all nodes (this would render renaming trivial). Instead, in practical systems, every node only needs to know a constant number of public keys (\emph{none} of which correspond to a node in the system), and verifying another node's signature is done via a mechanism known as certificate chain. See \cite{rfc5280} for more details.
}

Renaming has received considerable attention, algorithms as well as lower bounds have been developed for this problem in both the shared memory model~\cite{gafni06,alistarh11,alistarh13} and the message-passing model~\cite{okun2006renaming,okun08,denysyuk13,alistarh14,bai25podc}.
Early results on renaming have predominantly emphasized temporal efficiency, yielding solutions with poly-logarithmic time complexity, while state-of-the-art algorithms achieve sub-logarithmic round complexity~\cite{alistarh14}.
However, all these algorithms involve all-to-all communication, resulting in $\Omega(n^2)$ message complexity. Many of them also send large messages of $\Omega(n)$ bits during execution, leading to $\Omega(n^3)$ communication complexity.
Although subquadratic message complexity has been achieved recently~\cite{bai25podc}, that algorithm demands shared randomness a priori and has a time complexity scaling linearly with the number of Byzantine nodes.

Therefore, a natural and important question raise: \emph{Are there Byzantine resilient renaming algorithms with subquadratic communication complexity and poly-logarithmic time complexity?}

\subsection{Results and Contribution}

In this paper, we give positive answer to the above question by devising two randomized (strong and order preserving) renaming algorithms: one with shared randomness and one without shared randomness.
Before presenting the main results, we briefly discuss the key idea that enables us to reduce communication cost.

\paragraph{High-level strategy: Efficient and fault-tolerant distributed computing via committee.}
When designing distributed algorithms, the leader-based approach offers simplicity by centralizing computation at a single node, but is inherently vulnerable to failures. In contrast, the flooding-based approach achieves strong fault tolerance, at the cost of large communication overhead. The \emph{committee-based paradigm} balances these trade-offs by selecting a subset of nodes---in which the number of faulty nodes is bounded and correct nodes dominate---as a \emph{committee} to coordinate tasks collectively, ensuring both efficiency and robustness.

The committee-based paradigm is widely used in fault-tolerant distributed computing (see, e.g., \cite{king06focs,king06soda,gilbert10soda,augustine20,lenzen22,ben25podc,abraham25podc,dufoulon25podc}).
However, to achieve robust and efficient renaming via this approach, we must overcome two key challenges: (1) the committee must be elected without assuming nodes know the identities of each other; and (2) the committee election subroutine must have low time and communication cost.
As we elaborate below, a key contribution of this paper is a new scalable and efficient committee election algorithm that outperforms state-of-the-art.
\highlight[cyan]{In particular, that algorithm demonstrates scalability in both time and message complexity, which grow proportionally with the severity of failures (specifically, the number of faulty nodes and the number of messages faulty nodes sent).}
We believe it could be of independent interest.

\paragraph{Contribution I: Renaming with shared randomness.}
The guarantees enforced by our first renaming algorithm, in which nodes have access to shared randomness, are summarized in the following theorem.
It demonstrates that with this strong assumption, renaming can be done very fast (within \highlight[cyan]{$O(\text{poly-log}(n))$} rounds) with very low communication cost \highlight[cyan]{(the total number of messages sent by all nodes is $\Tilde{O}(n)$, each of $O(\log{N})$ bits).}\footnote{\highlight[cyan]{Throughout the paper, we use $\tilde{O}(\cdot)$ to hide poly-logarithmic factor in $n$ and $N$ (recall that we assume $N\in n^{O(1)}$).}}
In fact, due to the recent lower bound shown by Bai, Fu, Wang, Wang, and Zheng~\cite{bai25podc} which states that randomized strong renaming succeeding with constant probability needs $\Omega(n)$ message/communication cost even with shared randomness, this algorithm achieves near-optimal message complexity (within poly-logarithmic factor).

\begin{theorem}\label{thm:with-renaming}
Assume $n$ nodes are in a synchronous distributed system, among which $0\leq f<(1/3-\delta)n$ are faulty, where $\delta>0$ is a constant.
Assume nodes have access to shared randomness.
There exists a strong and order preserving renaming algorithm that is correct with high probability in $n$.
Moreover, the algorithm has time complexity $\Tilde{O}(1)$ and message complexity $\Tilde{O}(n)$, with high probability in $n$.\footnote{An event happens with high probability in $n$ if it occurs with probability at least $1-1/n^\beta$, for some desirable constant $\beta\geq 1$.}
\end{theorem}

Obtaining near-linear message complexity demands us to avoid all-to-all communication patterns, which can be achieved via the committee-based paradigm.
In particular, when shared randomness is available, it is relatively easy to elect a small subset of nodes to form a reliable committee that contains mostly correct nodes. All nodes then communicate exclusively with these committee members and rely on the committee's coordinated decisions to do renaming.

Nevertheless, renaming is still non-trivial, even when a committee is present.
Specifically, to achieve strong and order preserving renaming, each node should use the rank of its original identity as its new identity (i.e., the node with the smallest original identity receives new identity 1, the node with the second smallest original identity receives new identity 2, so on and so forth), but the presence of Byzantine nodes complicates this process. After nodes report their original identities to the committee, committee members may obtain different local identity lists. If committee members perform renaming solely based on their individual lists, inconsistent decisions will arise. On the other hand, achieving consensus on an identity list among committee members to ensure consistency would be too costly.
Our solution to this dilemma is to elect a leader within the committee. This leader, which is not necessarily a correct node, assigns new identities according to its own local identity list; while all committee members verify the validity of this assignment together. If legitimate, the assignment is accepted; otherwise, a new leader is selected for reassignment. This approach offers two key advantages: (1) it eliminates inconsistency in identity assignment by relying on a single leader rather than a committee; and (2) it avoids the substantial cost of reaching committee-wide consensus on the identity list, requiring only agreement on a single bit (the validity of this assignment). Consequently, our algorithm achieves favorable time and message complexity.

\paragraph{Contribution II: Scalable committee election.}
To arrive at a communication efficient renaming algorithm without shared randomness, we still rely on the committee-based paradigm.
However, when nodes do not know the identities of each other and do not have common coin, establishing a reliable committee becomes highly non-trivial, especially if we want to make this process time and communication efficient.

At a high-level, we take the following approach which involves multiple iterations: in each iteration, each node becomes a committee member with some probability $p$. Initially, $p$ is set to be a small value. If Byzantine nodes do not send excessive messages within an iteration, a reliable committee can be elected, and all nodes accept the current committee. Otherwise, all nodes agree to proceed to the next iteration and double the probability $p$.
Eventually, $p$ will be sufficiently large so that the committee size becomes $\Theta(f)$. In that iteration, a reliable committee will be elected regardless of the behavior of the Byzantine nodes.
This ``doubling'' technique is widely used in distributed computing, especially in the context of resource competitive algorithms~\cite{bender15}, where an algorithm's performance scales with the actual severity of failure.

The novelty of our committee election algorithm lies in the design of a better verification mechanism when determining the reliability of a candidate committee.
This mechanism allows us to achieve stronger fault tolerance ($f<(1/3-\delta)n$) than the state-of-the-art ($f<(1/4-\delta)n$ in \cite{augustine20}) in the same model, without assuming any knowledge on the value of $f$ (see \Cref{subsec:related-work} for detailed discussion).
Specifically, in previous best result by Augustine, King, Molla, Pandurangan, and Saia~\cite{augustine20}, each node makes local decision on the validity of a candidate committee based on the size of its own committee list. In contrast, our algorithm first filters out Byzantine committee members that do not appear in many nodes' committee lists (this step reduces the message complexity of subsequent steps). Nodes then forward the filtered committee lists in a specific manner (see \Cref{sec:overview} for more details), allowing each node to obtain a ``union list''. Based on this ``union list'', each node performs a more accurate judgment on the candidate committee's validity.

The following theorem summarizes the guarantees provided by our committee election procedure.
Note that its performance metrics scale with both the number of Byzantine nodes and the amount of messages Byzantine nodes sent.
For instance, if Byzantine nodes in total send $O(n)$ messages, then honest nodes can elect a reliable committee of size $O(\log{n})$ using $\tilde{O}(n)$ messages in total, regardless of the value of $f$.
Similarly, as long as $f=O(n^{1-\alpha})$ where $\alpha>0$ is an arbitrarily small constant, the total message complexity of committee election is $o(n^2)$.

\begin{theorem}\label{thm:committee}
Assume $n$ nodes are in a synchronous distributed system, among which $0\leq f<(1/3-\delta)n$ are faulty, where $\delta>0$ is a constant. Let $G'$ denote the set of correct nodes. There exists a committee election algorithm that, at each node $v\in G'$, outputs: (1) a boolean variable $elected_v$ indicating whether $v$ is in the committee; (2) a set $S_v$ containing the identities of committee members; (3) an upper bound $\mathcal{C}$ on the committee size and a lower bound $\mathcal{C'}$ on the number of correct nodes in the committee.
Moreover, the algorithm guarantees following properties with high probability in $n$.
\begin{enumerate}[topsep=0ex,itemsep=0ex,label=(\alph*)]
	\item Validity: for any $v\in G'$, it holds that $\{\id(u)\mid u\in G'\text{ and }elected_u=1\}\subseteq S_v$.
	\item Bounded committee size: $|\bigcup_{v \in G'}S_v|<\mathcal{C}$, and \highlight{$\mathcal{C}=O(\min \{ f, T/n \}+\log n)$}.
	\item Honest majority: $|\{\id(v)\mid v\in G'\text{ and }elected_v=1\}|>\mathcal{C'}>\left(\frac{2}{3}+\delta'\right)\mathcal{C}$, for some constant $\delta'>0$.
	\item The algorithm has time complexity \highlight[cyan]{$O(\text{poly-log}(n))$} and message complexity \highlight{$\tilde{O}(n+\min\{nf,T\})$}.
	\item All nodes stop execution simultaneously with identical $\mathcal{C}$ and $\mathcal{C'}$ values.
\end{enumerate}
Here, $T\geq 0$ is the total number of messages sent by faulty nodes during algorithm execution.\footnote{\highlight[cyan]{For any faulty node $w$, since the adversary is static and takes over the node prior to the start of the algorithm, every message sent by $w$ is accounted for in $T$, even if $w$ mimics the behavior of a non-faulty node.}}
\end{theorem}

We would also like to highlight that the above procedure can be easily integrated into more sophisticated algorithms to solve various distributed computing problems besides renaming, as it ensures all nodes terminate simultaneously.
For example, the following corollary demonstrates that it can be used to solve Byzantine binary consensus. The specific approach involves first reaching consensus within the committee, then broadcasting the result to all nodes, after which each node adopts the majority value among the received results as its final decision.
Though this approach does not outperform best-known algorithms, we stress that its message complexity scales with the actual number of Byzantine nodes and the amount of messages Byzantine nodes sent, which is a trait few existing algorithms posses. Moreover, for certain parameter regimes, it provides similar performance guarantees with the state-of-the-art consensus algorithms which do not utilize strong cryptographic assumptions.

\begin{corollary}\label{cor:consensus}
Assume $n$ nodes are in a synchronous distributed system, among which $0\leq f<(1/3-\delta)n$ are faulty, where $\delta>0$ is a constant.
There exists an algorithm that, with high probability in $n$, solves binary consensus in $\Tilde{O}(1)$ rounds with \highlight{$\tilde O(n+\min\{nf,T\})$} messages.
Here, $T\geq 0$ is the total number of messages sent by faulty nodes during algorithm execution.
\end{corollary}

\paragraph{Contribution III: Renaming without shared randomness.}
Once a reliable committee containing mostly correct nodes is elected, strong and order preserving renaming can be accomplished via a similar strategy as in the case where shared randomness is available, though additional care must be taken to handle the challenges that different nodes may hold different committee lists and that the committee's size may be large. See \Cref{sec:overview} for more detailed discussion.
The guarantees provided by our entire algorithm (including both the committee election part and the renaming part) are summarized in the following theorem.

\begin{theorem}\label{thm:without-renaming}
Assume $n$ nodes are in a synchronous distributed system, among which $0\leq f<(1/3-\delta)n$ are faulty, where $\delta>0$ is a constant.
There exists a strong and order preserving renaming algorithm that is correct with high probability in $n$.
Moreover, the algorithm has time complexity $\Tilde{O}(1)$ and message complexity \highlight{$\tilde{O}(n+\min\{nf,T\})$}, with high probability in $n$.
Here, $T\geq 0$ is the total number of messages sent by faulty nodes during algorithm execution.
\end{theorem}

As can been, the message complexity of this algorithm scales nicely with the actual number of Byzantine nodes and the actual amount of messages Byzantine nodes sent (instead of worst-case upper bounds), while maintaining poly-logarithmic time complexity.
In particular, the message complexity is $o(n^2)$ as long as $f\in O(n^{1-\alpha})$ for some arbitrarily small constant $\alpha>0$.
Moreover, the message complexity could be as low as $\tilde{O}(n)$ if the Byzantine nodes send limited messages in total ($O(n)$ in particular).
To the best of our knowledge, this is the first Byzantine resilient renaming algorithm that can achieve sub-quadratic communication cost.
Our result also indicates that fault-tolerant renaming is faster in synchronous systems, as Alistarh, Gelashvili, and Vladu~\cite{alistarh15} have showed an $\Omega(n^2)$ message complexity lower bound for solving crash resilient renaming in asynchronous systems.

\subsection{Related Work}\label{subsec:related-work}

\paragraph{Renaming.}
Renaming was first introduced in \cite{attiya90} for asynchronous message-passing systems with crash failures. It has since been extensively studied in both shared-memory and message-passing systems under various failure models. In this part, we focus on Byzantine resilient renaming in synchronous message-passing systems; interested readers can refer to surveys such as \cite{castaneda11,alistarh15survey} for broader overviews.

Solutions for consensus~\cite{lamport82} and reliable broadcast~\cite{bracha85} could be adapted to solve renaming in Byzantine settings by having nodes agree on the set of participating identities. However, this approach typically requires a number of rounds that grows linearly with the maximum number of failures~\cite{dolev82}.
The first explicit treatment of renaming in synchronous message-passing systems with Byzantine failures was presented by Okun and Barak~\cite{okun2006renaming}. Their work described how to achieve strong and order preserving renaming using Byzantine agreement, with linear time complexity. They also proposed a fast renaming algorithm that runs in $O(\log n)$ time, at the cost of a target namespace with size polynomial in $n$. Both algorithms tolerate up to \highlightblue{$n/3-1$} failures.
Subsequent work by Okun, Barak, and Gafni~\cite{okun08} introduced a strong, non-order-preserving renaming algorithm that tolerates \highlightblue{up to $n-1$} Byzantine failures with linear time complexity. They also presented an almost-strong renaming algorithm with $O(\log n)$ time complexity, tolerating up to \highlightblue{$n/3-1$} failures. All algorithms mentioned so far are deterministic.
Randomized approaches have also been explored for Byzantine resilient renaming. For example, Denysyuk and Rodrigues~\cite{denysyuk13} proposed a randomized algorithm that uses cryptographic primitives to solve strong renaming in $O(\log n)$ time, tolerating up to \highlightblue{$n-1$} Byzantine failures.

Despite strong fault tolerance and low time complexity, a common limitation of above results is their quadratic message complexity. Recently, \cite{bai25podc} broke this barrier at the cost of linear in $f$ time complexity and strong assumption (shared randomness).
In contrast,  \highlight[cyan]{our algorithm achieves sub-quadratic message complexity for a wide range of parameter regimes} while maintaining poly-logarithmic time complexity, without using shared randomness.

\highlight[cyan]{\Cref{tbl:results} summarizes the most relevant existing results alongside our contributions. 
}

\begin{table}[t]
\centering
\begin{footnotesize}
\begin{tabular}{cccccccc}
\toprule
& Fault-tolerance & Runtime & \makecell{Total\\ Messages} & \makecell{Total\\ Bits} & Strong & \makecell{Order\\ Preserving} & \highlight[cyan]{Randomness} \\
\midrule
\cite{okun2006renaming} & Byzantine ($f<n/3$) & \highlight[cyan]{$O(n^2)$} & $O(n^4)$ & $\tilde{O}(n^4)$ & yes & yes & \highlight{-} \\
\cite{okun2006renaming} & Byzantine ($f<n/3$) & \highlight[cyan]{$\tilde O(n)$ } & $\tilde  O(n^3)$ & $\tilde{O}(n^3)$ & yes & yes & \highlight{-} \\
\cite{okun2006renaming} & Byzantine ($f<n/3$) & $O(\log n)$ & $\tilde O(n^2)$ & $\tilde O(n^3)$ & - & - & - \\
\cite{okun08} & Byzantine ($f<n$) & $\tilde O(n)$ & $\tilde O(n^2)$ & $\tilde O(n^2)$ & yes & - & - \\
\cite{okun08} & Byzantine ($f<n/3$) & $O(\log n)$ & $\tilde O(n^2)$ & $\tilde O(n^3)$ & - & - & - \\
\cite{denysyuk13} & Byzantine ($f<n$) & $O(\log n)$ & $\tilde O(n^2)$ & $\tilde O(n^2)$ & yes & - & private \\
\cite{denysyuk13b} & Byzantine ($f<n/3$) & $O(\log n)$ & $\tilde O(n^2)$ & $\tilde O(n^3)$ & - & yes & - \\
\cite{bai25podc} & Byzantine ($f<(1/3-\delta)n$) & $\tilde O(f)$ & $\tilde O(n)$ & $\tilde O(n)$ & yes & yes & shared \\
This work & Byzantine ($f<(1/3-\delta)n$) & $O(\text{poly-log}(n))$ & $\tilde O(n)$ & $\tilde O(n)$ & yes &yes & shared \\
This work & Byzantine ($f<(1/3-\delta)n$) & $O(\text{poly-log}(n))$ & $\tilde O(nf)$ & $\tilde O(nf)$ & yes &yes & private \\
\bottomrule
\end{tabular}
\end{footnotesize}
\vspace{-1ex}
\caption{\highlight[cyan]{
Comparison of existing Byzantine fault-tolerant renaming algorithms in the synchronous message-passing model.
In the ``randomness'' column, ``private'' means each node uses private random bits, ``shared'' means all nodes additionally use shared random bits, and ``-'' means the algorithm is deterministic.
}}\label{tbl:results}
\vspace{-2ex}
\end{table}

\paragraph{Committee election.}
The committee-based paradigm is widely adopted in distributed computing to reduce time and/or communication overhead. It has served as a subroutine in many consensus and leader election algorithms. The concept of committee was first introduced by Bracha in 1987~\cite{bracha1987log}, with early works including~\cite{cooper1993fast, ostrovsky1994simple, russell2001perfect}.
When Byzantine nodes are present, shared randomness and cryptographic tools are often used to help elect a reliable committee~\cite{kumar2024sublinear,abraham25podc,bai25podc}.
If nodes do not know the identities of each other in prior, then committee election becomes much more challenging.

The most closely related work is by Augustine, King, Molla, Pandurangan, and Saia~\cite{augustine20}, which achieves committee election with $O(\text{poly-log}(n))$ round complexity and $O((\min\{t,n^2\}+n)\log n)$ bit complexity, where $t$ is the total number of bits sent by Byzantine nodes. Their algorithm tolerates up to $(1/4-\delta)n$ Byzantine failures and is the first to achieve a message complexity that scales with the actual severity of failures.
Our committee election process is inspired by \cite{augustine20}, but incorporates new mechanisms and techniques, resulting in an algorithm with similar performance, stronger fault-tolerance, and superior ease of integration. Specifically, our algorithm's time complexity and message/communication complexity are on par with that of \cite{augustine20}, with following advantages: (1) our algorithm tolerates $(1/3-\delta)n$ failures whereas \cite{augustine20} tolerates $(1/4-\delta)n$; (2) our algorithm does not require any knowledge on $f$ whereas \cite{augustine20} needs to know the exact value of $f$ (not knowing $f$ or replacing $f$ with an upper bound will render the complexity bounds of \cite{augustine20} invalid); (3) our algorithm ensures every correct node receives a reliable committee list whereas \cite{augustine20} only distributes committee lists to correct committee members.

We also stress that, though we assume message authentication, our committee election algorithm works as is in the same model as \cite{augustine20}. In particular, almost all claims in \Cref{thm:committee} still hold in the same model as \cite{augustine20}, with the only adjustment required being the committee size's upper bound $\mathcal{C}$ is $O(\min \{ f\log{n}, T/n \}+\log n)$ instead of $O(\min \{ f, T/n \}+\log n)$.
(\highlight[cyan]{Interested} readers can refer to Appendix~\ref{app-sec:committee-elect-alt-analysis} to see the alternative analysis of the committee election algorithm in the model of \cite{augustine20}. We encourage readers to finish \Cref{sec:analysis-no-shared-rnd} before proceeding to Appendix~\ref{app-sec:committee-elect-alt-analysis}.)

\section{Technical Overview}\label{sec:overview}

This section presents a technical overview of our two renaming algorithms, as well as the committee election scheme which enables them.
We begin by introducing renaming with shared randomness, which serves to showcase some key ideas that we employ during algorithm design.
Subsequent parts delve into the setting without shared randomness, which is much more challenging. In that case, we first tackle fast and scalable committee election. Then, we address the renaming problem, explaining what additional techniques are used along the way.

\paragraph{Renaming with shared randomness.}
The availability of a global coin dramatically simplifies the task as it allows us to: (1) effortlessly establish a reliable, small-sized committee with a guaranteed honest majority, and (2) leverage this committee to execute a time efficient and message competitive renaming procedure.

For committee election, a random global pool is established by independently sampling identities in $[N]$ with probability $\Theta((\log n)/n)$ using shared random bits.
Nodes whose identities reside in this pool self-elect into a candidate committee and broadcast their candidacy. Each node $v$ then verifies received election messages against the pool to construct a local candidate committee list $S_v$.
Candidate committee members execute a vector consensus algorithm to reach agreement on a committee list $\hat{S}$, which is then propagated throughout the network. For each node $v$, the final committee list $S'_v$ is determined through threshold-based validation, where nodes accept identities confirmed by sufficiently many candidate committee members. A critical property enforced by the above scheme is that all nodes have \emph{identical} committee lists.

For new identities generation and distribution, we implement an iterative leader-based identity assignment procedure based on the committee that is generated above.
This procedure contains multiple iterations, in which the node with the $i$-th smallest identity in the committee plays the role of leader in the $i$-th iteration.
Each iteration begins with nodes broadcasting their identities to committee members, allowing each committee member to construct a local list of identities that are present. Then, each committee member sends its local identity list to the leader using random intermediate nodes for forwarding.
The mechanism of randomized intermediate forwarding---particularly effective for transmitting large messages efficiently---is employed extensively throughout our algorithms. To modularize this frequently used mechanism, we abstract it as a dedicated communication primitive called $\bounce$, formalized in \Cref{subsec:bounce}.
The current leader aggregates the forwarded messages to construct its identity list, and assigns new rank-based identities for the nodes that are in the list.
The leader then distributes this rename-list to all committee members, again leveraging the $\bounce$ primitive to reduce time cost.
After receiving the rename-list, each committee member checks if every node in its local identity list receives a new identity, and distributes new identities to corresponding nodes. Upon receiving these messages, a node simply accepts an arbitrary new identity in these messages, and send echo messages to the committee to inform the new identity it has chosen. These echo messages are used for the following validation.
Upon gathering the new identities of all nodes, the committee members check if there exist: duplicate identities (two nodes with identical new identities), or invalid identities (nodes with new identities not in $[n]$), or identities with invalid order ($\id(u)<\id(v)$ yet $\nid(u)>\nid(v)$). Any violation triggers a committee member to vote to continue into next iteration.
Byzantine resilience is achieved through consensus on whether to continue into next iteration, ensuring termination occurs only if all consistency checks pass.
In summary, this propose-verify-commit approach with rotating leadership and distributed auditing enables reliable renaming.

Before proceeding to the next part, we briefly describe the \bounce\ primitive mentioned above.
Suppose node $u$ wants to send a message of size $\Theta(n)$ to node $v$. Since our model limits each message to be of $O(\log N)$ bits, direct sending costs $\Omega(n/\log{N})$ rounds. Instead, $u$ can split the message into $\Theta(n)$ chunks each of size $O(1)$; and for each chunk asks $\Theta(\log{n})$ random nodes to forward that chunk to $v$. In this way, $v$ can receive the entire message within one round reliably, as each chunk is given to at least one correct node with high probability.
We often call this mechanism ``bounce forwarding'' throughout the paper, see \Cref{subsec:bounce} for more details.

\paragraph{Committee election without shared randomness.}
As have been sketched in \Cref{sec:intro}, we implement an adaptive committee election algorithm through exponentially increasing election probabilities. This strategy ensures that the protocol efficiently converges to a committee with sufficient correct nodes while maintaining logarithmic round complexity and competitive message complexity.

The protocol initializes with a base probability $p = \Theta((\log n)/n)$ and doubles the election probability after each iteration until either a valid committee is formed or all nodes are included in the committee.
In each iteration, nodes independently elect themselves with probability $p$ and broadcast election messages; meanwhile, each node $v$ also collect received election messages and form a local candidate committee list $S_v$.
The core verification step then occurs, which processes the raw candidate set $S_v$ into a verified committee set $S_v'$, and determines whether next iteration is required via a boolean variable $redo_v$. This verification step is the key that enables the algorithm to tolerate up to $(1/3-\delta)n$ Byzantine nodes.

To clearly explain our verification process, we must first clarify the requirements for our committee composition. We aim for a committee satisfying:
(1) every correct committee member is included in the committee list of every correct node;
and (2) the size of the union of committee lists held by all correct nodes has an upper bound denoted by $com\_all$, in which correct committee members constitute the majority.
(Note that compared with the case where shared randomness is available, we no longer demand all correct nodes hold identical committee lists.)

Now, we describe our verification process, which guarantees one of two outcomes: either the committee meets all the above requirements and all nodes acknowledge its validity, or the committee fails to meet the requirements and all nodes reject it.
This process begins by letting each node $v$ perform a preliminary validation on its local committee list $S_v$: for each identity $\id(u)$ in $S_v$, node $v$ asks $\Theta(\log n)$ random nodes about whether $\id(u)$ appears in their local committee lists, and removes $\id(u)$ from $S_v$ if the number of received confirmations falls below a predefined threshold.
This step helps identify Byzantine committee members who have not widely announced their presence.
In the subsequent algorithm, this threshold-based ``random sampling'' verification mechanism will be employed repeatedly. We therefore abstract it into a verification primitive, denoted as $\sample$, formalized in \Cref{subsec:sample}.
We would like to highlight this ``random sampling'' technique, which is inspired by \cite{augustine20}, in that it has several favorable properties.
First, by doing random sampling, the set of nodes sampled by a correct node 
\highlight[cyan]{with high probability}
contain mostly correct nodes, regardless of the behavior of the adversary.
Second, this process allows nodes to accurately estimate the distribution of a variable among all nodes with limited cost: the node does not need to query all nodes.
Third, if the adversary wants to bias the outcome of this process, it has to spend a lot of resources. This allows the resource expenditure of our algorithm to scale with the actual severity of failures.

After preliminary validation, each node broadcasts a random subset of its local committee list.
\highlightblue{The subset size is carefully designed so that the broadcast process does not incur too much time or too many messages, while each honest committee member appears in sufficiently many subsets.}
Upon receiving these list fragments from other nodes, every node computes the union of its local list and the received list fragments.
If the union exceeds the size limit $com\_all$, the node notifies all committee members.
If a committee member receives excessive notifications, it deems current committee invalid. Otherwise, it considers current committee to be valid.
At this point, correct committee members reach agreement on committee's validity using a leader-based phase king style consensus algorithm, followed by a global consensus among all nodes to finalize the decision.

This verification mechanism ensures that for a committee member to pass the preliminary validation, it must have presented itself to most nodes. Consequently, it will likely be broadcast by a correct node to all others and included in the union list. This step is critical for tolerating up to almost $1/3$ fraction of Byzantine nodes.
Particularly, in prior algorithms \highlight[cyan]{(e.g., \cite{augustine20})}, a scenario could arise where each local committee list appears valid (that is, containing mostly correct nodes), yet each such locally-valid list contains \emph{distinct} Byzantine nodes. As a result, the union of these lists exceeds the size limit and contains too many Byzantine nodes.
\highlightblue{To mitigate this, prior algorithms can only tolerant $1/4$ fraction of Byzantine nodes.}
In our algorithm, each node obtains a local union list, and any Byzantine node that passes preliminary validation will with high probability appear in the local union list of every correct node. If too many such Byzantine committee members exist, every correct node can detect the anomaly locally and vote to invalidate the committee, thereby achieving stronger fault tolerance.

\highlight[cyan]{
We conclude this part with a concrete example to illustrate the effectiveness of the above procedure, and demonstrate why previous algorithms such as \cite{augustine20} fail when the fraction of faulty nodes approaches $1/3$.
Consider a setting where $3001$ nodes---among which $1000$ are faulty---are building a committee of size $301$.
Suppose $201$ correct nodes (denoted as $X$) elect themselves as candidates, and two distinct groups of faulty nodes each of size $100$ (denoted as $A$ and $B$ respectively) also elect themselves as candidates.
Faulty nodes can control their candidacy announcements so that roughly half of the correct nodes initially observe $X \cup A$ as the candidate set, while the other half observe $X \cup B$ as the candidate set.
In previous algorithms, the random sampling mechanism will validate both $A$ and $B$, as each identity in $A\cup B$ is supported by $2/3$ fraction of nodes.
Consequently, for any correct node, the final committee consists a near-equal split of correct and faulty nodes. Moreover, no correct node is aware of this, as each node only maintains a local view and lacks a global perspective of the entire committee.\footnote{Conversely, when the fraction of faulty nodes is restricted to $1/4$, groups $A$ and $B$ cannot be validated simultaneously, as each node requires support from $3/4$ fraction of all nodes to be admitted into the final committee.}
Our algorithm addresses this vulnerability by the aforementioned fragment-union mechanism.
Even if the correct nodes' initial observations on the candidate committee are partitioned (i.e., $X \cup A$ and $X \cup B$), our algorithm requires each node to broadcast a small and carefully sized fragment of its local candidate list to all other nodes. By aggregating these fragments, every node can reconstruct the union $X \cup A \cup B$. This aggregated union exposes the anomaly, allowing all nodes to locally invalidate the current candidate committee.
}

\paragraph{New identity assignment without shared randomness.}
In case global shared coin is not available, we rely on a leader to assign new identities and delegate verification authority to the committee.
Our algorithm ensures that any improper new identity assignment will be detected by the committee, triggering subsequent iterations with new leaders, whereas valid assignments receive committee approval to finalize the renaming process.
We maintain the invariant that correct leaders always produce valid assignments that pass all validity checks.
By employing a \highlight{scalable leader election algorithm based on \cite{king06soda}} that returns a correct node with constant probability, $O(\log n)$ iterations suffice to accomplish renaming with high probability.

We now give more detail regarding each iteration. First, all nodes transmit their identities to committee members, who compile local identity lists from received messages.
The identities on these local lists first go through a threshold-based random sampling validation process similar to $\sample$, and then are relayed to the leader through the $\bounce$ communication primitive.
This enables the leader to construct a comprehensive identity list that contains every correct committee member's local list.
The leader assigns rank-based new identities to the nodes in its list and disseminates these assignments to the committee via the $\bounce$ primitive.
Each committee member subsequently verifies the validity of assignments by examining following critical aspects: whether all nodes in the local list receive assignments, whether identity collisions occur, whether all identities fall within $[n]$, and whether all identities are order preserving.
The verification culminates in a two-tier consensus process: committee members first reach agreement on the assignments' validity, followed by global consensus among all nodes.

To conclude this section, we remark that proper combination of the leader-based paradigm and the committee-based paradigm is another key integrant that leads us to the final algorithm. Indeed, though a correct committee contains mostly honest nodes, the views of individual committee members are inherently different due to Byzantine nodes. Eliminating such inconsistency could be too costly, in that the cost to find a correct leader might be lower.

\section{Preliminary}\label{sec:preliminary}

Recall that we assume message authentication.
Throughout algorithm description, whenever a node $v$ constructs and sends a message of the form $\langle\cdot\rangle$, it indicates that $v$ signs this message, allowing the receiver to verify that this message indeed is sent by $v$.
In addition, if a node $w$ forwards a message of the form $\langle\cdot\rangle$ received from $v$ to another node $u$, then $u$ can also verify that this forwarded message indeed originates from $v$.
This transitive verifiability is a crucial component in our algorithmic construction for achieving Byzantine-resilient renaming.

The reminder of this section introduces two distributed subroutines that serve as key building blocks for our algorithms: bounce forwarding (\Cref{subsec:bounce}) and random sampling (\Cref{subsec:sample}).
\highlightblue{Note that in the pseudocode, $\epsilon$ and $C$ are constants to be specified later.}

\subsection{Bounce Forwarding}\label{subsec:bounce}

\begin{figure}[!t]
\hrule
\vspace{1ex}
{$\bounce(sender_v, Src_v, Dest_v, Item_v, prob_v)$ executed at node $v$.}
\vspace{.5ex}
\hrule
\begin{footnotesize}
\begin{algorithmic}[1]
\If{($sender_v==1$)}
	\For{($item_v \in Item_v$)}
		Send $\langle ITEM, item_v \rangle$ to $C\log n$ random nodes.
	\EndFor
\EndIf
\For{($i \in [N]$)}
	\State $Cnt_v[i] \gets 0$.
	\State Set $Accept_v[i] \gets 1$ with probability $prob_v$, otherwise set $Accept_v[i] \gets 0$.
\EndFor
\For{(every $\langle ITEM, item_u \rangle$ message received from some node $u \in Src_v$)}
	\If{($Accept_v[\id(u)]==1$ \textbf{and} $Cnt_v[\id(u)] < (1+\epsilon)\log n$)}
		\State Send $\langle ECHO, item_u \rangle$ to all nodes in $Dest_v$.
		\State \highlight{$Cnt_v[\id(u)] \gets Cnt_v[\id(u)]+1$.}
	\EndIf
\EndFor
\State $Item_v' \gets \emptyset$.
\For{(every $\langle ECHO, item_u \rangle$ message received)}
	$Item_v' \gets Item_v' \cup \{ item_u \}$.
\EndFor
\State \textbf{return} $Item_v'$.
\end{algorithmic}
\end{footnotesize}
\hrule
\vspace{1ex}
\caption{Pseudocode of \bounce.}\label{fig:alg-bounce}
\vspace{-3ex}
\end{figure}

The ``bounce forwarding'' mechanism, which is shown in \Cref{fig:alg-bounce}, allows a set of source nodes to each send a large message to a set of destination nodes via a set of intermediate nodes. This enables time-efficient large message transmission in bounded bandwidth model. To illustrate this mechanism, consider an instance in which nodes in $\tilde{S}\subseteq V$ want to send large messages to nodes in $\tilde{D}\subseteq V$, where $V$ denotes the set of all nodes. Every node executes \bounce\xspace with $Src=\tilde{S}$ and $Dest=\tilde{D}$. Moreover, for every node $\tilde{v}\in\tilde{S}$, for the message it intends to transmit, node $\tilde{v}$ will partition it into multiple fragments such that each fragment fits into one message in the considered communication model. When executing \bounce, each node in $\tilde{S}$ sets $sender=1$ and sets $Item$ to its message fragments, while each node not in $\tilde{S}$ sets $sender=0$ and sets $Item=\emptyset$. Now, during the execution, as the pseudocode suggests, each sender node will send each of its message fragments to $\Theta(\log{n})$ random nodes for forwarding. This helps ensure each fragment is received by at least one correct node. On the other hand, for any node $v$, if it receives a request from node $u$ to forward a fragment, node $v$ will send this fragment to nodes in $Dest$ if the following criteria are met: (1) $u$ is in $Src$; (2) the total number of fragments that $v$ has forwarded for $u$ is not too large; and (3) node $u$ is ``accepted'' by $v$, where each node in $Src$ is ``accepted'' by $v$ with a probability specified when invoking \bounce. Note that the last two criteria are introduced to allow us further tweak the total communication cost of \bounce, and to counter the behavior that Byzantine nodes in $Src$ spam fragments.

\subsection{Random Sampling}\label{subsec:sample}

\begin{figure}[t!]
\hrule
\vspace{1ex}
{$\sample(query_v, reply_v, Item_v, thresh_v)$ executed at node $v$.}
\vspace{.5ex}
\hrule
\begin{footnotesize}
\begin{algorithmic}[1]
\If{($query_v==1$)}
	\State $Map_v \gets$ an empty map.
	\For{(every $item_v \in Item_v$)}
		\State $Map_v[item_v] \gets C\log n$ random nodes.
		\State {Send $\langle ASK, item_v\rangle$ {to nodes in} $Map_v[item_v]$.}
	\EndFor
\EndIf
{
\If{($reply_v==1$)}
\For{($i \in [N]$)}
	$Cnt_v[i] \gets 0$.
\EndFor
\For{(every $\langle ASK, item_u\rangle$ message received from node $u$)}
	\If{($item_u \in Item_v$ \textbf{and} $Cnt_v[\id(u)] < (1+\epsilon)C\log n$)}
		\State Reply $\langle YES, item_u\rangle$.
		\State $Cnt_v[\id(u)] \gets Cnt_v[\id(u)]+1$.
	\EndIf
\EndFor
\EndIf
}
\State $Item_v' \gets \emptyset$.
\If{($query_v==1$)}
\For{(every $item_v \in Item_v$)}
	\If{(at least $thresh_v \cdot C\log n$ messages of the form $\langle YES, item_v\rangle$ received from nodes in $Map_v[item_v]$)}
		\State $Item_v' \gets Item_v' \cup \{ item_v \}$.
	\EndIf
\EndFor
\EndIf
\State \textbf{return} $Item_v'$.
\end{algorithmic}
\end{footnotesize}
\hrule
\vspace{1ex}
\caption{Pseudocode of \sample.}\label{fig:alg-sample}
\vspace{-3ex}
\end{figure}

The ``random sampling" mechanism, shown in \Cref{fig:alg-sample}, enables a set of querying nodes to efficiently verify the existence of their data items with limited communication costs.
To illustrate this mechanism, consider an instance where a set of nodes $\tilde{Q} \subseteq V$ wish to verify whether their locally held data items are also present within another set of nodes $\tilde{R} \subseteq V$, where $V$ denotes the set of all nodes.
For each querying node in $\tilde{Q}$, it sets $query=1$ and sets $Item$ to be the set of data items it wants to verify. Nodes not in $\tilde{Q}$ set $query=0$. For each responding node in $\tilde{R}$, it sets $reply=1$ and sets $Item$ to be its own local set of data items. Nodes not in $\tilde{R}$ set $reply=0$.
During execution, as the pseudocode indicates, each querying node sends each of its data items to $\Theta(\log n)$ randomly chosen nodes in $V$.
\highlightblue{This random, limited-scope dissemination ensures low communication cost while preventing Byzantine nodes from manipulating the results unless they spend a lot of resources to let $\tilde{R}$ be aware of the queried items in the first place.}
On the responder side, if a node $v\in\tilde{R}$ receives a query for a specific data item from some node $u$, it replies with a positive acknowledgment only if the following criteria are met: (1) the queried item exists in $v$'s local $Item$ set; and (2) the total number of queries $v$ has already processed from $u$ is below a predefined threshold. The second criterion is introduced to prevent Byzantine querying nodes from spamming responders with excessive requests.
Finally, each querying node collects acknowledgments for its queried items. A data item is included in the output set $Item'_v$ only if the number of positive replies received from the specific nodes originally queried for that item reaches a threshold.
This decision step provides robustness against node failures or malicious non-responders, effectively filtering out items that are not widely held among the replying nodes.

\section{Renaming Algorithm with Shared Random Bits}\label{sec:alg-with-shared-rnd}

\begin{figure}[t!]
\hrule
\vspace{1ex}
{Renaming algorithm with shared random bits executed at node $v$.}
\vspace{.5ex}
\hrule
\begin{footnotesize}
\begin{algorithmic}[1]
\State $elected_v \gets 0$, $Pool \gets \emptyset$, $S_v \gets\emptyset$, $S_v' \gets\emptyset$.
\State $com\_all \gets (1+\epsilon)C\log n$, $com\_g \gets (1-\epsilon)(2/3 + \delta)C\log n$, $com\_b \gets com\_all - com\_g$.
\For{($i \gets 1 \textbf{ to } N$)}
	\State Add $i$ to $Pool$ with probability $(C\log n)/n$ using shared randomness.
\EndFor
\If{($ \id(v) \in Pool $)}
	\State $elected_v \gets 1$.
	\State Send $\langle ELECT, \id(v) \rangle$ to all nodes.
\EndIf
\For{(every $\langle ELECT, \id(u) \rangle$ message received)}
	\If{($\id(u) \in Pool$)}
		$S_v \gets S_v \cup \{ \id(u) \}$.
	\EndIf
\EndFor
\If{($ elected_v==1 $)}
	\State $\hat{S}_v\gets $ Run $\texttt{VectorConsensus}(S_v, com\_all, com\_g, com\_b)$ with nodes in $S_v$.
\EndIf
\If{($elected_v==1$)}
	\For{($\id(u)\gets 1$ \textbf{to} $N$)}
		\If{($\id(u) \in \hat{S}_v$)} Send $\langle LIST, \id(u) \rangle$ to all nodes. \EndIf
	\EndFor
\EndIf
\For{(every $\langle LIST, \id(u) \rangle$ message received)}
	\If{(at least $com\_b$ messages of the form $\langle LIST, \id(u) \rangle$ received from nodes in $S_v$)}
		\State $S_v' \gets S_v' \cup \{ \id(u)\}$.
	\EndIf
\EndFor
\State $\nid(v)\gets $ Run $\shareid(elected_v, S_v', com\_all, com\_g,  com\_b)$ with all nodes.
\State \textbf{return} $\nid(v)$.
\end{algorithmic}
\end{footnotesize}
\hrule
\vspace{1ex}
\caption{Pseudocode of the renaming algorithm with shared random bits.}\label{fig:alg-srb-main}
\vspace{-3ex}
\end{figure}

Recall that we assume there are less than $(1/3-\delta)n$ Byzantine nodes, where $\delta>0$ is an arbitrarily small constant.
\Cref{fig:alg-srb-main} gives the high-level pseudocode of our renaming algorithm when shared randomness is available.
Specifically, our algorithm first elects a committee with the following guarantees:
(1) the size of the committee is at most $com\_all= (1+\epsilon)C\log n$; (2) the committee contains more than $com\_g = {(1-\epsilon)(2/3 + \delta)C\log n}$ correct nodes; and (3) the committee contains less than $com\_b = com\_all - com\_g$ Byzantine nodes.
\highlightblue{Here, $\epsilon>0$ is a sufficiently small constant, and $C>0$ is a sufficiently large constant. ($\epsilon$ and $C$ will be specified in the analysis.)}
Then, we utilize this committee to distributed new identities.

\paragraph{Committee election.}
With shared randomness, committee election can be done efficiently with relative ease.
The process begins with all nodes constructing an identical candidate pool: add each $i\in [N]$ to the pool with probability $(C\log n)/n$ using shared randomness.
For a node $v$, if its identity $\id(v)$ is in the pool, the node is considered elected and broadcasts an $ELECT$ message to announce its candidacy.
Upon receiving such a message from a node $u$, any node $v$ independently verifies $u$'s claim by checking if $u$'s identity is in the pool. If the verification passes, $u$ is added to $v$'s local committee list $S_v$. This yields a local view of the committee at each node. To transform these (potentially different) local views into a single, consistent global committee list, all elected nodes subsequently participate in a \vectorcon{} procedure.
Specifically, each elected node's local committee list can be interpreted as a bit vector of length $N$, in which the $i$-th entry is $1$ iff the elected node believes the node with identity $i$ is in its local committee list. These vectors are sparse as $O(n)$ entries could be $1$. Our \vectorcon{} procedure ensures all correct committee members agree on an identical vector, see Appendix~\ref{app-sec:vector-con} for more details.
Once consensus is reached, the agreed-upon committee list $\hat{S}$ is propagated to all nodes by committee members.
For any node $v$, its final committee list $S'_v$ is determined by a threshold mechanism: a node $u$ in $\hat{S}$ is recognized as a committee member by $v$ if node $v$ receives endorsements from at least $com\_b$ nodes that are in $S_v$ (recall $com\_b$ bounds the number of Byzantine nodes in the committee), ensuring robustness against adversarial behavior.
We conclude this part by noting that the committee election process with shared randomness guarantees strong consistency: all correct nodes output identical committee lists.

\begin{figure}[t!]
\hrule
\vspace{1ex}
$\shareid(elected_v, S_v', com\_all, com\_g,  com\_b)$ executed at node $v$.
\vspace{.5ex}
\hrule
\begin{footnotesize}
\begin{algorithmic}[1]
\State $\nid(v) \gets \perp$.
\For{($j\gets 1$ \textbf{to} $|S'_v|$ )}
	\State $lead_v \gets$ the $j$-th smallest identity in $S_v'$.
	\State Send $\langle ID, \id(v) \rangle$ to all nodes in $S'_v$.
	\If{($elected_v==1$)}
		\State $List_v \gets \emptyset$, $ListMsg_v \gets \emptyset$.
		\For{(every $\langle ID, \id(u) \rangle$ message received)}
			\State $List_v \gets List_v \cup \{ \id(u) \}$.
			\State $ListMsg_v \gets List_v \cup \{ \langle ID, \id(u) \rangle \}$.
		\EndFor
	\EndIf
	\State $ListMsg_v' \gets $ Run $\bounce(elected_v, S_v', \{lead_v\}, ListMsg_v, 1)$.\label{line:shared-rename-bounce-1}
	\If{($\id(v)==lead_v$)}
		\State $List_v' \gets$ all identities in $ListMsg_v'$.
		\State $RankMsg_v \gets \emptyset$.
		\For{(every $\id(u) \in List_v'$)}
			\State $rank_u \gets$ the rank order of $\id(u)$ in $List_v'$.
			\State $RankMsg_v \gets RankMsg_v \cup \{ \langle NewID, \id(v), \id(u), rank_u \rangle \}$.
		\EndFor
	\EndIf
	\State $RankMsg_v' \gets $ Run $\bounce(\mathbb{1}_{\id(v)==lead_v}, \{lead_v\}, S_v', RankMsg_v, 1)$.\label{line:shared-rename-bounce-2}
	\If{($elected_v==1$)}
		\State $com\_redo_v \gets 0$, $Map_v \gets $ an empty map.
		\For{(every $\langle NewID, lead_v, \id(u), rank_u \rangle \in RankMsg_v'$)}
			\State \highlight{$Map_v[\id(u)] \gets \langle NewID, lead_v, \id(u), rank_u \rangle$.}
		\EndFor
		\For{(every $\id(u) \in List_v$)}
			\If{($Map_v[\id(u)] == \perp$)}
				$com\_redo_v \gets 1$
			\algorithmicelse
				\xspace Send $\langle ECHO1, Map_v[\id(u)] \rangle$ to node $u$.
			\EndIf
		\EndFor
	\EndIf
	\State $Msg_v \gets \bot$.
	\For{(every $\langle ECHO1, \langle NewID, lead_v, \id(v), rank_v\rangle \rangle$ message received from nodes in $S_v'$)}
		\State $Msg_v \gets \langle NewID, lead_v, \id(v), rank_v\rangle$, $\nid(v)\gets rank_v$.
	\EndFor
	\State Send $\langle ECHO2, \id(v),  Msg_v \rangle$ to nodes in $S_v'$.
	\If{($elected_v==1$)}
		\State $Name_v \gets$ an empty map.
		\For{(every $\langle ECHO2, \id(u), \langle NewID, lead_v, \id(u), rank_u\rangle \rangle$ received from node $u$ with $\id(u)\in List_v$)}
			\State {$Name_v[\id(u)] \gets rank_u$.}
		\EndFor
		\If{($\exists u_1 \neq u_2, Name_v[u_1]==Name_v[u_2]$)}
			$com\_redo_v \gets 1$.
		\EndIf
		\If{($\exists u, Name_v[u]>n$)}
			$com\_redo_v \gets 1$.
		\EndIf
		\If{($\exists u_1 < u_2, Name_v[u_1] > Name_v[u_2]$)}
			$com\_redo_v \gets 1$.
		\EndIf
		\State  $com\_redo_v' \gets $ Run $\sharecon(com\_redo_v)$ with nodes in $S_v'$.
		\State Send $\langle RET,  com\_redo_v' \rangle$ to all nodes.
	\EndIf
	\State $redo_v \gets$ the majority value of $com\_redo'$ among $RET$ messages received from nodes in $S'_v$.
	\If{($redo_v==0$)}
		\textbf{break}.
	\EndIf
\EndFor
\State \textbf{return} $\nid(v)$.
\end{algorithmic}
\end{footnotesize}
\hrule
\vspace{1ex}
\caption{Pseudocode of \shareid.}\label{fig:alg-srb-shareid}
\vspace{-3ex}
\end{figure}

\paragraph{Distribute new identities.}
Once a committee is elected, we utilize it to accomplish the renaming task; see \Cref{fig:alg-srb-shareid}.
The procedure \shareid{} contains multiple iterations. In each iteration, a committee member becomes the leader and coordinates with the rest of the committee to assign new identities to all nodes. Following this assignment process, the committee further conducts a validity assessment and reach agreement on whether to halt execution (in case the assignment is valid) or continue into next iteration (in case the assignment is invalid).
Next, we describe each iteration in detail.

The leader of iteration $j$ is the node with the $j$-th smallest identity in the committee.
(Recalled that above committee election process ensures all nodes obtain identical committee lists, hence nodes agree on identical leader in each iteration.)
At the start of each iteration, each node broadcasts its original identity to the committee.
Each committee member $v$ adds these identities into its local identity list $List_v$ and stores these identity announcement messages in a local list $ListMsg_v$. Subsequently, every committee member $v$ transmits $ListMsg_v$ to the leader via the \bounce{} primitive.
The leader then aggregates all identities in the received messages into a unified list $List'$ and assigns new identities to the nodes by computing the rank order of each node's original identity in $List'$. Then, the leader transmits the new identity assignment to the committee via the \bounce{} primitive.

In the rest of the iteration, each committee member first validates and distributes the new identity assignment, and then reach consensus on whether to continue into next iteration.
More specifically, each committee member $v$ verifies whether every identity in its local $List_v$ has been assigned a new identity in the leader's proposal. If any identity is missing an assignment, $v$ sets $com\_redo_v$ to $1$, voting for proceeding to the next iteration.
Otherwise, for each node $u \in List_v$, committee member $v$ forwards the leader's proposal of $u$'s new identity via an $ECHO1$ message.
\highlightblue{Once node $u$ received $ECHO1$ messages, it picks an arbitrary new identity appearing in $ECHO1$ messages (all new identities for $u$ would be identical in a correct execution).
Node $u$ will then notify the committee about its final new identity via $ECHO2$ messages.
Critically, $u$ incorporates the leader's proposal appearing in $ECHO1$ messages into $ECHO2$ messages; this prevents $u$ from incorrectly accusing the leader or the committee for disseminating invalid new identity.}
Based on received $ECHO2$ messages, each committee member $v$ constructs a mapping from nodes' original identities to their new identities.
Then, three critical checks are performed: identity uniqueness (no new identity is claimed by multiple nodes), range validity (no new identity exceeds $n$) and order preserving (new identities preserve the relative order of original identities). Any violation triggers committee member $v$ to vote for next iteration.
To reach agreement on whether the next iteration is needed, the committee executes \sharecon{}, \highlightblue{which is a standard binary consensus algorithm~\cite{lynch1996distributed}}. Then, the committee nodes broadcast the consensus result to all nodes. A node only terminates at the end of this iteration and commits to its new identity when the majority of the consensus results it received confirms a collective decision of not executing the next iteration. Otherwise, the node proceeds to the next iteration.

\section{Analysis for the Algorithm with Shared Randomness}\label{sec:analysis-shared-rnd}

In this section, we formally prove the correctness and the complexity of the renaming algorithm with shared random bits.
Throughout the analysis, we assume there are less than $(1/3-\delta)n$ Byzantine nodes where $\delta>0$ is an arbitrarily small constant, $c>4$ is a constant of our choice, $\epsilon$ is a sufficiently small constant satisfying $\epsilon < \min \{ 1/100, \delta/3,({1-3\delta})/9\}$, $C$ is a sufficiently large constant satisfying $ C > (100c)/\epsilon ^ 4$.
For the ease of presentation, we often let $G$ denote the set of correct nodes with $elected=1$ (i.e., correct committee members), and let $G'$ denote the set of all correct nodes.

We begin the analysis by proving that all nodes obtain identical committee lists (i.e., all nodes have identical $S'$), which contains all correct committee members. Moreover, more than $2/3$ fraction of the list correspond to correct nodes.

\begin{lemma}\label{lem:with-s}
For the algorithm in \Cref{fig:alg-srb-main}, the following properties hold with probability at least $1-n^{-8c}$:
\begin{itemize}[topsep=1ex,itemsep=0ex]
	\item for any $v, u \in G'$, it holds that $S_v' = S_u'$;
	\item $G \subseteq \bigcap_{v \in G'} S_v'$;
	\item $|\bigcup_{v \in G'}S_v'|<com\_all$,  $|G| > com\_g > \frac{2}{3}com\_all$, $com\_b = com\_all-com\_g$.
\end{itemize}
\end{lemma}

\begin{proof}
Notice that shared randomness are generated after the adversary chooses Byzantine nodes, hence the probability that a Byzantine node is added into the candidate pool equals to the the probability that a correct node is added into the candidate pool. We prove the lemma via three claims.
In \Cref{clm:with-s-1}, we show that the lists $S$ of all nodes have desirable properties, but they are not necessarily identical.
Since the union of all nodes' $S$ has limited size, in \Cref{clm:with-s-2}, we show that correct committee nodes output identical $\hat{S}$.
Lastly, in \Cref{clm:with-s-3}, we show that after the committee distributes $\hat{S}$, all nodes obtain identical $S'$.

\begin{claiminline} \label{clm:with-s-1}
With probability at least $1-n^{-10c}$, it holds that $G \subseteq \bigcap_{v \in G'} S_v$, $|\bigcup_{v \in G'}S_v|<com\_all$,  $|G| > com\_g > \frac{2}{3}com\_all$, and $com\_b = com\_all-com\_g$.
\end{claiminline}

\begin{proof}
$G \subseteq \bigcap_{v \in G'} S_v$ is true since $G \subseteq Pool$ and nodes in $G$ send $ELECT$ messages.
Since each identity in $[N]$ is added into $Pool$ with probability $(C\log n)/ n$ and there are $n$ identities in total, by a Chernoff bound~\cite{motwani95}, with probability at least $1-n^{-11c}$, it holds that $|\bigcup_{v \in G'}S_v| < com\_all$.
Since there are at least $(2/3+\delta)n$ correct nodes, again by a Chernoff bound, with probability at least $1-n^{-11c}$, we have $|G| > com\_g$.
Moreover, $com\_g > \frac{2}{3}com\_all$ is true by the definition of $\epsilon$, and $com\_b = com\_all-com\_g$ is true by the definition of $com\_b$.
\end{proof}

\begin{claiminline} \label{clm:with-s-2}
With probability at least $1-n^{-10c}$, the output of \vectorcon{} satisfies:
\begin{itemize}[topsep=1ex,itemsep=0ex]
	\item agreement: for any $v, u \in G$, it holds that $\hat S_v = \hat S_u$;
	\item validity: for any $v \in G$, if $\id(u) \in \hat S_v$, then \highlight{$\id(u) \in S_w$ for some $w \in G$}.
\end{itemize}
\end{claiminline}
\begin{proof}
With probability at least $1-n^{-10c}$, \Cref{clm:with-s-1} holds.
Assume this is indeed the case. Then, by \Cref{lem:consensus-vector} in Appendix~\ref{app-sec:vector-con}, agreement and validity properties hold.
\end{proof}

\begin{claiminline} \label{clm:with-s-3}
With probability at least $1-n^{-10c}$, if there exists an $S^*$ such that $\hat S_v = S^*$ for every $v\in G$, then $S_u' = S^*$ for every $u \in G'$.
\end{claiminline}

\begin{proof}
With probability at least $1-n^{-10c}$, properties in \Cref{clm:with-s-1} hold.
Assume there exists an $S^*$ such that $\hat S_v = S^*$ for every $v\in G$.
Fix an $\id(w)$ and a $u \in G'$.
If $\id(w) \in S^*$, there are more than $com\_g > com\_b$ nodes in $S_u$ sending $\langle LIST, \id(w) \rangle$, which means $\id(w) \in S_u'$.
Otherwise, if $\id(w) \notin S^*$, there are less than $com\_all - com\_g = com\_b$ nodes in $S_u$ sending $\langle LIST, \id(w) \rangle$, which means $\id(w) \notin S_u'$.
Consider all choices of $\id(w)$ and $u$, the claim is proved.
\end{proof}

Combine above three claims, the lemma is immediate.
\end{proof}

\Cref{lem:with-naming} assumes the existence of a global leader, which can be easily achieved by letting the node with the $j$-th smallest identity in $S'$ be the leader of the $j$-the phase, as all correct nodes have identical $S'$.
If the leader is correct, then all nodes quit \shareid{} in this phase;
otherwise, if the global leader is Byzantine and creates invalid renaming, then all nodes continue into the next phase.

\begin{lemma}\label{lem:with-naming}
For the \shareid{} algorithm in \Cref{fig:alg-srb-shareid}, during any phase, the following properties hold with probability at least $1-n^{-7c}$:
\begin{itemize}[topsep=1ex,itemsep=0ex]
	\item there exists $lead^*$ such that for any $v \in G'$, $lead_v = lead^*$;
	\item for any $v, u \in G'$, $redo_v = redo_u$;
	\item if $lead^* \in G$, then for any $v \in G'$, it holds that $redo_v=0$;
	\item if there exists $v \in G'$ with $redo_v=0$, then $\nid(\cdot)$ forms a strong and order preserving renaming.
\end{itemize}
\end{lemma}

\begin{proof}
We first sketch the proof.
Based on identical committee lists, a global leader can be easily obtained for each phase.
Our proof relies on four key claims.
\Cref{clm:with-naming-1} shows that binary consensus can be reached within the committee, and the result can be distributed to all nodes.
\Cref{clm:with-naming-2} shows that all correct committee members will send their local identity lists to the leader, allowing the leader to create a more comprehensive list containing all correct committee members' local identity lists, or the leader must be a Byzantine node.
\Cref{clm:with-naming-3} proves that if the leader is correct, the committee will decide to quit \shareid{} and all nodes will make the same decision.
\Cref{clm:with-naming-4} shows that if the leader is Byzantine and creates an invalid renaming, the committee will decide to continue into the next phase and all nodes will make the same decision.

We now start the lemma proof. Note that with probability at least $1-n^{-8c}$, \Cref{lem:with-s} holds.
Hence, throughout this proof, we assume that: (1) for any $v, u \in G'$, it holds that $S_v' = S_u'$; and (2) for any $v, u \in G'$, it holds that $lead_v = lead_u$.

\begin{claiminline}\label{clm:with-naming-1}
The output of \sharecon{} satisfies:
\begin{itemize}[topsep=1ex,itemsep=0ex]
	\item agreement: for any $v, u \in G$, it holds that $com\_redo'_v = com\_redo'_u$;
	\item validity: for any $v \in G$, $com\_redo'_v = com\_redo_u$ for some $u \in G$.
\end{itemize}
\end{claiminline}

\begin{proof}
Standard binary consensus algorithm will yield above results, such as \sharecon{} in \cite{lynch1996distributed}.
\end{proof}

By \Cref{clm:with-naming-1}, there exists $com\_redo^*$ such that $com\_redo'_v=com\_redo^*$ for any $v \in G'$. Then by \Cref{lem:with-s},  since $|G| > com\_g > \frac{2}{3}com\_all$, we have $redo_v=com\_redo^*$ for any $v \in G'$.
That is, all nodes always agree on whether to continue into the next phase based on the committee's consensus result.

\begin{claiminline}\label{clm:with-naming-2}
The following properties hold with probability at least $1-n^{-10c}$:
\begin{itemize}[topsep=1ex,itemsep=0ex]
	\item $G' \subseteq \bigcap_{v \in G} List_v$;
	\item if $lead^* \in G$, then $\bigcup_{v \in G} List_v \subseteq List_{lead^*}'$.
\end{itemize}
\end{claiminline}

\begin{proof}
Fix any $u \in G'$.
By \Cref{lem:with-s}, $G \subseteq S_u'$, hence $u$ sends $\langle ID, \id(u) \rangle$ messages to all nodes in $G$.
As a result, for any node $v \in G$, we have $\id(u) \in List_v$.
This implies $G' \subseteq \bigcap_{v \in G} List_v$.

Fix any node $x \in G$.
By the above proof, $G' \subseteq List_x$.
Node $x$ sends each identity in $List_x$ to $C\log n$ random nodes during the execution of \bounce{} in Line~\ref{line:shared-rename-bounce-1} of \Cref{fig:alg-srb-shareid}. By a Chernoff bound, with probability at least $1-n^{-11c}$, node $x$ sends less than $(1+\epsilon)C\log n$ messages to any node during the execution of \bounce{} in Line~\ref{line:shared-rename-bounce-1} of \Cref{fig:alg-srb-shareid}.
Thus the messages sent by $x$ to correct nodes during the execution of \bounce{} in Line~\ref{line:shared-rename-bounce-1} of \Cref{fig:alg-srb-shareid} are forwarded to $lead^*$.
Fix any $\id(y) \in List_x$.
Since there are more than $2n/3$ correct nodes,
by a Chernoff bound, with probability at least $1-n^{-11c}$, there exists a correct node that forwards $\id(y)$ to $lead^*$ during the execution of \bounce{} in Line~\ref{line:shared-rename-bounce-1} of \Cref{fig:alg-srb-shareid}.
Take a union bound over at most $n^2$ choices of $x$ and $\id(y)$, we conclude that with probability at least $1-n^{-10c}$,
every identity in $\bigcup_{v \in G} List_v$ is forwarded to $lead^*$ during the execution of \bounce{} in Line~\ref{line:shared-rename-bounce-1} of \Cref{fig:alg-srb-shareid}.
Therefore, if $lead^* \in G$, then $\bigcup_{v \in G} List_v \subseteq List_{lead}'$.
\end{proof}

\begin{claiminline}\label{clm:with-naming-3}
With probability at least $1-n^{-9c}$, if $lead^* \in G$, then for any $v \in G'$, it holds that $redo_v=0$.
\end{claiminline}

\begin{proof}
By a Chernoff bound, with probability at least $1-n^{-10c}$, leader node $lead^*$ sends less than $(1+\epsilon)C\log n$ $NewID$ messages to any node, so all correct nodes forward $NewID$ messages during the execution of \bounce{} in Line~\ref{line:shared-rename-bounce-2} of \Cref{fig:alg-srb-shareid}.
Fix any $\id(u) \in List_{lead^*}'$.
Leader node $lead^*$ sends $\id(u)$ in $NewID$ messages to $C \log n$ random nodes.
By \Cref{lem:with-s}, all correct nodes have $G$ in their $S'$.
So all nodes in $G$ receive $\id(u)$ during the execution of \bounce{} in Line~\ref{line:shared-rename-bounce-2} of \Cref{fig:alg-srb-shareid}.
Consider all choices of $\id(u)$, we conclude that all nodes in $G$ receive $List_{lead^*}'$ during the execution of \bounce{} in Line~\ref{line:shared-rename-bounce-2} of \Cref{fig:alg-srb-shareid}.
By \Cref{clm:with-naming-2}, with probability at least $1-n^{-10c}$,
$G' \subseteq \bigcap_{v \in G} List_v$ and
$\bigcup_{v \in G} List_v \subseteq List_{lead^*}'$.
Therefore, after $ECHO1$ messages are delivered, all correct nodes have $Msg \neq \bot$ and all nodes in $G$ have $com\_redo=0$.
Since $lead^* \in G$, the renaming is strong and order preserving, implying all nodes in $G$ have $com\_redo=0$ after $ECHO2$ messages are delivered.
By \Cref{clm:with-naming-1}, all nodes in $G$ have $com\_redo'=0$.
Since $|G| > com\_g > \frac{2}{3}com\_all$, we conclude that $redo_v=0$ for any $v \in G'$.
\end{proof}

By \Cref{clm:with-naming-3}, we know with probability at least $1-n^{-9c}$, if the current leader is correct, then all nodes will stop execution after current phase.

\begin{claiminline} \label{clm:with-naming-4}
The following properties hold with probability at least $1-n^{-9c}$:
\begin{itemize}[topsep=1ex,itemsep=0ex]
	\item if there exists $u \in G'$ with $Msg_u = \bot$, then $redo_v=1$ for every $v \in G'$;
	\item if $\nid(\cdot)$ does not form a strong and order preserving renaming among all correct nodes, then $redo_v=1$ for every $v \in G'$.
\end{itemize}
\end{claiminline}

\begin{proof}
By \Cref{lem:with-s} and \Cref{clm:with-naming-2}, if there exists a correct node $u$ with $Msg_u = \bot$, then $\id(u)$ exists in $List$ of all nodes in $G$, and all nodes in $G$ have $Map[\id(u)]=\bot$ and $com\_redo=1$.
By \Cref{clm:with-naming-1}, 
all nodes in $G$ have  $com\_redo'=1$.
Apply \Cref{lem:with-s} again, since $|G| > com\_g > \frac{2}{3}com\_all$, for every $v \in G'$, we have $redo_v=1$.

Now consider correct nodes with $Msg \neq \bot$.
By \Cref{lem:with-s}, $G \subseteq \bigcap_{v \in G'} S_v'$, so all correct nodes send new identities to all nodes in $G$.
Note that if $\nid(\cdot)$ does not form a strong and order preserving renaming among all correct nodes, there are three possible reasons.
(1) $\nid(\cdot)$ violates identity uniqueness, so there exists $u_1 \neq u_2$ with $\nid(u_1)=\nid(u_2)$, in which case all nodes in $G$ have $Name[u_1]=Name[u_2]$.
(2) $\nid(\cdot)$ violates identity validity, so there exists $u_3$ with $\nid(u_3)\notin[n]$, in which case all nodes in $G$ have $Name[u_3]>n$.
(3) $\nid(\cdot)$ violates order preserving,
so there exists $\id(u_4) < \id(u_5)$ with $\nid(u_4) > \nid(u_5)$, in which case all nodes in $G$ have $Name[u_4] > Name[u_5]$.
In any of these three cases, all nodes in $G$ have $com\_redo=1$.
By \Cref{clm:with-naming-1}, all nodes in $G$ have $com\_redo'=1$.
Then, by \Cref{lem:with-s}, since $|G| > com\_g > \frac{2}{3}com\_all$, all correct nodes have $redo=1$.
\end{proof}

By \Cref{clm:with-naming-4}, we know with probability at least $1-n^{-9c}$, if there exists $v \in G'$ with $redo_v=0$, then $\nid(\cdot)$ forms a strong and order preserving renaming among all correct nodes.

Combine above claims immediately lead to the lemma.
\end{proof}

The following lemma shows the complexity of our first renaming algorithm.
We note that the round complexity bottleneck comes from executing \vectorcon{}; while the message complexity bottleneck comes from executing \bounce{} multiple times.

\begin{lemma}\label{lem:with-complex}
For the algorithm in \Cref{fig:alg-srb-main}, with probability at least $1-n^{-6c}$, it has time complexity $O(\log^6 n)$ and message complexity $O(n\log^3 n)$.
\end{lemma}

\begin{proof}
By \Cref{clm:with-s-1} of \Cref{lem:with-s}, with probability at least $1-n^{-8c}$, $|G| \le |\bigcup_{v \in G'}S_v|<com\_all$, where $com\_all=O(\log n)$.
Hence, sending $ELECT$ messages incurs time complexity $O(1)$ and message complexity $O(n\log n)$.
\vectorcon{} incurs time complexity $O(\log^6 n)$ and message complexity $O(\log^8 n)$.
Sending $LIST$ messages incurs time complexity $O(1)$ and message complexity $O(n\log n)$.

Now consider \shareid{} in \Cref{fig:alg-srb-shareid}.
Consider an arbitrary phase among the $O(\log n)$ phases.
Sending $ID$ messages incurs time complexity $O(1)$ and message complexity $O(n\log n)$.
In the first call of \bounce{},
by a Chernoff bound, with probability at least $1-n^{-8c}$, any $v \in G$ sends less than $com\_all$ messages to any node, thus this invocation incurs time complexity $O(\log n)$ and message complexity $O(n\log^2 n)$.
Similarly, the second call of \bounce{} has time complexity $O(\log n)$ and message complexity $O(n\log^2 n)$.
Sending $ECHO1$ and $ECHO2$ messages incur time complexity $O(1)$ and message complexity $O(n\log n)$.
By the analysis in \cite{lynch1996distributed}, the execution of \sharecon{} incurs time complexity $O(\log n)$ and message complexity $O(\log^2 n)$.
Lastly, sending $RET$ messages incurs time complexity $O(1)$ and message complexity $O(n\log n)$.

To sum up, the total time complexity is $O(\log^6n)$, and the total message complexity is $O(n\log^3 n)$.
\end{proof}

We conclude this section with a proof of \Cref{thm:with-renaming}.

\begin{proof}[Proof of \Cref{thm:with-renaming}]
By \Cref{lem:with-s}, with probability at least $1-n^{-8c}$, there are more than $\frac{2}{3}com\_all$ correct nodes with $elected=1$, so the condition $lead^* \in G$ in \Cref{lem:with-naming} is satisfied at least once among all phases. Then by \Cref{lem:with-naming}, with probability at least $1-n^{-7c}$, all correct nodes halt with $\nid(\cdot)$ forming a strong and order preserving renaming among all correct nodes.
By \Cref{lem:with-complex}, with probability at least $1-n^{-6c}$, the entire algorithm has time complexity $\Tilde{O}(1)$ and message complexity $\Tilde{O}(n)$.
\end{proof}

\section{Renaming Algorithm without Shared Random Bits}\label{sec:alg-without-shared-rnd}

\subsection{Committee Election}\label{subsec:alg-com-without-shared-rnd}

\begin{figure}[t!]
\hrule
\vspace{1ex}
$\committee$ executed at node $v$.
\vspace{.5ex}
\hrule
\begin{footnotesize}
\begin{algorithmic}[1]
\State $p\gets(C\log n)/n$, $elected_v \gets 0$, $S_v \gets \emptyset$, $S_v' \gets \emptyset$.
\For{($j \gets 1$ \textbf{to} $\log n$)}
	\State $p \gets (C\log n)/n \cdot 2^{j-1}$.
	\If{(\highlightblue{$p \ge  1/C$})}
		$p\gets 1$. 
	\EndIf
	\State $ S_v \gets \emptyset$, $ElectMsg_v \gets$ an empty map.
	\State Set $elected_v \gets 1$ with probability $p$, otherwise $elected_v \gets 0$. 
	\If {($elected_v==1$)}
		Send $\langle ELECT, \id(v)\rangle$ to all nodes.
	\EndIf
	\For {(every $\langle ELECT, \id(u)\rangle$ message received)}
		\State $S_v \gets S_v \cup \{ \id(u) \}$.
		\State $ElectMsg_v[\id(u)] \gets \langle ELECT, \id(u)\rangle$.
	\EndFor
\If{($p==1$)}
	\highlight{$S_v' \gets S_v$, $redo\gets 0$, $com\_all \gets n+1$, $com\_g \gets (2/3+\epsilon)n$, $com\_b \gets com\_all - com\_g$}.
\Else
	\State {$com\_all \gets (1+2\epsilon) pn$, $com\_g \gets (1-\epsilon)(2/3 + \delta) pn$, $com\_b \gets com\_all - com\_g$}.
	\State $(S_v', redo_v)\gets $ Run $ \clst( elected_v, S_v, ElectMsg_v, com\_all, com\_g, com\_b)$ with all nodes.
\EndIf
	\If{($redo_v==0$)}
		\textbf{break}.
	\EndIf
\EndFor
\end{algorithmic}
\end{footnotesize}
\hrule
\vspace{1ex}
\caption{Pseudocode of \committee.}\label{fig:alg-nsrb-comelect}
\vspace{-3ex}
\end{figure}

Similar to the setting when shared randomness is available, to accomplish renaming, we begin with electing a committee of bounded size.
However, without shared randomness, this process becomes much more complicated.
Moreover, after this process, all nodes may have different views on the constitution of the committee.
Nevertheless, by guaranteeing that the intersection of all correct nodes' committee lists contains many correct nodes, and that the union of all correct nodes' committee lists contains limited Byzantine nodes, this committee is sufficient to help us accomplish remaining (and many other tasks) efficiently.

\Cref{fig:alg-nsrb-comelect} gives the pseudocode of the committee election process when shared randomness is not available. This algorithm operates in iterations, progressively increasing the probability of nodes joining the committee until a valid committee with desirable properties is obtained.

More specifically, the algorithm is initialized with the base probability $(C\log n)/n$, where $C$ is a sufficiently large constant of our choice. After each iteration, the probability $p$ of joining the committee doubles. When $p$ exceeds the threshold $1/C$, all nodes join the committee (set $p$ to $1$), ensuring termination within $\log n$ iterations. In each iteration, each node $v$ independently elects itself as a committee member with probability $p$ by setting $elected_v $ to $1$ and broadcasting an $ELECT$ message that contains its identity $\id(v)$.
Each node $v$ maintains a local view of the committee, denoted as $S_v$, by aggregating the $ELECT$ messages received.
Each node also computes parameters relevant to the committee's size.
Particularly, when $p < 1$, the size of the committee is bounded by $com\_all=(1+2\epsilon)pn$, with at least $com\_g = (1-\epsilon)(2/3 + \delta)pn$ correct nodes and less than $com\_b = com\_all - com\_g$ Byzantine nodes. Then, the $\clst$ procedure is executed at each node $v$ to transform the initial set of committee nodes $S_v$ into a new \highlightblue{filtered} set of committee nodes $S_v'$ and produces a status indicator $redo_v$ signaling the validity of $S_v'$.
The committee election process is successful if $redo_v$ equals 0. (Our algorithm ensures all nodes obtain identical $redo$ value.)
On the other hand, when $p=1$ (meaning all nodes join the committee), parameters are set such that $redo_v=0$, $com\_all = n+1$, \highlight{$com\_g=(2/3+\epsilon)n$}, and $com\_b=com\_all - com\_g$.
Next, we introduce the $\clst$ procedure in detail.

\begin{figure}[t!]
\hrule
\vspace{1ex}
$\clst( elected_v, S_v, ElectMsg_v, com\_all, com\_g, com\_b)$ executed at node $v$.
\vspace{.5ex}
\hrule
\begin{footnotesize}
\begin{algorithmic}[1]
\State $S_v' \gets $ Run $\sample(1, 1, S_v, (1-\epsilon)(2/3+\delta))$.
\State $Serve_v \gets \emptyset$.
\For{(every $\id(u) \in S_v$)}
	Add $\id(u)$ to set $Serve_v$ with probability $(C\log n) / n$.
\EndFor
\For{(every $\id(u) \in Serve_v$)}
	Send $\langle SERVE, ElectMsg_v[\id(u)]\rangle$ to all nodes.
\EndFor
\State $Long_v \gets S_v$.
\For{($i \in [N]$)}
	$Cnt_v[i] \gets 0$, $IdCnt_v[i] \gets 0$.
\EndFor
\For{(every $\langle SERVE, \langle ELECT, \id(u) \rangle \rangle$ message received from node $w$)}
	\If{($Cnt_v[\id(w)] < (1+\epsilon)C\log n$)}
		\State $Cnt_v[\id(w)] \gets Cnt_v[\id(w)]+1$, $IdCnt_v[\id(u)] \gets IdCnt_v[\id(u)] +1$.
		\If{($IdCnt_v[\id(u)] \ge (C\log n)/10$)}
			\State $Long_v \gets Long_v \cup \{ \id(u) \}$.
		\EndIf
	\EndIf
\EndFor
\If{($|Long_v| \ge com\_all$)}
	Send $\langle LONG\rangle$ to all nodes in $S_v$.
\EndIf
\State $com\_out_v \gets \bot$.
\If{($elected_v == 1$)}
	\State $ com\_ready_v \gets 0 $, $ com\_in_v \gets 0$.
	\If{(receive at least $\epsilon n$ $\langle LONG\rangle$ messages)}
		$com\_in_v \gets 1$.
	\EndIf
	\If{(receive at least $(1/3-\delta+\epsilon)n$ $\langle LONG\rangle$ messages)}
		$com\_ready_v \gets 1$.
	\EndIf
	\State $com\_out_v\gets $ Run $\comcon( com\_ready_v, com\_in_v, S_v', com\_all, com\_g,  com\_b)$ with nodes in $S_v'$.
\EndIf
\State $redo_v\gets $ Run $\redo(elected_v, S_v', com\_out_v, com\_all)$ with all nodes.
\State \textbf{return} $(S_v', redo_v)$.
\end{algorithmic}
\end{footnotesize}
\hrule
\vspace{1ex}
\caption{Pseudocode of $\clst$. }\label{fig:alg-clst}
\vspace{-3ex}
\end{figure}

\paragraph{The $\clst$ procedure.}
The $\clst$ procedure filters the initial set of committee members $S_v$ into a smaller (and more reliable) set of committee members $S_v'$ and produces an indicator $redo_v$ denoting the validity of $S_v'$; see \Cref{fig:alg-clst} for its pseudocode.

The $\clst$ procedure executed at node $v$ initiates with a call to the $\sample$ primitive (recall \Cref{fig:alg-sample}). Here, every node participates both as a querying node and as a responding node, using its local committee list $S_v$ as the set of data items to verify.
This allows every node $v$ to detect Byzantine candidates in $S_v$ which do not widely announce their candidacy, producing a refined committee list $S_v'$.

Following the executions of $\sample$, for any node $v$, for each node $u$ in $S_v$, node $v$ adds $u$ to the set $Serve_v$ with probability $(C\log n) / n$.
Node $v$ then broadcasts a set of $SERVE$ messages, each containing an original $ELECT$ message for the nodes in $Serve_v$. (Recall $ELECT$ message are collected during \committee, see \Cref{fig:alg-nsrb-comelect}.)
This enables each node $v$ to construct an extended committee list $Long_v$ by augmenting its local list $S_v$ with identities that appear in sufficiently many $SERVE$ messages.

At this point, all nodes need to reach consensus on the validity of the current committee.
To this end, we first achieve consensus within the current committee.
For each node $v$, it sends a $LONG$ message to the nodes in $S_v$ if the size of $Long_v$ exceeds $com\_all$.
For each committee member $v$, it sets $com\_in_v$ to $1$ if at least $\epsilon n$ $LONG$ messages are received, and it sets $com\_ready_v $ to $1$ if at least $(1/3-\delta+\epsilon)n$ $LONG$ messages are received.
Intuitively, $com\_in_v=1$ indicates committee member $v$ votes for ``my current committee seems invalid and next iteration \emph{might be} needed'', whereas $com\_ready_v=1$ indicates ``next iteration \emph{is} needed'' (our algorithm ensures any $com\_ready_v=1$ triggers next iteration).
Once $com\_in_v$ and $com\_ready_v$ are set, each committee member $v$ executes the $\comcon$ Byzantine consensus procedure, which will be introduced later, with the refined committee members in $S_v'$ to reach consensus on $com\_out_v$.
Then, each node $v$ executes the $\redo$ Byzantine consensus procedure, which again will be introduced later, with all nodes to reach consensus on $redo_v$.
The $\clst$ procedure stops once all nodes find $redo=0$.
At that point, each node $v$ uses $S_v'$ as its final committee list.

\begin{figure}[t!]
\hrule
\vspace{1ex}
$\comcon( com\_ready_v, com\_in_v, S_v', com\_all, com\_g,  com\_b)$ executed at node $v$.
\vspace{.5ex}
\hrule
\begin{footnotesize}
\begin{algorithmic}[1]
\State $com\_out_v \gets com\_in_v$.
\For{($i\gets 1$ \textbf{to} $\Theta(\log n)$)}
	\State Send $\langle INIT, com\_out_v \rangle$ to nodes in $S_v'$.
	\If{(at least $com\_g$ $\langle INIT, 1 \rangle$ messages are received from nodes in $S_v'$)}
		Send $\langle ECHO, 1 \rangle$ to nodes in $S_v'$.
	\ElsIf{(at least $com\_g$ $\langle INIT, 0 \rangle$ messages are received from nodes in $S_v'$)}
		Send $\langle ECHO, 0 \rangle$ to nodes in $S_v'$.
	\EndIf
	\If{($com\_ready_v==0$)}
	\State $strong_v \gets 0$.
		\If{(at least $com\_b$ $\langle ECHO, 1 \rangle$ messages are received from nodes in $S_v'$)}
			$com\_out_v \gets 1$.
		\ElsIf{(at least $com\_b$ $\langle ECHO, 0 \rangle$ messages are received from nodes in $S_v'$)}
			$com\_out_v \gets 0$.
		\EndIf
		\If{(at least $com\_g$ $\langle ECHO, 1 \rangle$ messages are received from nodes in $S_v'$)}
			$strong_v \gets 1$.
		\ElsIf{(at least $com\_g$ $\langle ECHO, 0 \rangle$ messages are received from nodes in $S_v'$)}
			$strong_v \gets 1$.
		\EndIf
	\EndIf
	\State $lead_v\gets $ Run {$\largelead(S_v', com\_all, com\_g, com\_b)$} with all nodes in $S_v'$.
	\If{($\id(v)==lead_v$)}
		Send $\langle LEAD, com\_out_v \rangle$ to all nodes in $S_v'$.
	\EndIf
	\If{($com\_ready==0$ \textbf{and} $strong_v==0$)} 
		\State $com\_out_v \gets$ value of $ com\_out_u$ in the $\langle LEAD, com\_out_u \rangle$ message received from the node with identity $lead_v$.
	\EndIf
\EndFor
\State \textbf{return} $com\_out_v$.
\end{algorithmic}
\end{footnotesize}
\hrule
\vspace{1ex}
\caption{Pseudocode of $\comcon$.}\label{fig:alg-comcon}
\vspace{-3ex}
\end{figure}

\begin{figure}[t!]
\hrule
\vspace{1ex}
$\redo(elected_v, S_v', com\_out_v, com\_all)$ executed at node $v$.
\vspace{.5ex}
\hrule
\begin{footnotesize}
\begin{algorithmic}[1]
\State $redo_v \gets 0$.
\If {($elected_v==1$)}
	Send $\langle COM, com\_out_v\rangle$ to all nodes.
\EndIf
\highlight{
\If{($|S_v'| < com\_all$)}
	$redo_v \gets$ the majority value of $com\_out$ among all received $COM$ messages from $S'_v$.
\EndIf
\State $redo_v' \gets$ Run $\sample{}(\mathbb{1}_{|S_v'| \geq com\_all},\mathbb{1}_{|S_v'| < com\_all}, \{redo_v\}, 1/2)$.
\If{($|S_v'| \geq com\_all$ \textbf{and} $0 \notin redo_v'$)}
	$redo_v \gets 1$.
\EndIf
}
\State \textbf{return} $redo_v$.
\end{algorithmic}
\end{footnotesize}
\hrule
\vspace{1ex}
\caption{Pseudocode of \redo.}\label{fig:alg-redo}
\vspace{-3ex}
\end{figure}

\paragraph{$\comcon$ and $\redo$.}
Now we introduce the two Byzantine consensus algorithms mentioned above: namely $\comcon$ and $\redo$.

The pseudocode of $\comcon$ is given in \Cref{fig:alg-comcon}.
This procedure operates through $\Theta(\log n)$ iterations. For each node $v$, the protocol initializes its output $com\_out_v$ with its input $com\_in_v$, then enters the core loop.
In each iteration, node $v$ first sends its current decision via $INIT$ messages to the committee members in $S_v'$. Then, it processes the incoming $INIT$ messages from the committee members in $S_v'$ as follows: upon receiving at least $com\_g$ $INIT$ messages with identical values (either 1 or 0), the node echos this value to the nodes in $S_v'$ via $ECHO$ messages.
At this point, for a node $v$ with input $com\_ready_v = 0$, it initializes $strong_v$ to $0$.
(If $v$ has $com\_ready_v = 1$, then it is confident that all nodes' eventual agreed value will be identical to $v$'s input, hence $v$ does not need to participate in subsequent process.)
Then, it sets $com\_out_v$ and $strong_v$ as follows. If at least $com\_b$ $ECHO$ messages that contain $1$ are received from $S_v'$, it sets $com\_out_v$ to $1$; otherwise, if at least $com\_b$ $ECHO$ messages that contain $0$ are received, it sets $com\_out_v$ to $0$.
Moreover, if at least $com\_g$ $ECHO$ messages that contain $1$ are received from $S_v'$, it sets $strong_v$ to $1$; otherwise, if at least $com\_g$ $ECHO$ messages that contain $0$ are received, it sets $strong_v$ to $0$.
Intuitively, once a node $v$ finds $strong_v=1$, then it is confident that all nodes have identical decision value.
For the nodes that are still uncertain with their decisions (i.e., $com\_ready_v = 0$ and $strong_v = 0$), extra leader coordination is employed.
Specifically, committee members execute $\largelead$ to elect a leader $u$, which broadcasts its decision $com\_out_u$ via $LEAD$ messages to all nodes in the committee.
{($\largelead$ employs the leader election algorithm in \cite{king06soda}, augmented with the techniques introduced in \cite{augustine20} to ensure compatibility with our specific model.)}
Only nodes with uncertain decisions ($com\_ready_v = 0$ and $strong_v = 0$) adopt the leader's proposal, preserving consensus integrity while preventing indecision.

The $\redo$ procedure disseminates the committee's decision to all nodes, see \Cref{fig:alg-redo} for its pseudocode.
Within the procedure, each node $v$ initializes its decision $redo_v$ to $0$. Each committee member $u$ broadcasts its consensus result $com\_out_u$ (obtained during \comcon) via $COM$ messages to all nodes.
Then, we handle nodes differently based on whether they possess oversized committee lists.
On the one hand, if the size of the verified committee $S_v'$ is less than $com\_all$, then node $v$ updates $redo_v$ to be the majority value contained within the $COM$ messages received from $S_v'$.
On the other hand, if node $v$ detects an oversized committee (i.e., $|S_v'| \geq com\_all$), then it triggers an execution of the $\sample$ primitive. In this invocation, nodes with oversized committees act as querying nodes, issuing queries with data item $0$, while nodes with proper-size committees act as responding nodes, using their local $redo$ as the data item. By setting proper threshold for $\sample$, each querying node can determine whether $0$ is held by a majority of nodes with proper-size committees and set its $redo$ accordingly. This ensures all nodes eventually obtain a consistent value of $redo$.

\subsection{Distribute New Identities}\label{subsec:alg-rename-without-shared-rnd}

\begin{figure}[t!]
\hrule
\vspace{1ex}
$\oldid(elected_v, S_v', com\_all, com\_g,  com\_b)$ executed at node $v$.
\vspace{.5ex}
\hrule
\begin{footnotesize}
\begin{algorithmic}[1]
\State $ redo_v \gets 1 $, $\nid(v) \gets \id(v)$.
\For{($i\gets 1$ \textbf{to} $\Theta(\log n)$)}
	\State  $lead_v\gets $ Run $\lead(elected_v, S_v', com\_all, com\_g, com\_b)$ with all nodes.
	\State  $(List_v, List_v') \gets$ Run $\echoone(elected_v, S_v', lead_v, com\_all)$ with all nodes.
	\State $(\nid(v), Msg_v, com\_redo_v)\gets$ Run $\echotwo(elected_v, S_v', lead_v, List_v, List_v')$ with all nodes.
	\State  $redo_v\gets $ Run $\echothree(elected_v, S_v', lead_v, List_v, List_v', com\_redo_v, Msg_v, com\_all, com\_g, com\_b)$ with all nodes.
	\If{($redo_v==0$)}
		\textbf{break}.
	\EndIf
\EndFor
\State \textbf{return} $\nid(v)$.
\end{algorithmic}
\end{footnotesize}
\hrule
\vspace{1ex}
\caption{Pseudocode of \oldid.}\label{fig:alg-oldid}
\vspace{-3ex}
\end{figure}

After a reliable committee $S_v'$ is elected (that is, after the execution of \committee), committee members try to distribute new identities to all nodes using the \oldid{} procedure, as specified in \Cref{fig:alg-oldid}. This algorithm transforms nodes' identity namespace from $[N]$ to $[n]$.
At a high-level, \oldid{} employs similar strategies of \shareid{}. However, without shared randomness, the elected committee is not necessarily of size $O(\log n)$ and different nodes may hold different views on the constitution of the committee. Therefore, \oldid{} relies on extra mechanisms to enforce correctness and efficiency, resulting in increased complexity.

Within \oldid, each node $v$ initializes $redo_v$ to $1$ and $\nid(v)$ to $\id(v)$ before entering a core iterative loop for $O(\log n)$ times. Each iteration progresses through four subroutines.
First, the $\lead$ procedure elects a new leader among the committee members.
Subsequently, the \echoone{} procedure allows committee members to collect nodes' current identities.
Specifically, if node $v$ is a committee member, the returned list $List_v$ contains $v$'s local view on the identities that are present. If node $v$ is the leader, \echoone{} further returns an augmented list $List_v'$, which encompasses a more comprehensive set of identities.
The core new identity assignment process occurs in the $\echotwo$ procedure, where the leader computes and distributes new rank-based identities for all nodes, after which committee members endorse and confirm these assignments via \highlight{$ECHO1$} messages.
After the execution of this procedure, each node $v$ holds its new identity (i.e., $\nid(v)$), as well as the message that announced this new identity sent by the leader (i.e., $Msg_v$).
For each committee member, \echotwo{} further returns an indicator signaling the validity of the current new identity assignment.
Finally, the $\echothree$ procedure performs comprehensive consistency verification across all identity assignments, and set $redo$ value accordingly.
After nodes reach consensus on their $redo$ values, they continue into the next iteration if and only if the consensus outcome is $1$.

Next, we describe these four subroutines in detail.

\begin{figure}[t!]
\hrule
\vspace{1ex}
$\lead(elected_v, S_v', com\_all, com\_g, com\_b)$ executed at node $v$.
\vspace{.5ex}
\hrule
\begin{footnotesize}
\begin{algorithmic}[1]
\If{($elected_v==1$)}
	\State $com\_lead_v\gets $ Run {$\largelead(S_v', com\_all, com\_g, com\_b)$} with all nodes in $S_v'$.
	\State $com\_lead_v' \gets com\_lead_v$.
	\State Send $\langle INIT, com\_lead_v' \rangle$ to all nodes in $S_v'$.
	\If{(at least $com\_g$ of $\langle INIT, com\_lead_w' \rangle$ messages are received from nodes in $S_v'$ for some $ com\_lead_w'$ )}
		\State Send $\langle ECHO, com\_lead_w' \rangle$ to all nodes in $S_v'$.
	\EndIf
	\If{(at least $com\_b$ messages of $\langle ECHO, com\_lead_w' \rangle$ are received from nodes in $S_v'$ for some $ com\_lead_w'$ )}
		\State $com\_lead_v' \gets com\_lead_w'$.
	\EndIf
	\If{(at least $com\_g$ messages of $\langle ECHO, com\_lead_w' \rangle$ are received from nodes in $S_v'$ for some $ com\_lead_w'$)}
		\State $strong_v \gets 1$.
	\EndIf
	\State $strong_v'\gets$ Run $\comcon(0, strong_v, S_v', com\_all, com\_g,  com\_b)$ with all nodes in $S_v'$.
	\If{($strong_v'==0$)}
		$com\_lead_v' \gets \bot$.
	\EndIf
	\State Send $\langle LEAD, com\_lead_v' \rangle$ to all nodes.
\EndIf
\State $lead_v \gets $ the majority value of $com\_lead'$ among $LEAD$ messages received from $S_v'$.
\State \textbf{return} $lead_v$.
\end{algorithmic}
\end{footnotesize}
\hrule
\vspace{1ex}
\caption{Pseudocode of \lead.}\label{fig:alg-lead}
\vspace{-3ex}
\end{figure}

\paragraph{The $\lead$ procedure.}
This procedure implements a leader election mechanism that combines ``consensus within committee'' with ``global notification'' to elect a leader that is known to all nodes, as formalized in \Cref{fig:alg-lead}.
The procedure begins by letting committee members execute a leader election process: for a committee member $v$, it first runs the $\largelead$ subroutine within the verified committee $S_v'$ to obtain the identity of a candidate leader $com\_lead_v$.
Node $v$ also stores the value of $com\_lead_v$ in another variable $com\_lead_v'$.
Then, committee member $v$ broadcasts an $INIT$ message containing the identity of $com\_lead_v'$ to the committee $S_v'$.
Upon receiving $com\_g$ $INIT$ messages with identical leader identity from nodes in $S_v'$, node $v$ sends an $ECHO$ message containing this leader identity to the committee $S_v'$.
Once $com\_b$ $ECHO$ messages are received from the committee $S_v'$ containing an identical leader identity, node $v$ updates $com\_lead_v'$ to that identity.
Moreover, once $com\_g$ $ECHO$ messages are received from the committee $S_v'$ containing an identical leader identity, node $v$ sets $strong_v$ to $1$.
At this point, each committee member $v$ runs the $\comcon$ procedure with the committee $S_v'$ to reach agreement on $strong_v$.
Let the consensus outcome be $strong_v'$.
In case $strong_v'=0$, implying consensus on leader cannot be achieved, node $v$ sets $com\_lead_v'$ to $\bot$ to prevent unreliable leadership. 
Committee members then disseminate their final leadership decisions through $LEAD$ messages to all nodes.
Every node $v$ sets $lead_v$ as the majority value among the $LEAD$ messages received from its committee list $S_v'$.
In a nutshell, \lead{} guarantees: the procedure either produces a global leader, or all nodes decide that no reliable leader can be established, formally indicated by setting $lead = \perp$.

\begin{figure}[t!]
\hrule
\vspace{1ex}
$\echoone(elected_v, S_v', lead_v, com\_all)$ executed at node $v$.
\vspace{.5ex}
\hrule
\begin{footnotesize}
\begin{algorithmic}[1]
\State $PreList_v \gets \emptyset$, $PreListMsg_v \gets \emptyset$, $List_v \gets \emptyset$, $List_v' \gets \emptyset$.
\State Send $\langle ID, \id(v) \rangle$ to all nodes in $S_v'$.
\If{($elected_v==1$)}
	\State $Node_v \gets$ an empty map.
	\For{(every $\langle ID, \id(u) \rangle$ message received)}
		\State $PreList_v \gets PreList_v \cup \{ \id(u) \}$, \highlight[cyan]{$PreListMsg_v \gets PreListMsg_v \cup \{ \langle ID, \id(u) \rangle \}$.}
		\State Reply $\langle ANS, \id(v), \id(u) \rangle$ to node $u$.
		\State $Node_v[\id(u)] \gets C\log n$ random identities in $S_v'$.
		\For{(every $\id(w) \in Node_v[\id(u)]$)}
			Send $\langle ASK, \id(v), \id(u), \id(w) \rangle$ to node $u$.
		\EndFor
	\EndFor
\EndIf
\State $Ans_v \gets$ an empty map, $AnsCnt_v \gets$ an empty map.
\For{(every $\langle ANS, \id(w), \id(v) \rangle$ message received from {some node $w$ with $\id(w)\in S_v'$})}
	\State $Ans_v[\id(w)] \gets \langle ANS, \id(w), \id(v) \rangle$.
\EndFor
\For{(every $\id(u) \in S_v'$)}
	$AnsCnt_v[\id(u)] \gets 0$.
\EndFor
\For{(every $\langle ASK, \id(w), \id(v), \id(u) \rangle$ message received from some node $w$ with $\id(w)\in S_v'$)}
	\If{(($|Ans_v[\id(u)]|>0$) \textbf{and} ($AnsCnt_v[\id(w)] \le C\log n$))}
		\State Reply $\langle VALID, \id(v), Ans_v[\id(u)] \rangle$ to node $w$.
		\State $AnsCnt_v[\id(w)] \gets AnsCnt_v[\id(w)]+1$.
	\EndIf
\EndFor
\If{($elceted_v==1$)}
	\For{(every $\id(u) \in PreList_v$)}
		\For{($\langle VALID, \id(u), \langle ANS, \id(w), \id(u) \rangle \rangle$ messages received from node $u$)}
			\If{(there are at least $(C\log n)/2$ messages with $ \id(w) \in  Node_v[\id(u)]$)}
				\State $ List_v \gets List_v \cup \{ \id(u) \}$.
			\EndIf
		\EndFor
	\EndFor
\EndIf
\State $PreListMsg_v' \gets $ Run $\bounce(elected_v, S_v', \{lead_v\}, PreListMsg_v, (C\log n)/com\_all)$.
\If{($\id(v)==lead_v$)}
	\State $List_v' \gets$ all identities in \highlight{$PreListMsg_v'$}.
\EndIf
\State \textbf{return} $(List_v, List_v')$.
\end{algorithmic}
\end{footnotesize}
\hrule
\vspace{1ex}
\caption{Pseudocode of \echoone.}\label{fig:alg-echoone}
\vspace{-3ex}
\end{figure}

\paragraph{The $\echoone$ procedure.}
This procedure aggregates the identities of all nodes to the committee and the leader for new identity assignment, as formalized in \Cref{fig:alg-echoone}.
$\echoone$ begins with all nodes broadcasting their original identities to the committee via $ID$ messages.
For each committee member $v$, it builds a preliminary identity list $PreList_v$ according to the received $ID$ messages, \highlightblue{it also stores all received $ID$ messages in $PreListMsg_v$}.
For each $ID$ message received from node $u$, committee member $v$ sends an acknowledgment back to $u$ via an $ANS$ message.
The \echoone{} procedure also implements a validation mechanism.
Specifically, for each $ID$ message received from node $u$, committee member $v$ selects $C\log n$ random nodes from the committee and send $ASK$ messages to $u$ to ask for the $ANS$ messages that $u$ received from these committee members.
On the other hand, for any node $u$, upon receiving an $ASK$ message from committee $v$ inquiring about committee $w$'s acknowledgment for $u$'s $ID$ message, $u$ performs two checks: (1) whether the number of responses already sent to $v$ remains below threshold $C \log n$, and (2) whether an $ANS$ message from $w$ has been previously received. If both conditions are satisfied, $u$ responds a $VALID$ message to $v$ containing the $ANS $ message received from $w$.

At this point, each committee member $v$ constructs its final identity list $List_v$ through the following process: for each identity $\id(u)$ in the preliminary list $PreList_v$, $\id(u)$ is retained if $v$ has received at least $(C\log n)/2$ acknowledgments for $\id(u)$ from the nodes that were queried via $ASK$ messages.
Then, each committee member $v$ sends $PreListMsg_v$ to the leader via the $\bounce$ primitive. In this invocation of \bounce, each node $u$ accepts any node in its committee list $S'_u$ as a message source with probability $(C\log n)/com\_all$, avoiding excessive message complexity. The leader aggregates all received $PreListMsg$ into \highlight{$PreListMsg'$} and then extracts the identities contained within \highlight{$PreListMsg'$} to construct a comprehensive list $List'$.
$\echoone$ ultimately returns a local identity list \highlight{$List$} for each committee member, and a comprehensive identity list $List'$ for the leader.

\begin{figure}[t!]
\hrule
\vspace{1ex}
$\echotwo(elected_v, S_v', lead_v, List_v, List_v')$ executed at node $v$.
\vspace{.5ex}
\hrule
\begin{footnotesize}
\begin{algorithmic}[1]
\If{($\id(v)==lead_v$)}
	\State $RankMsg_v \gets \emptyset$.
	\For{(every $\id(u) \in List_v'$)}
		\State $rank_u \gets$ the rank order of $\id(u)$ in $List_v'$.
		\State $RankMsg_v \gets RankMsg_v \cup \{ \langle NewID, \id(v), \id(u), rank_u \rangle \}$.
	\EndFor
\EndIf
\State $RankMsg_v' \gets $ Run $\bounce(\mathbb{1}_{\id(v)==lead_v}, \{lead_v\}, S_v', RankMsg_v, 1)$.
\If{($elected_v==1$)}
	\State $com\_redo_v \gets 0$, $Map_v \gets $ an empty map.
	\For{(every $\langle NewID, lead_v, \id(u), rank \rangle \in RankMsg_v'$)}
		\State \highlight[cyan]{$Map_v[\id(u)] \gets \langle NewID, lead_v, \id(u), rank \rangle$.}
	\EndFor
	\For{(every $\id(u) \in List_v$)}
		\If{($Map_v[\id(u)] == \perp$)}
			$com\_redo_v \gets 1$
		\algorithmicelse
			\xspace Send $\langle ECHO1, Map_v[\id(u)] \rangle$ to node $u$.
		\EndIf
	\EndFor
\EndIf
\State $Msg_v \gets \bot$.
\For{(every $\langle ECHO1, \langle NewID, lead_v, \id(v), rank\rangle \rangle$ message received from nodes in $S_v'$)}
	\State $Msg_v \gets \langle NewID, lead_v, \id(v), rank\rangle$, $\nid(v)\gets rank$.
\EndFor
\State \textbf{return} $(\nid(v), Msg_v, com\_redo_v)$.
\end{algorithmic}
\end{footnotesize}
\hrule
\vspace{1ex}
\caption{Pseudocode of \echotwo.}\label{fig:alg-echotwo}
\vspace{-3ex}
\end{figure}

\paragraph{The $\echotwo$ procedure.}
This procedure combines leader-driven identity assignment with committee-based verification to ensure reliable new identity allocation, as formalized in \Cref{fig:alg-echotwo}.
Within $\echotwo$, the leader node first computes and distributes new identity assignments based on its identity list $List'$. Specifically, for each node $u$ with $\id(u)$ in $List'$, the leader calculates the new identity $rank_u$ based on the rank of $\id(u)$ in $List'$.
Then, the leader transmits the new identity assignments to the committee via the $\bounce$ primitive.
Each committee node $v$ uses a mapping $Map_v$ to record the received new identity assignments. It checks through this mapping for each $\id(u)\in List_v$: if the node $u$ lacks a proper assignment (i.e., $Map_v[\id(u)]$ does not have a record), committee node $v$ sets $com\_redo_v$ to $1$; otherwise, $v$ forwards new identity assignment to node $u$ via an $ECHO1$ message. Each node $v$ extracts its new identity $\nid(v)$ from an arbitrary $ECHO1$ message.
Ultimately, when this procedure ends, each node $v$ obtains its new identity. If $v$ is a committee node, it has one more piece of information: an indicator $com\_redo_v$ that will be used in subsequent step.

\begin{figure}[t!]
\hrule
\vspace{1ex}
$\echothree(elected_v, S_v', lead_v, List_v, List_v', com\_redo_v, {Msg_v}, com\_all, com\_g, com\_b)$\\
executed at node $v$.
\vspace{.5ex}
\hrule
\begin{footnotesize}
\begin{algorithmic}[1]
\State Send $\langle ECHO2, \id(v),  Msg_v \rangle$ to nodes in $S_v'$.
\If{($elected_v==1$)}
	\State $Name_v \gets$ an empty map.
	\For{(every $\langle ECHO2, \id(u), \langle NewID, lead_v, \id(u), rank_u\rangle \rangle$ received from node $u$ with $\id(u)\in List_v$)}
		\State $Name_v[\id(u)] \gets rank_u$.
	\EndFor
	\If{($\exists u_1 \neq u_2, Name_v[u_1]==Name_v[u_2]$)}
		$com\_redo_v \gets 1$.
	\EndIf
	\If{($\exists u, Name_v[u]>n$)}
		$com\_redo_v \gets 1$.
	\EndIf
	\If{($\exists u_1 < u_2, Name_v[u_1] > Name_v[u_2]$)}
		$com\_redo_v \gets 1$.
	\EndIf
	\State $com\_redo_v'\gets $ Run $\comcon( 0, com\_redo_v, S_v', com\_all, com\_g,  com\_b)$ with nodes in $S_v'$.
	\State Send $\langle RET,  com\_redo_v' \rangle$ to all nodes.
\EndIf
\State $redo_v \gets$ the majority value of $com\_redo'$ among $RET$ messages received from $S_v'$.
\State \textbf{return} $redo_v$.
\end{algorithmic}
\end{footnotesize}
\hrule
\vspace{1ex}
\caption{Pseudocode of \echothree.}\label{fig:alg-echothree}
\vspace{-3ex}
\end{figure}

\paragraph{The $\echothree$ procedure.}
This procedure, shown in \Cref{fig:alg-echothree}, determines the validity of the new identity assignment in current iteration.
In \echothree, each node first broadcasts its new identity to the committee via $ECHO2$ messages.
\highlightblue{(Note that each $ECHO2$ message contains the original $NewID$ message from the leader.)}
Then, each committee member $v$ constructs a mapping $Name_v$ from old identities to new ones according to the received $ECHO2$ messages.
Committee member $v$ then performs following checks: identity uniqueness (no new identity is claimed by multiple nodes), range validity (no new identity exceeds $n$), and order preserving (new identities preserve the relative order of original identities). Any violation triggers $v$ to vote for next iteration: $v$ sets $com\_redo_v$ to $1$.
Once $com\_redo$ values are set, committee members run $\comcon$ to reach agreement on the value of $com\_redo$: for each committee member $v$, denote the return value of $\comcon$ as $com\_redo_v'$.
The final part is to disseminate the committee's consensus outcome to all nodes: each committee node $v$ broadcasts \highlight{$com\_redo_v'$} to all nodes via $RET$ messages, while every node sets its $redo$ to be the majority value among the received $RET$ messages from the committee.
To sum up, $\echothree$ allows all nodes to correctly reach agreement on whether current new identities are valid or another iteration is required.

\section{Analysis for the Algorithm without Shared Randomness}\label{sec:analysis-no-shared-rnd}

In this section, we analyze the renaming algorithm where shared randomness is not required.
In the following, we still assume less than $(1/3-\delta)n$ nodes are Byzantine, where $\delta > 0$ is an arbitrarily small constant. The definitions of $c$, $\epsilon$, and $C$ are the same with \Cref{sec:analysis-shared-rnd}; \highlightblue{further define $\epsilon'$ to be a small constant taking value $\epsilon/2$}.
We use $G$ to denote correct committee members and use $G'$ to denote all correct nodes.

\subsection{Preliminary}

We begin with the analysis of a subroutine, namely \largelead, which is an extension of known results~\cite{king06soda, augustine20}.
\largelead{} is a leader election algorithm with a time complexity poly-logarithmic in the number of participating nodes, and a message complexity (almost) quadratic in the number of the participating nodes.
Hence, when the participating nodes are the committee members, we obtain a leader with a message complexity scaling with the size of the committee, and therefore the actual number of faulty nodes.
\largelead{} is used in two places in our renaming algorithm:
(1) in \comcon{}, where a leader is needed for a phase king style consensus algorithm;
and (2) in \oldid{}, where a leader is needed for assigning new identities.\footnote{\highlight{Recently, \cite{bhangale25} raised a concern about \cite{king06soda} and provided a solution without affecting the leader election process's performance. We have verified this fix can be incorporated into our setting as well.}}

\begin{lemma}[\cite{king06soda, augustine20}]\label{lem:without-cite-leader}
Assume nodes in $G$ are executing \largelead{} in \Cref{fig:alg-comcon} (and in \Cref{fig:alg-lead}).
Assume the following properties hold:
\begin{itemize}[topsep=1ex,itemsep=0ex]
	\item $G \subseteq \bigcap_{v \in G'} S_v'$;
	\item $|\bigcup_{v \in G}S_v'|<com\_all$, $|G| > com\_g >(\frac{2}{3}+\epsilon')com\_all$, $com\_b = com\_all-com\_g$.
\end{itemize}
Then the following properties hold:
\begin{itemize}[topsep=1ex,itemsep=0ex]
	\item validity:  with constant probability, for any $v \in G$, $lead_v \in G$;
	\item agreement:  with constant probability, for any $v, u \in G$, $lead_v=lead_u$;
	\item it has time complexity $O(\text{poly-log}(com\_all))$ and message complexity $O(com\_all^2\cdot\text{poly-log}(com\_all))$.
\end{itemize}
\end{lemma}

\begin{proof}
King, Saia, Sanwalani, and Vee~\cite{king06soda} devise a leader election algorithm that achieves validity and almost everywhere agreement (i.e., agreement among $1-o(1)$ fraction of correct participating nodes), with a time complexity of $O(\text{poly-log}(com\_all))$ and a message complexity of $O(com\_all \cdot \text{poly-log}(com\_all))$.
However, to make it work in our model, we make the same adaptions as in \cite{augustine20}, where \cite{augustine20} calls \cite{king06soda} as a subroutine.
After the leader election algorithm is executed, to achieve everywhere agreement, we need another round in which each participant that already knows the leader broadcasts its decision while each participant that does not know the leader adopts the majority value among the received messages.
\highlightblue{This incurs one additional round.}
Since the total time complexity is $O(\text{poly-log}(com\_all))$, and there are at most $com\_all^2$ messages in each round, the total message complexity is $O(com\_all^2\cdot\text{poly-log}(com\_all))$.
\end{proof}

\subsection{Correctness of the committee election process}

We now start the analysis for the committee election process. We start with the following lemma which states that committee consensus can be reached under particular conditions.
\comcon{} incorporates a phase king style consensus algorithm inspired by \cite{lenzen22}.
It contains multiple phases.
In each phase, a leader is elected that will broadcast its opinion, and nodes that are not confident about their local opinions will adopt the leader's opinion.
Under lemma assumptions, either all nodes have the same input which makes all nodes confident about their local opinions and never adopt the leader's opinion, or there are limited Byzantine nodes and the algorithm works like a standard phase king consensus.

\begin{lemma}\label{lem:comcon}
Assume nodes in $G$ are executing \comcon{} in \Cref{fig:alg-comcon}.
Assume the following properties hold:
\begin{itemize}[topsep=1ex,itemsep=0ex]
	\item $G \subseteq \bigcap_{v \in G'} S_v'$;
	\item if there exists $v \in G$ with $com\_ready_v=1$, then for any $v \in G$,  it holds that $com\_in_v=1$;
	\item if there exists $v \in G$ with $com\_ready_v=0$, then $|\bigcup_{v \in G}S_v'|<com\_all$, \\
	$|G| > (1+\epsilon)com\_g >(1+\epsilon)(\frac{2}{3}+\epsilon')com\_all$, $com\_b = com\_all-com\_g$.
\end{itemize}
Then the following properties hold with probability at least {$1-n^{-10c}$}:
\begin{itemize}[topsep=1ex,itemsep=0ex]
	\item validity: if there are at least $com\_g$ nodes in $G$ with $com\_in=value$, then all nodes in $G$ have $com\_out=value$;
	\item agreement: for any $v, u \in G$, it holds that $com\_out_v=com\_out_u$.
\end{itemize}
\end{lemma}

\begin{proof}
We first give an overview of the proof. Recall the algorithm shown in \Cref{fig:alg-comcon}.
Flag $strong$ and flag $com\_ready$ have similar effects, in that setting either flag to $1$ means the node is confident about the output, thus no longer adopts the opinion from the leader.
The algorithm is designed so that if there exists node with $com\_ready=1$, then all nodes have $strong=1$.
On the other hand, as long as the leader is a correct node, nodes that are not confident will have the same output with the confident nodes.
In \Cref{clm:comcon-1}, we prove that if consensus is reached in any phase, then nodes will not change opinions in later phases.
Then in \Cref{clm:comcon-2}, we prove that a correct leader is elected at least once, thus consensus is reached at least once, which means the algorithm still achieves consensus after all phases.

We now prove the lemma. Consider three cases.

\smallskip\textsc{Case one}: all nodes in $G$ have $com\_ready=1$. By assumption, all nodes in $G$ have $com\_in=1$.
Since all nodes in $G$ have $com\_ready=1$, all nodes in $G$ set $com\_out=1$ at the beginning and make no further changes. Therefore, validity and agreement are enforced.

\smallskip\textsc{Case two}: some node in $G$ has $com\_ready=1$ while some other node has $com\_ready=0$.
By assumption, $|\bigcup_{v \in G}S_v'|<com\_all$, {$|G| > com\_g$}.
Also by assumption, all nodes in $G$ have $com\_in=1$, so more than $com\_g$ nodes in $G$ send $\langle INIT, 1 \rangle$ messages, and more than $com\_g$ nodes in $G$ send $\langle ECHO, 1 \rangle$ messages.
\highlightblue{Therefore, after $ECHO$ messages are delivered, all nodes in $G$ with $com\_ready=1$ have $com\_out=1$, while all nodes in $G$ with $com\_ready=0$ have $strong=1$ and $com\_out=1$.}
Note that for any $u \in G$, node $u$ only receives messages from $|S_u'|<com\_all$ nodes, so there are less than $com\_all-com\_g=com\_b$ messages other than $\langle INIT, 1 \rangle$, so $u$ never sets $com\_out_u=0$ because of $ECHO$ messages.
\highlightblue{Lastly, since all nodes in $G$ with $com\_ready=0$ have $strong=1$, no such node changes $com\_out$ any more.}
Therefore, all nodes in $G$ have $com\_out=1$ eventually, as required.

\smallskip\textsc{Case three}: all nodes in $G$ have $com\_ready=0$.
This is the most complicated case.
By assumption, $|\bigcup_{v \in G}S_v'|<com\_all$, $|G| > (1+\epsilon)com\_g >(1+\epsilon)(\frac{2}{3}+\epsilon')com\_all$.

\begin{claiminline}\label{clm:comcon-1}
During a phase of \comcon{}, if at least $com\_g$ nodes in $G$ start with identical $com\_out=value^* \in \{ 0, 1 \}$, then all nodes in $G$ end with identical $com\_out=value^*$.
\end{claiminline}

\begin{proof}
If at least $com\_g$ nodes in $G$ have $com\_out=value^* \in \{ 0, 1 \}$, then at least $com\_g$ nodes in $G$ send $\langle INIT, value^* \rangle$ messages, and at least $com\_g$ nodes in $G$ send $\langle ECHO, value^* \rangle$ messages. Therefore, all nodes in $G$ have $strong=1$ and $com\_out = value^*$ after $ECHO$ messages are delivered.
Note that for any $u \in G$, node $u$ only receives messages from $|S_u'|<com\_all$ nodes, so there are less than $com\_all-com\_g=com\_b$ messages other than $\langle INIT, value^* \rangle$, so $u$ never sets $com\_out_u= \neg value^*$ because of $ECHO$ messages.
Lastly, since all nodes in $G$ have $strong=1$, no node changes $com\_out$ any more, thus all nodes in $G$ have $com\_out = value^*$ by the end of the phase.
\end{proof}

Above claim implies validity for case three. Specifically, if at least $com\_g$ nodes in $G$ have $com\_in= value^* \in \{ 0, 1 \}$ at the beginning of \comcon, then all nodes in $G$ have $com\_out=value^*$ by the end of phase one. By \Cref{clm:comcon-1}, after all phases, all nodes in $G$ have $com\_out=value^*$.

\begin{claiminline}\label{clm:comcon-2}
With probability at least $1-n^{-10c}$, there exists a phase $i$ of \comcon{} such that,
at the end of phase $i$,
for any nodes $v, u \in G$, $com\_out_v=com\_out_u$.
\end{claiminline}

\begin{proof}
By \Cref{lem:without-cite-leader}, during any phase, with constant probability, there exists a node $w \in G$ satisfying $lead_v = \id(w)$ for any $v \in G$. Since there are sufficient many phases, we know with probability at least $1-n^{-10c}$, there exists a phase $i$ and a node $w \in G$ satisfying $lead_v = \id(w)$ for any $v \in G$ in phase $i$.

Consider such a phase $i$, there are two cases.

In the first case, all correct nodes in $G$ have $strong=0$. Then after receiving $LEAD$ messages, all nodes in $G$ have $com\_out = com\_out_w$.

In the second case, there exists $x \in G$ with $strong_x=1$.
Then $x$ receives at least $com\_g$ messages of $\langle ECHO, com\_out_x \rangle$. (Since $com\_g > \frac{2}{3}com\_all > \frac{2}{3}|S_x'|$, node $x$ receives less than $com\_g$ messages of $\langle ECHO, \neg com\_out_x \rangle$.)
Among these messages, at least $com\_g - com\_b > com\_b$ are from nodes in $G$.
So after $ECHO$ messages are delivered, every node in $G$ receives at least $com\_b$ messages of $\langle ECHO, com\_out_x \rangle$, and sets $com\_out$ to $com\_out_x$, including the leader $w$.
Moreover, since at least $com\_b$ correct nodes send $\langle ECHO, com\_out_x \rangle$, we know at least $com\_g-com\_b$ correct nodes send $\langle INIT, com\_out_x \rangle$, which means any correct node receives less than $com\_all - com\_g + com\_b < com\_g$ messages of $\langle INIT, \neg com\_out_x \rangle$, implying no correct node sends $\langle ECHO, \neg com\_out_x \rangle$.
So after $ECHO$ messages are delivered, there is no correct node that receives at least $com\_b$ messages of $\langle ECHO, \neg com\_out_x \rangle$, and {no correct node sets} $com\_out$ to $ \neg com\_out_x$, including the leader $w$.
Therefore, after receiving $LEAD$ messages, for any node $y \in G$ with $strong_y=1$, $com\_out_y$ does not change and $com\_out_y = com\_out_x=com\_out_w$; for any node $z \in G$ with $strong_z=0$, $com\_out_z = com\_out_w$.

We conclude that, in either case, all nodes in $G$ have 
$com\_out = com\_out_w$.
\end{proof}

Above claim implies agreement for case three.
Specifically, by \Cref{clm:comcon-2}, with probability at least $1-n^{-10c}$, there exists a phase $i$ such that $com\_out_v=com\_out_u$ for any nodes $v, u \in G$ at the end of phase $i$.
By \Cref{clm:comcon-1}, for any phase after $i$, for any nodes $v, u \in G$, we still have $com\_out_v=com\_out_u$.
\end{proof}

The following lemma states that the assumptions of \Cref{lem:comcon} indeed are guaranteed.
In particular, it describes the output of procedure \clst{}.

\begin{lemma} \label{lem:without-largecore}
For the \clst{} algorithm shown in \Cref{fig:alg-clst}, the following properties hold with probability at least $1-n^{-7c}$:
\begin{itemize}[topsep=1ex,itemsep=0ex]
	\item $G \subseteq \bigcap_{v \in G'} S_v'$;
	\item if there exists $v \in G$ with $com\_ready_v=1$, then for any $v \in G$, $com\_in_v=1$;
	\item if there exists $v \in G$ with $com\_ready_v=0$, then $|\bigcup_{v \in G}S_v'|<com\_all$, \\
	$|G| > (1+\epsilon)com\_g >(1+\epsilon)(\frac{2}{3}+\epsilon')com\_all$, $com\_b = com\_all-com\_g$.
\end{itemize}
\end{lemma}

\begin{proof}
Recall the pseudocode shown in \Cref{fig:alg-clst}.
For each node, the committee list $S'$ is initialized as $S$.
In \Cref{clm:without-largecore-1}, we prove that the initial committee list $S$ contains all correct committee members.
In \Cref{clm:elected-large-1}, we further prove that these correct committee members are still contained in the filtered list $S'$, giving the first property of the lemma.
\Cref{clm:elected-large-2} gives the second property of the lemma.
In \Cref{clm:elected-large-3}, we show that $SERVE$ messages help each node builds a local long list, which is guaranteed to contain the union list of all correct committee members' filtered lists.
Lastly, in \Cref{clm:elected-large-4}, we show that if there exists a node's local long list that does not exceed the length limit, then all committee members have trust-worthy committee lists, giving the last property of the lemma.

\begin{claiminline} \label{clm:without-largecore-1}
With probability at least $1-n^{-10c}$, it holds that $G \subseteq \bigcap_{v \in G'} S_v$, {$|G| > (1+\epsilon)com\_g >(1+\epsilon)(\frac{2}{3}+\epsilon')com\_all$}, $com\_b = com\_all-com\_g$.
\end{claiminline}

\begin{proof}
$G \subseteq \bigcap_{v \in G'} S_v$ is true since all nodes in $G$ send $ELECT$ messages during \committee{} (see \Cref{fig:alg-nsrb-comelect}).
Note that there are at least $(2/3+\delta)n$ correct nodes and at most $n$ correct nodes, hence if $p=1$, we have $com\_all=n+1$,  $com\_g=(2/3+\epsilon)n$ and $|G|=(2/3+\delta)n$, so the claim holds; if $p<1$, by a Chernoff bound, with probability at least $1-n^{-10c}$, $|G| > (1+\epsilon)com\_g$ and $com\_g > (\frac{2}{3}+\epsilon')com\_all$ are true by the definition of $\epsilon$ and $\epsilon'$, $com\_b = com\_all-com\_g$ is true by the definition of $com\_b$.
\end{proof}

\begin{claiminline}\label{clm:elected-large-1}
With probability at least $1-n^{-9c}$, $G \subseteq \bigcap_{v \in G'} S_v'$.
\end{claiminline}

\begin{proof}
By \Cref{clm:without-largecore-1}, with probability at least $1-n^{-10c}$, it holds that $G \subseteq \bigcap_{v \in G'} S_v$.
Fix a node $w\in G'$ and a node $u \in G$.
When executing \sample{} in \Cref{fig:alg-sample}, node $w$ samples $C\log n$ random nodes to validate $\id(u)$.
By a Chernoff bound, with probability at least $1-n^{-10c}$, node $w$ sends less than $(1+\epsilon)C\log n$ $ASK$ messages to any node.
Since there are at least $(2/3+\delta)n$ correct nodes each having $\id(u)$ in its $S$, by a Chernoff bound, with probability at least $1-n^{-10c}$, node $w$ receives at least $(1 - \epsilon)(2/3+\delta) C \log n$ $YES$ messages for $\id(u)$.
Hence, $\id(u) \in S_w'$.
Take a union bound over at most $n^2$ choices of $w$ and $u$, with probability at least $1-n^{-9c}$, we have $G \subseteq \bigcap_{v \in G'} S_v'$.
\end{proof}

Note that \Cref{clm:elected-large-1} gives the first property in the lemma statement.

\begin{claiminline}\label{clm:elected-large-2}
With probability at least $1-n^{-10c}$,
if there exists some node $v \in G$ with $com\_ready_v=1$, then for any node $v \in G$, it holds that $com\_in_v=1$.
\end{claiminline}

\begin{proof}
By \Cref{clm:without-largecore-1}, with probability at least $1-n^{-10c}$, we have $G \subseteq \bigcap_{v \in G'} S_v$.
Assume there exists $v \in G$ with $com\_ready_v=1$.
Then $v$ receives at least $(1/3-\delta+\epsilon)n$ $LONG$ messages, implying that at least $\epsilon n$ correct nodes have $|Long| \ge com\_all$. 
Since every correct node sends $LONG$ messages to all nodes in $G$,
we know all nodes in $G$ receive at least $\epsilon n$ $LONG$ messages and have $com\_in=1$.
\end{proof}

Note that \Cref{clm:elected-large-2} gives the second property in the lemma statement.

\begin{claiminline}\label{clm:elected-large-3}
With probability at least $1-n^{-9c}$,
$\bigcup_{v \in G'}S_v' \subseteq \bigcap_{v \in G'} Long_v$.
\end{claiminline}

\begin{proof}
There are at most $n$ identities in $S_v$ for any correct node, and there are at most $n$ correct nodes in total, by a Chernoff bound, with probability at least $1-n^{-11c}$, any correct node sends at most $(1+\epsilon)C\log n$ $SERVE$ messages.
\highlightblue{So, for any $y\in G'$ and any identity $\id(z)$,
as long as at least $(C\log n)/10$ correct nodes have $\id(z) \in Serve$, it will be the case that $\id(z) \in Long_y$.}

Fix a node $v \in G'$ and an identity $\id(w) \in S_v'$.

Assume there are less than $n/9$ correct nodes with $\id(w) \in S$.
In such case, when nodes execute \sample{}, node $v$ samples $C\log n$ random nodes and sends $ASK$ messages about $\id(w)$.
Let random variable $X$ be the number of nodes with $\id(w) \in S$ that $w$ samples.
Then $\ee[X] < (1/9+1/3-\delta)C \log n < \frac{1-\epsilon}{1+\epsilon}(2/3+\delta)C \log n$.
Apply a Chernoff bound, we conclude that $\Pr[X \ge (1-\epsilon)(2/3+\delta)C \log n]  \le e^{-{(\epsilon^2/9)} \cdot C \log n} < n^{-11c}$.
That is, if there are less than $n/9$ correct nodes with $\id(w) \in S$, then with probability at least $1-n^{-11c}$, we have $\id(w) \notin S_v'$.

So, assume there are at least $n/9$ correct nodes with $\id(w) \in S$.
By a Chernoff bound, with probability at least $1-n^{-11c}$, there are at least $(C\log{n})/10$ correct nodes with $\id(w) \in Serve$. So every correct node receives at least $(C\log n)/10$ times of $SERVER$ messages containing $\id(w)$ and has $\id(w) \in Long$.

Take a union bound over at most $n^2$ choices of $v$ and $\id(w)$, we know with probability at least $1-n^{-9c}$, it holds that $\bigcup_{v \in G'}S_v' \subseteq \bigcap_{v \in G'} Long_v$.
\end{proof}

\begin{claiminline}\label{clm:elected-large-4}
With probability at least $1-n^{-8c}$, if there exists some $v \in G$ with $com\_ready_v=0$, then  $|\bigcup_{v \in G'}S_v'|<com\_all$, $|G| > (1+\epsilon)com\_g > (1+\epsilon)(\frac{2}{3}+\epsilon')com\_all$, $com\_b = com\_all-com\_g$.
\end{claiminline}

\begin{proof}
Assume there exists some $v \in G$ with $com\_ready_v=0$.
Node $v$ receives less than $(1/3-\delta+\epsilon)n$ $LONG$ messages, which means there are less than $(1/3-\delta+\epsilon)n$ correct nodes with $|Long| \ge com\_all$.
Therefore, there exists a correct node $w$ with $Long_w < com\_all$.
By \Cref{clm:elected-large-3}, with probability at least $1-n^{-9c}$, $\bigcup_{v \in G'}S_v' \subseteq \bigcap_{v \in G'} Long_v$.
Therefore $|\bigcup_{v \in G'}S_v'| \le |Long_w| < com\_all$.
By \Cref{clm:without-largecore-1}, with probability at least $1-n^{-10c}$, we further have $|G| > (1+\epsilon)com\_g >(1+\epsilon)(\frac{2}{3}+\epsilon')com\_all$ and $com\_b = com\_all-com\_g$.
\end{proof}

We conclude the proof by noting that \Cref{clm:elected-large-4} gives the last property in the lemma statement.
\end{proof}

The following lemma shows that all nodes quit \clst{} in the same phase, and the decision is made by committee consensus.

\begin{lemma}\label{lem:without-redo}
For the \clst{} algorithm shown in \Cref{fig:alg-clst}, the following properties hold with probability at least $1-n^{-5c}$:
\begin{itemize}[topsep=1ex,itemsep=0ex]
	\item validity: for any $v \in G'$, $redo_v = com\_out_u$ for some $u \in G$;
	\item agreement: for any $v, u \in G'$, $redo_v = redo_u$.
\end{itemize}
\end{lemma}

\begin{proof}
By \Cref{lem:comcon} and  \Cref{lem:without-largecore}, with probability at least $1-n^{-6c}$, we have $G \subseteq \bigcap_{v \in G'} S_v'$ where $|G|>\frac{2}{3}com\_all$; moreover, there exists some $x \in G$ such that all nodes in $G$ have $com\_out = com\_out_x$.

Recall \redo{} shown in \Cref{fig:alg-redo}.
For any $w \in G'$ with $|S'_w|<com\_all$, after taking the majority value of $COM$ messages, $redo_w = com\_out_x$.

For any $w \in G'$ with $|S_w'| \ge  com\_all$, since there are less than $n/3$ Byzantine nodes, by a Chernoff bound, with probability at least $1-n^{-10c}$, node $w$ samples less than $\frac{1.1}{3}C\log n$ Byzantine nodes during \sample{}.
Consider two cases.

First case, there are at least $\epsilon n$ correct nodes with $|S'| \ge com\_all$. Then there are at least $\epsilon n$ correct nodes with $|Long| \ge com\_all$, so all nodes in $G$ have $com\_in=1$, implying all nodes in $G$ have $com\_out=1$ (which means $com\_out_x=1$).
Hence all correct nodes with $|S'|<com\_all$ have $redo=1$.
This implies no correct nodes reply $YES$ messages on $redo=0$ during \sample, while at most $\frac{1.1}{3}C\log n$ Byzantine nodes do so.
As a result, $redo_w=com\_out_x=1$.

Second case, there are less than $\epsilon n$ correct nodes with $|S'| \ge com\_all$.
Then there are more than $\frac{2}{3}n$ correct nodes with $|S'| < com\_all$ and $redo=com\_out_x$.
By a Chernoff bound, with probability at least $1-n^{-10c}$, $w$ samples more than $\frac{1.9}{3}C\log n$ correct nodes with $|S'| < com\_all$ and less than $\frac{1.1}{3}C\log n$ Byzantine nodes.
If $com\_out_x=1$, then $w$ receives less than $\frac{1.1}{3}C\log n$ $YES$ messages on $redo=0$ during \sample; if $com\_out_x=0$, then $w$ receives more than $\frac{1.9}{3}C\log n$ $YES$ messages on $redo=0$ during \sample.
In either case, we have $redo_w=com\_out_x$.

To conclude, we have shown that with probability at least $1-n^{-5c}$, there exists some $x \in G$ such that all correct nodes have $redo=com\_out_x$.
\end{proof}

We conclude this part with the following lemma, which describes the properties of the eventual committee lists held by nodes.
In a nutshell, though different nodes may have different committee lists, each node's list contains all correct committee members, and correct committee members take the majority in the union of all nodes' lists.
Moreover, the committee size scales with the actual number of Byzantine nodes and the amount of messages sent by Byzantine nodes.
As a result, when there are few Byzantine nodes or Byzantine nodes send few messages, the committee is also small in size.
An immediate corollary of the following lemma is, the common decision of correct committee members can be safely spread to all nodes by requiring each node taking majority value from its local view of the committee.

\begin{lemma}\label{lem:without-s}
For the \committee{} algorithm shown in \Cref{fig:alg-nsrb-comelect}, after its execution, the following properties hold with probability at least $1-n^{-4c}$:
\begin{itemize}[topsep=1ex,itemsep=0ex]
	\item $G \subseteq \bigcap_{v \in G'} S_v'$;
	\item $|\bigcup_{v \in G'}S_v'|<com\_all$, {$|G| > (1+\epsilon)com\_g >(1+\epsilon)(\frac{2}{3}+\epsilon')com\_all$}, $com\_b = com\_all-com\_g$;
	\item $com\_all=O(\min \{ f, T/n \}+\log n)$, where $f$ is the actual number of Byzantine nodes and \highlight{$T$ is the total number of messages sent by Byzantine nodes during \committee}.
\end{itemize}
\end{lemma}

\begin{proof}
The first two properties can be derived from previous lemmas, while the last property can be derived from the following claim (which is due to the ``doubling'' technique used by our algorithm).

\begin{claiminline}\label{clm:without-s-1}
With probability at least $1-n^{-5c}$, during \clst{}, there exists a constant $\delta'>0$ such that, if there exists some $v \in G$ with $com\_out_v=1$, then $f\ge  \delta'\cdot  com\_all$ and $T \ge \delta' \cdot com\_all \cdot n$.
\end{claiminline}

\begin{proof}
By \Cref{lem:comcon} and  \Cref{lem:without-largecore}, with probability at least $1-n^{-6c}$, all nodes in $G$ have the same $com\_out$. Moreover, if there are at least $com\_g$ nodes in $G$ with $com\_in=value$, then $com\_out=value$.

Assume there exists a node in $G$ with $com\_out=1$, then there are at least $\epsilon \cdot com\_g$ nodes in $G$ with $com\_in=1$.
Consider two scenarios.

Scenario one, there are less than $\frac{\epsilon}{2}n$ nodes in $G'$ with $|Long| \ge com\_all$. Since there are at least $\epsilon \cdot com\_g$ nodes in $G$ with $com\_in=1$, we know each of them receives at least $\frac{\epsilon}{2} n$ $LONG$ messages from Byzantine nodes.
So $f \ge \frac{\epsilon}{2} n \ge \frac{\epsilon}{2(1+2\epsilon)}com\_all$ and $T \ge \frac{\epsilon^2}{2}n \cdot com\_g \ge \frac{\epsilon^2}{3} \cdot com\_all \cdot n$.

Scenario two, there are at least $\frac{\epsilon}{2}n$ nodes in $G'$ with $|Long| \ge com\_all$. Let these nodes be node set $Q$.
By a Chernoff bound, with probability at least $1-n^{-10c}$, every node in $Q$ receives at least $\epsilon pn$ Byzantine identities from $ELECT$ messages or $SERVE$ messages, which means $f \ge \epsilon pn \ge \frac{\epsilon}{1+2\epsilon}com\_all$.

To bound $T$ in the second scenario, further consider two cases.

First case, at least $\frac{\epsilon}{4}n$ nodes in $Q$ each receives at least $\frac{\epsilon}{2}pn$ Byzantine identities from $ELECT$ messages. This implies $T \ge \frac{\epsilon^2}{8} pn^2$.

Second case, at least $\frac{\epsilon}{4}n$ nodes in $Q$ each receives at least $\frac{\epsilon}{2}pn$ Byzantine identities from $SERVE$ messages.
Without loss of generality, assume a set $W$ of Byzantine identities are received by at least $\frac{\epsilon}{4}n$ nodes in $Q$ via $SERVE$ messages, and these nodes in $Q$ add these Byzantine identities to $Long$, where $|W| \ge \frac{\epsilon}{2}pn$.
Fix any Byzantine identity $\id(w)$ in $W$.
Note that a Byzantine identity is added to some node's $Long$ list if that identity appears in at least $(C\log n)/10$ $SERVE$ messages.
If there are less than $n/40$ correct nodes with $\id(w) \in S$, then by a Chernoff bound, with probability at least $1-n^{-10c}$, for any node in $Q$, that node receives at most $(C\log n)/20$ $SERVE$ message from correct nodes that include $\id(w)$.
So in total there are at least $(C\log n)/20 \cdot \frac{\epsilon}{4}n$ $SERVE$ messages that include $\id(w)$ are from Byzantine nodes.
If there are at least $n/40$ correct nodes with $\id(w) \in S$, then there are at least $n/40$ $ELECT$ messages which include $\id(w)$ are from Byzantine nodes.
Take a union bound over at most $n$ choices of $\id(w)$, with probability at least $1-n^{-9c}$, Byzantine nodes in total send at least $\frac{\epsilon}{2}pn \cdot \frac{\epsilon}{40}n$ $ELECT$ or $SERVE$ messages, so $T \ge \frac{\epsilon^2}{80}pn^2$.

In either case, with probability at least $1-n^{-9c}$, $T \ge \frac{\epsilon^2}{80(1+2\epsilon)}com\_all \cdot n$.
\end{proof}

We now prove the lemma.
By \Cref{lem:without-redo}, with probability at least $1-n^{-5c}$, all correct nodes quit \committee{} in the same phase $i$.
Consider phase $i$, there are three cases.

Case one, $p<1$ and $i=1$.
Since $i=1$, we know $com\_all=\Theta(\log n)$ by definition, hence $com\_all=O(\min \{ f, T/n \}+\log n)$ is satisfied.

Case two, $p<1$ and $i>1$.
Since $com\_all$ doubles from phase $i-1$ to phase $i$, by \Cref{clm:without-s-1}, with probability at least $1-n^{-5c}$, 
$com\_all/2=O(\min \{ f, T/n \}+\log n)$ is satisfied in phase $i-1$.
So $com\_all=O(\min \{ f, T/n \}+\log n)$.

Case three, $p=1$.
Since we have $p=(C\log n)/n \cdot 2^{i-2}<1/C$ in phase $i-1$ and $p=(C\log n)/n \cdot 2^{i-1} \ge 1/C$ in phase $i$, we know that $com\_all$ grows at most $2C$ times from phase $i-1$ to phase $i$.
By \Cref{clm:without-s-1}, with probability at least $1-n^{-5c}$, 
$com\_all/(2C)=O(\min \{ f, T/n \}+\log n)$ is satisfied in phase $i-1$.
So $com\_all=O(\min \{ f, T/n \}+\log n)$.
\end{proof}

\subsection{Correctness of the renaming process}

We now shift our attention to the renaming process, and we begin with the following lemma.
It states that \lead{} elects a global leader, which, roughly speaking, is achieved by first running \largelead{} within the committee and then propagating the result to all nodes.

\begin{lemma}\label{lem:without-our-leader}
For the \lead{} algorithm shown in \Cref{fig:alg-lead}, the following properties hold with probability at least $1-n^{-3c}$:
\begin{itemize}[topsep=1ex,itemsep=0ex]
	\item validity: with constant probability, for any $v \in G'$, $lead_v \in G$;
	\item agreement: for any $v, u \in G'$, $lead_v=lead_u$.
\end{itemize}
\end{lemma}

\begin{proof}
\lead{} is inspired by the phase king style algorithm in \cite{lenzen22}.
Each committee member holds an indicator $strong$ about whether the committee has the same leader.
Then they reach consensus about the indicator.
If the consensus output suggests the committee has the same leader, the committee distributes the leader to all nodes;
otherwise, the committee distributes $\bot$ to all nodes.
In \Cref{clm:our-leader-1},
we prove that after consensus, the committee either holds the same leader or outputs $\bot$.
In \Cref{clm:our-leader-2},
we prove that if \largelead{} successes, the committee outputs the leader chosen by \largelead{}.

Note that with probability at least $1-n^{-3c}$, \Cref{lem:comcon} and \Cref{lem:without-s} hold. Throughout this proof, assume these lemmas hold.

\begin{claiminline} \label{clm:our-leader-1}
For any nodes $v, u \in G$, $com\_lead_v'=com\_lead_u'$.
\end{claiminline}

\begin{proof}
Consider two cases.
In the first case, all nodes in $G$ have $strong=0$.
Since all nodes in $G$ call \comcon{} (see \Cref{fig:alg-comcon}) with input $com\_ready=0$, by \Cref{lem:comcon}, all nodes in $G$ have $strong'=0$, thus all nodes in $G$ have $com\_lead'=\bot$.

In the second case, there exists $x \in G$ with $strong_x = 1$.
Then $x$ receives at least $com\_g$ messages of $\langle ECHO, com\_lead^* \rangle$ for some $com\_lead^*$. (Moreover, since $com\_g > \frac{2}{3}com\_all > \frac{2}{3}|S_x'|$, $x$ receives less than $com\_g$ messages of $\langle ECHO, com\_lead^\dagger \rangle$ for any $com\_lead^\dagger\neq com\_lead^*$.)
Among these $\langle ECHO, com\_lead^* \rangle$ messages, at least $com\_g - com\_b > com\_b$ are from nodes in $G$.
So all nodes in $G$ receive at least $com\_b$ messages of $\langle ECHO, com\_lead^* \rangle$, and set $com\_lead'$ to $com\_lead^*$.
On the other hand, since at least $com\_b$ correct nodes send $\langle ECHO, com\_lead^* \rangle$, there are at least $com\_g-com\_b$ correct nodes that have sent $\langle INIT, com\_lead^* \rangle$, which means any correct node receives less than $com\_all - com\_g + com\_b < com\_g$ messages of $\langle INIT, com\_lead^\dagger \rangle$ for any $com\_lead^\dagger\neq com\_lead^*$.
So no correct node sends $\langle ECHO, com\_lead^\dagger \rangle$ for any $com\_lead^\dagger\neq com\_lead^*$.
Hence, no correct node receives at least $com\_b$ messages of $\langle ECHO, com\_lead^\dagger \rangle$ for any $com\_lead^\dagger\neq com\_lead^*$, and no correct node sets $com\_lead'$ to $com\_lead^\dagger$ for any $com\_lead^\dagger\neq com\_lead^*$.
At this point, we conclude that, all nodes in $G$ have $com\_lead' = com\_lead^*$ before calling \comcon{}.
Since all nodes in $G$ call \comcon{} with input $com\_ready=0$, by \Cref{lem:comcon}, all nodes in $G$ have the same $strong'$.
If all nodes in $G$ have $strong'=1$, then eventually all nodes in $G$ have $com\_lead'=com\_lead^*$;
otherwise, if all nodes in $G$ have $strong'=0$, then eventually all nodes in $G$ have $com\_lead'=\bot$.
\end{proof}

\begin{claiminline} \label{clm:our-leader-2}
\highlightblue{If there exists $com\_lead^*$ such that every node $v\in G$ has $com\_lead_v=com\_lead^*$}, then:
\begin{itemize}[topsep=1ex,itemsep=0ex]
	\item for any node $v\in G$, $com\_lead_v'=com\_lead^*$;
	\item for any node $v \in G$, $strong_v' =1$.
\end{itemize}
\end{claiminline}

\begin{proof}
Since all nodes in $G$ have  $com\_lead=com\_lead^*$, every node in $G$ sends $\langle INIT, com\_lead^* \rangle$ to all nodes in its $S'$.
By \Cref{lem:without-s}, every node in $G$ receives at least $com\_g$ messages of $\langle INIT, com\_lead^* \rangle$ and less than $com\_b$ messages of $\langle INIT, com\_lead^\dagger \rangle$ for any $com\_lead^\dagger\neq com\_lead^*$.
Hence, every node in $G$ receives at least $com\_g$ messages of $\langle ECHO, com\_lead^* \rangle$ and less than $com\_b$ messages of $\langle ECHO, com\_lead^\dagger \rangle$ for any $com\_lead^\dagger\neq com\_lead^*$.
Therefore, all nodes in $G$ have $com\_lead' = com\_lead^*$ and $strong = 1$ before calling \comcon{}.
That is, all nodes in $G$ invoke \comcon{} with input $com\_ready=0$.
By \Cref{lem:comcon}, all nodes in $G$ eventually have $strong'=1$ and $com\_lead' = com\_lead^*$.
\end{proof}

With above claims, we now prove the lemma.

First, consider validity.
By \Cref{lem:without-cite-leader}, with constant probability, all nodes in $G$ have $com\_lead=com\_lead^* \in G$.
Then by \Cref{clm:our-leader-2},
all nodes in $G$ have $strong'=1$ and 
$com\_lead'=com\_lead^*$.
By \cref{lem:without-s}, after taking majority value among $LEAD$ messages, all correct nodes have $lead=com\_lead^*$.

Next, consider agreement.
By \Cref{clm:our-leader-1}, there exists $com\_lead^*$ such that $com\_lead_v'=com\_lead^*$ for any $v \in G$.
By \cref{lem:without-s}, after taking majority value among $LEAD$ messages, all correct nodes have $lead=com\_lead^*$.
\end{proof}

We conclude this part with \Cref{lem:without-naming}, which states the correctness guarantees enforced by the renaming process.
It assumes the existence of a global leader, which can be obtained by \Cref{lem:without-our-leader}.
If the global leader is a correct node, then all nodes quit \shareid{} in current phase;
otherwise, if the global leader is a Byzantine node and creates an invalid renaming, then all nodes continue into the next phase.

\begin{lemma}\label{lem:without-naming}
For the \oldid{} algorithm shown in \Cref{fig:alg-oldid}, let $G$ be the set of correct nodes with $elected_v=1$ and let $G'$ be the set of all correct nodes, then during any phase, the following properties hold with probability at least {$1-n^{-c}$}:
\begin{itemize}[topsep=1ex,itemsep=0ex]
	\item there exists some $lead^*$ such that $lead_v = lead^*$ for any $v \in G'$;
	\item for any $v, u \in G'$, $redo_v = redo_u$;
	\item if $lead^* \in G$, then $redo_v=0$ for any $v \in G'$;
	\item if there exists some $v \in G'$ with $redo_v=0$, then $\nid(\cdot)$ forms a strong and order preserving renaming among all correct nodes.
\end{itemize}
\end{lemma}

\begin{proof}
The proof strategy is similar to that of \Cref{lem:with-naming}.
In \Cref{clm:without-naming-1}, we show binary consensus can be reached within the reliable committee $S'$, and the result can be distributed to all nodes.
Here, a key difference with \Cref{clm:with-naming-1} is that in \echoone{}, nodes do not forward all messages from the committee (which could cost too much time as the committee might be large).
Instead, each node only forwards messages for $\Theta((\log{n})/|S'|)$ committee members.
We will show that this ``fractional bounce forwarding'' is sufficient, since a message requires forwarding only if it exists in $\Theta({|S'|})$ committee members.
\Cref{clm:without-naming-2} and \Cref{clm:without-naming-3}
are similar with
\Cref{clm:with-naming-3} and \Cref{clm:with-naming-4}, respectively.

We now start the proof for the lemma. With probability at least {$1-n^{-2c}$}, \Cref{lem:comcon}, \Cref{lem:without-s}, and \Cref{lem:without-our-leader} hold. Throughout this proof, assume these lemmas hold.

The first property in the lemma is immediate by \Cref{lem:without-our-leader}.

Then, for the second property, recall \echothree{} shown in \Cref{fig:alg-echothree}.
Since all nodes in $G$ call \comcon{} with input $com\_ready=0$, by \Cref{lem:comcon}, all nodes in $G$ have the same $com\_redo'$.
Therefore, by \cref{lem:without-s}, after taking the majority value among $RET$ messages, all nodes in $G'$ have the same $redo$.

\begin{claiminline} \label{clm:without-naming-1}
For the \echoone{} algorithm shown in \Cref{fig:alg-echoone}, the following properties hold with probability at least $1-n^{-8c}$:
\begin{itemize}[topsep=1ex,itemsep=0ex]
	\item $G' \subseteq \bigcap_{v \in G} List_v$;
	\item if $lead^* \in G$, then $\bigcup_{v \in G} List_v \subseteq List_{lead^*}'$.
\end{itemize}
\end{claiminline}

\begin{proof}
Fix some $u \in G'$.
Since $u$ sends $ID$ messages to all nodes in $G$, we know $\id(u) \in \bigcap_{w \in G}PreList_w$ and $u$ holds an $ANS$ massage from every node in $G$.
Fix some $w \in G$.
When $w$ samples $C \log n$ random identities from $S_w'$, by a Chernoff bound, with probability at least $1-n^{-10c}$, there are more than $\frac{1.9}{3}C\log n$ identities in $G$.
Then, after $w$ sends $ASK$ messages to $u$, node $w$ receives at least $(C\log n)/2$ $VALID$ messages, implying $\id(u) \in List_w$.
Take a union bound over at most $n^2$ choices of $u$ and $w$, with probability at least $1-n^{-9c}$, we have $G' \subseteq \bigcap_{v \in G} List_v$.

During \bounce{}, by a Chernoff bound, with probability at least $1-n^{-10c}$, every node in $G$ sends less than $(1+\epsilon)C\log n$ identities to any node for forwarding.
Moreover, when a node in $G$ sends an identity to a correct node, with probability $(C\log n)/com\_all$ that correct node forwards the message to $lead^*$.

Now fix some $x \in G$ and some $\id(y) \in List_x$.
For the sake of contradiction, assume there are less than $\frac{0.3}{3}com\_all$ correct nodes in $G$ with $\id(y) \in PreList$.
Together with the Byzantine nodes in $\bigcup_{v \in G'}S_v'$, there are less than $\frac{1.3}{3}com\_all$ nodes that will reply \highlight{$VALID$} messages when queried about $\id(y)$.
Hence, when $x$ samples $C \log n$ random identities from $S_x'$ and sends $ASK$ messages about $\id(y)$, by a Chernoff bound, with probability at least $1-n^{-10c}$, less than $\frac{1.4}{3}C\log n$ nodes would reply \highlight{$VALID$} messages, implying $\id(y) \notin List_x$.
That is, with probability at least $1-n^{-10c}$, there are at least $\frac{0.3}{3}com\_all$ correct nodes in $G$ with $\id(y) \in PreList$.

Assume indeed there are at least $\frac{0.3}{3}com\_all$ correct nodes in $G$ with $\id(y) \in PreList$.
Then there are at least $\frac{0.3}{3}com\_all \cdot C \log n$ messages containing $\id(y)$ sent from nodes in $G$ during \bounce{}, and these messages are forwarded to $lead^*$ by correct nodes with probability $(C \log n)/com\_all$ independently.
\highlightblue{Therefore, by a Chernoff bound, with probability at least $1-n^{-10c}$, at least $\frac{0.2}{3}(C\log n)^2$ messages containing $\id(y)$ are echoed to $lead^*$ during \bounce.}
Take a union bound over at most $n^2$ choices of $x$ and $\id(y)$, we conclude that if $lead^* \in G$, then $\bigcup_{v \in G} List_v \subseteq List_{lead^*}'$.
\end{proof}

\begin{claiminline} \label{clm:without-naming-2}
For the \echothree{} algorithm shown in \Cref{fig:alg-echothree}, with probability at least $1-n^{-7c}$, if $lead^* \in G$, then $redo_v=0$ for any $v \in G'$.
\end{claiminline}

\begin{proof}
Assume $lead^* \in G$, then by \Cref{clm:without-naming-1}, with probability at least $1-n^{-8c}$, $\bigcup_{v \in G} List_v \subseteq List_{lead^*}'$.
Recall \echotwo{} shown in \Cref{fig:alg-echotwo}, all identities in $\bigcup_{v \in G} List_v$ are assigned with new identities in $NewID$ messages, and this renaming is both strong and order preserving.
During \bounce{}, by a Chernoff bound, with probability at least $1-n^{-10c}$,
each new identity is sent to at least one correct node, and leader node $lead^*$ sends less than $(1+\epsilon)C\log n$ messages to any single node.
Hence, for every node in $G$, it receives a new identity for every node in its $List$.
So all nodes in $G$ have $com\_redo=0$, and all correct nodes receive their new identities via $ECHO1$ messages.

During the execution of \echothree{}, all correct nodes send their new identities to nodes in $G$ by $ECHO2$ messages.
Since $lead^* \in G$, all $ECHO2$ messages form a strong and order preserving renaming, so all nodes in $G$ still have $com\_redo=0$.
By \Cref{lem:comcon}, all nodes in $G$ have $com\_redo'=0$.
Then, by \Cref{lem:without-s}, since $|G| > com\_g > \frac{2}{3}com\_all$, all correct nodes eventually have $redo=0$.
\end{proof}

At this point, we note that the third property in the lemma is immediate by \Cref{clm:without-naming-2}.

\begin{claiminline} \label{clm:without-naming-3}
For the \echothree{} algorithm shown in \Cref{fig:alg-echothree}, the following properties hold with probability at least $1-n^{-2c}$:
\begin{itemize}[topsep=1ex,itemsep=0ex]
	\item if there exists $u \in G'$ with $Msg_u = \bot$, then $redo_v=1$ for every $v \in G'$;
	\item if $\nid(\cdot)$ does not form a strong and order preserving renaming among all correct nodes, then then $redo_v=1$ for every $v \in G'$.
\end{itemize}
\end{claiminline}

\begin{proof}
By \Cref{clm:without-naming-1}, with probability at least $1-n^{-8c}$, $G' \subseteq \bigcap_{v \in G} List_v$. Assume this is indeed the case.
If there exists $u \in G'$ with $Msg_u = \bot$, then all nodes in $G$ have $Map[\id(u)]=\bot$ and $com\_redo=1$ during \echotwo{}. So all nodes in $G$ have $com\_redo=1$ during \echothree{}.
By \Cref{lem:comcon}, all nodes in $G$ have $com\_redo'=1$.
Then, by \Cref{lem:without-s}, since $|G| > com\_g > \frac{2}{3}com\_all$, all correct nodes have $redo=1$.

Now consider the case where all correct nodes have $Msg \neq \bot$.
By \Cref{lem:without-s}, $G \subseteq \bigcap_{v \in G'} S_v'$, so every correct node sends its new identity to all nodes in $G$.
If $\nid(\cdot)$ does not form a strong and order preserving renaming among all correct nodes, there are three possible reasons.
(1) $\nid(\cdot)$ violates identity uniqueness, so there exists $u_1 \neq u_2$ with $\nid(u_1)=\nid(u_2)$, thus all nodes in $G$ have $Name[u_1]=Name[u_2]$.
(2) $\nid(\cdot)$ violates identity validity, so there exists $u_3$ with $\nid(u_3)\notin [n]$, thus all nodes in $G$ have $Name[u_3]\notin[n]$.
(3) $\nid(\cdot)$ violates order preserving, so there exists $\id(u_4)<\id(u_5)$ with $\nid(u_4) > \nid(u_5)$, thus all nodes in $G$ have $Name[u_4] > Name[u_5]$.
In either of these three cases, all nodes in $G$ have $com\_redo=1$.
By \Cref{lem:comcon}, all nodes in $G$ have $com\_redo'=1$.
Then, by \Cref{lem:without-s}, since $|G| > com\_g > \frac{2}{3}com\_all$, all correct nodes have $redo=1$.
\end{proof}

We conclude the proof by noting that the last property in the lemma is immediate by \Cref{clm:without-naming-3}.
\end{proof}

\subsection{Complexity analysis}

The following lemma summarizes the time and the message complexity of our entire algorithm.
The committee election process (see \Cref{fig:alg-nsrb-comelect}) and the renaming process (see \Cref{fig:alg-oldid}) together take \highlight[cyan]{$O(\text{poly-log}(n))$} rounds and {$\tilde O(n+\min\{nf, T \})$} messages, where $f$ is the actual number of Byzantine nodes, and $T$ is the number of messages Byzantine nodes sent.

\begin{lemma}\label{lem:without-complex}
For the \committee{} algorithm shown in \Cref{fig:alg-nsrb-comelect}, and the \oldid{} algorithm in shown \Cref{fig:alg-oldid}, the following properties hold with probability at least $1-n^{-c}$:
\begin{itemize}[topsep=1ex,itemsep=0ex]
	\item The \committee{} algorithm has time complexity $O(\text{poly-log}(n))$ and message complexity $O((n+\min\{nf,T\})\cdot\text{poly-log}(n))$;
	\item \committee{} and \oldid{} together have time complexity $O(\text{poly-log}(n))$ and message complexity $O((n+\min\{nf,T\})\cdot\text{poly-log}(n))$.
\end{itemize}
\end{lemma}

\begin{proof}
We start by analyzing the complexity of the committee election process.

\highlightblue{Note that by the analysis of \Cref{lem:without-s}, with probability at least $1-n^{-4c}$, within each iteration of \committee, we have $G \subseteq \cap_{v\in G'}S_v'$ and $|G|=O(com\_all)$.}

\begin{claiminline}\label{clm:without-complex-comcon}
For \comcon{} shown in \Cref{fig:alg-comcon}, \highlight{with probability at least $1-n^{-10c}$}, it has time complexity $O(\text{poly-log}(n))$ and message complexity $O((com\_all+\min \{ f, T/n \})^2 \cdot \text{poly-log}(n))$.
\end{claiminline}

\begin{proof}
Let $Q=|\cup_{v \in G'} S_v'|$.
There are $O(Q)$ nodes executing \comcon{}.
During each phase, propagating $INIT$, $ECHO$, and $LEAD$ messages together incur time complexity $O(1)$ and message complexity $O(Q^2)$.
By \Cref{lem:without-cite-leader}, {\largelead{} incurs time complexity $O(\text{poly-log}(n))$ and message complexity $O(Q^2\cdot\text{poly-log}(n))$}.
\comcon{} has $O(\log n)$ phases, hence each invocation in total incurs time complexity $O(\text{poly-log}(n))$ and message complexity $O(Q^2\cdot\text{poly-log}(n))$.
Recall that during \clst{}, with probability at least $1-n^{-10c}$, for any $\id(u) \in \cup_{v \in G'} S_v'$, there are $\Theta(n)$ correct nodes with $\id(u) \in S$.
Therefore, $Q=O( com\_all+ \min \{ f, T_1/n \})$, where $T_1$ is the number of $ELECT$ messages Byzantine nodes sent during \committee{}.
As a result, the message complexity of \comcon{} is $O((com\_all+ \min \{ f, T/n \})^2 \cdot \text{poly-log}(n))$.
\end{proof}

\begin{claiminline}\label{clm:without-complex-redo}
For \redo{} shown in \Cref{fig:alg-redo}, it has time complexity $O(1)$ and message complexity $O(com\_all \cdot n + \min\{T, nf\})$.
\end{claiminline}

\begin{proof}
There are $n$ nodes executing \redo{}.
Propagating $COM$ messages incurs time complexity $O(1)$ and message complexity $O(n \cdot com\_all)$.

Executing \sample{} incurs time complexity $O(1)$.
There are at most $n$ nodes with $|S'| \ge com\_all$ each sending $O(\log n)$ messages during \sample{}, while nodes with $|S_v'|<com\_all$ send $O(n \log n+\min\{T, nf\})$ messages in total. Note that the addition $O(\min\{T, nf\})$ is because that correct nodes send $YES$ messages on receiving Byzantine $ASK$ messages during \sample{}, and each Byzantine node sends at most $n$ $ASK$ messages.
\end{proof}

\begin{claiminline}\label{clm:without-complex-clst}
For \clst{} shown in \Cref{fig:alg-clst}, with probability at least \highlight{$1-n^{-9c}$}, it has time complexity $O(\text{poly-log}(n))$ and message complexity $O((com\_all^2 + com\_all\cdot n + \min\{nf,T\})\cdot\text{poly-log}(n))$.
\end{claiminline}

\begin{proof}
Executing \sample{} incurs time complexity $O(\log n)$ and message complexity $O( n \cdot com\_all \cdot \log n + \min\{ T_1, nf \}\cdot\log n +  \highlight{\min\{ T_2, nf\cdot\log{n} \}})$, where $T_1$ is the number of Byzantine $ELECT$ messages during \committee{}, and \highlight{$T_2$ is the number of Byzantine $ASK$ messages during \sample{}}.

By a Chernoff bound, with probability at least $1-n^{-10c}$, sending $SERVE$ messages incur time complexity $O(\log n)$ and message complexity $O( n \log n \cdot com\_all + \min\{ T_1, nf \}\cdot \log n)$, where $T_1$ is the number of Byzantine $ELECT$ messages during \committee{}.

Sending $LONG$ messages incur time complexity $O(1)$ and message complexity $O(com\_all \cdot n+\min\{ T_1, nf \})$, where $T_1$ is the number of Byzantine $ELECT$ messages during \committee{}.

Recall \Cref{clm:without-complex-comcon} and \Cref{clm:without-complex-redo}, the total time complexity is $O(\text{poly-log}(n))$, and message complexity is $O((com\_all^2 + com\_all\cdot n + \min\{nf,T\})\cdot\text{poly-log}(n))$.
\end{proof}

\begin{claiminline}\label{clm:without-complex-committee}
For \committee{} shown in \Cref{fig:alg-nsrb-comelect}, with probability at least \highlightblue{$1-n^{-3c}$}, it has time complexity $O(\text{poly-log}(n))$ and message complexity $O((n + \min\{nf,T\})\cdot\text{poly-log}(n))$.
\end{claiminline}

\begin{proof}
\committee{} has $O(\log n)$ phases.
In each phase, sending $ELECT$ messages incurs time complexity $O(1)$ and message complexity $O(n \cdot com\_all)$.
Let $com\_all^*$ be the value of $com\_all$ in the last phase, hence $com\_all^*=O(\min\{f, T/n\}+\log n)$ by \Cref{lem:without-s}.
By \Cref{clm:without-complex-clst}, with probability at least \highlight{$1-n^{-8c}$}, in every phase, \clst{} incurs time complexity $O(\text{poly-log}(n))$ and message complexity
$O(((com\_all^*)^2 + com\_all^*\cdot n + \min\{nf,T\})\cdot\text{poly-log}(n))$.
Therefore, the total time complexity and message complexity of \committee{} across all phases are $O(\text{poly-log}(n))$ and $O((n + \min\{ fn, T \})\cdot\text{poly-log}(n))$, respectively.
\end{proof}

\smallskip Next, we shift our attention to the renaming process. Let $com\_all^*=O(\min\{f, T/n\}+\log n)$ be the value of $com\_all$ after committee election.

\begin{claiminline}\label{clm:without-complex-lead}
For \lead{} shown in \Cref{fig:alg-lead}, it has time complexity $O(\text{poly-log}(n))$ and message complexity $O((com\_all^*)^{2}\cdot\text{poly-log}(n)+n \cdot com\_all^*)$.
\end{claiminline}

\begin{proof}
By \Cref{lem:without-cite-leader}, {\largelead{} incurs time complexity $O(\text{poly-log}(n))$ and message complexity $O((com\_all^*)^{2}\cdot \text{poly-log}(n))$}.
Propagating $INIT$ and $ECHO$ messages incurs time complexity $O(1)$ and message complexity $O((com\_all^*)^{2})$.
By \Cref{clm:without-complex-comcon}, the execution of \comcon{} incurs time complexity $O(\text{poly-log}(n))$ and message complexity $O((com\_all^*)^{2}\cdot\text{poly-log}(n))$.
Propagating $LEAD$ messages incurs time complexity $O(1)$ and message complexity $O(n \cdot com\_all^*)$.
In total, an execution of \lead{} incurs time complexity $O(\text{poly-log}(n))$ and message complexity $O((com\_all^*)^{2}\cdot\text{poly-log}(n)+n \cdot com\_all^*)$.
\end{proof}

\begin{claiminline}\label{clm:without-complex-one}
For \echoone{} shown in \Cref{fig:alg-echoone}, with probability at least $1-n^{-10c}$, it has time complexity $O(\log n)$ and message complexity $O(com\_all^*\cdot n\log n)$.
\end{claiminline}

\begin{proof}
Propagating $ID$ and $ANS$ messages incur time complexity $O(1)$ and message complexity $O(n \cdot com\_all^*)$.
Propagating $ASK$ and $VALID$ messages incur time complexity $O(\log n)$ and message complexity $O(com\_all^*\cdot n \log n)$.
By a Chernoff bound, with probability at least $1-n^{-10c}$, \bounce{} incurs time complexity $O(\log n)$ and message complexity $O(com\_all^*\cdot n \log n)$.
\end{proof}

\begin{claiminline}\label{clm:without-complex-two}
For \echotwo{} shown in \Cref{fig:alg-echotwo}, with probability at least $1-n^{-10c}$, it has time complexity $O(\log n)$ and message complexity $O(com\_all^*\cdot n \log n)$.
\end{claiminline}

\begin{proof}
By a Chernoff bound, we know that with probability at least $1-n^{-10c}$, \bounce{} incurs time complexity $O(\log n)$ and message complexity $O(com\_all^*\cdot n \log n)$.
Propagating $ECHO1$ messages incurs time complexity $O(1)$ and message complexity $O(com\_all^*\cdot n)$.
\end{proof}

\begin{claiminline}\label{clm:without-complex-three}
For \echothree{} shown in \Cref{fig:alg-echothree}, it has time complexity $O(\text{poly-log}(n))$ and message complexity $O( (com\_all^*)^2\cdot\text{poly-log}(n)+com\_all^*\cdot n)$.
\end{claiminline}

\begin{proof}
Propagating $ECHO2$ messages incurs time complexity $O(1)$ and message complexity $O(com\_all^*\cdot n)$.
By \Cref{clm:without-complex-comcon}, \comcon{} incurs time complexity $O(\text{poly-log}(n))$ and message complexity $O( (com\_all^*)^2\cdot\text{poly-log}(n))$.
Propagating $RET$ messages incurs time complexity $O(1)$ and message complexity $O(com\_all^*\cdot n)$.
\end{proof}

Now we can conclude the cost for the renaming process, i.e., the cost for running \oldid{}.

\begin{claiminline}\label{clm:without-complex-oldid}
For \oldid{} shown in \Cref{fig:alg-oldid}, with probability at least $1-n^{-8c}$, it has time complexity $O(\text{poly-log}(n))$ and message complexity $O( (\min\{ f, T/n \})^2\cdot\text{poly-log}(n) + \min\{ f, T/n \}\cdot n\log^2 n + n\log^3n)$.
\end{claiminline}

\begin{proof}
By \Cref{clm:without-complex-lead}, \Cref{clm:without-complex-one}, \Cref{clm:without-complex-two}, \Cref{clm:without-complex-three}, with probability at least $1-n^{-9c}$, every phase of \oldid{} has time complexity $O(\text{poly-log}(n))$ and message complexity $O( (com\_all^*)^2\cdot \text{poly-log}(n)+com\_all^*\cdot n\log n)$.
There are $O(\log n)$ phases, so with probability at least $1-n^{-8c}$, in total, the time complexity is $O(\text{poly-log}(n))$ and the message complexity is $O( (com\_all^*)^2\cdot \text{poly-log}(n)+com\_all^*\cdot n\log^2 n)$.
Since $com\_all^*=O(\min \{ f, T/n \}+\log n)$, the total message complexity is $O( (\min\{ f, T/n \})^2\cdot \text{poly-log}(n) + \min\{ f, T/n \}\cdot n\log^2 n + n\log^3n)$.
\end{proof}

\smallskip Combing \Cref{clm:without-complex-committee} and \Cref{clm:without-complex-oldid} completes the proof of the lemma.
\end{proof}

\subsection{Proof of main theorem}

In this part, we provide the proofs for \Cref{thm:committee}, \Cref{cor:consensus}, and \Cref{thm:without-renaming}.

\begin{proof}[Proof of \Cref{thm:committee}]
The desired committee can be obtained by running \committee{} shown in \Cref{fig:alg-nsrb-comelect}.
The properties are satisfied with high probability in $n$ by \Cref{lem:without-s}.
By \Cref{lem:without-complex}, with high probability in $n$, it has time complexity \highlight[cyan]{$O(\text{poly-log}(n))$} and message complexity $\tilde{O}(n+\min\{nf,T\})$.
\end{proof}

\begin{proof}[Proof of \Cref{cor:consensus}]
By \Cref{thm:committee}, with high probability in $n$, all correct nodes obtain a committee list satisfying the demand of \Cref{lem:comcon}.
Then all nodes in the committee run \comcon{} in \Cref{fig:alg-comcon} with $com\_ready=0$ and the input of binary consensus.
By \Cref{lem:comcon}, the output of the committee solves binary consensus.
Then the committee distributes the result to all other nodes.
By \Cref{lem:without-complex}, with high probability in $n$, this entire consensus algorithm has time complexity \highlight[cyan]{$O(\text{poly-log}(n))$} and message complexity $\tilde O(n+\min\{nf,T\})$.
\end{proof}

\begin{proof}[Proof of \Cref{thm:without-renaming}]
By \Cref{lem:without-our-leader}, with probability at least $1-n^{-2c}$, the condition $lead^* \in G$ in \Cref{lem:without-naming} is satisfied at least once among all phases.
Therefore, by \Cref{lem:without-naming}, with probability at least $1-n^{-c}$, in that phase, all correct nodes halt with $\nid(\cdot)$ forming a strong and order preserving renaming among all correct nodes.
By \Cref{lem:without-complex}, with probability at least $1-n^{-c}$, the entire algorithm has time complexity \highlight[cyan]{$O(\text{poly-log}(n))$} and message complexity $\tilde{O}(n+\min\{nf,T\})$.
\end{proof}

\appendix

\section*{Appendix}

\section{Vector Consensus Algorithm}\label{app-sec:vector-con}

We summarize the guarantees of our vector consensus algorithm in the following lemma.

\begin{lemma}\label{lem:consensus-vector}
Assume all nodes' identities are in $[N]$.
Let $G$ denote the set of correct committee nodes. For each node $v\in G$, let $S_v$ denote the committee in its view. Let $\hat{c}$ denote an upper bound of $|\bigcup_{v\in G} S_v|$ and let $\hat{b}$ denote an upper bound on the number of bad nodes in $\bigcup_{v\in G} S_v$. Suppose $\hat{b}<1/3 \hat{c}$ and $|G|\geq 2\hat{b}+1$.
Assume each $v\in G$ receives an input bit vector $in_v\in \{0,1\}^{N}$ with the guarantee that $\sum_{i\in [N]}(\bigvee_{v\in G} in_v[i])\leq d$. There exists an $O(\hat{c}^5d)$-round algorithm with message complexity $O(\hat{c}^7d)$ and bit complexity $O(\hat{c}^7d\log N)$ that produces a bit vector $out_v\in \{0,1\}^{N}$ satisfying:
\begin{itemize}[topsep=1ex,itemsep=0ex]
	\item validity: for any $i\in [N]$, $out_v[i]=in_u[i]$ for some $u\in G$;
	\item agreement: for any $u\in G$ and any $i\in [N]$, $out_u[i]=out_v[i]$.
\end{itemize}
\end{lemma}

\begin{figure}[t!]
\hrule
\vspace{1ex}
\highlight{$\texttt{VectorConsensus}(S_v, \hat{c}, \hat{c}-\hat{b}, \hat{b})$ executed at node $v$.}
\vspace{1ex}
\hrule
\begin{footnotesize}
\begin{algorithmic}[1]
\State $Broad_v\gets\{\text{an empty map of ($\id\rightarrow$ boolean value) }\}$, $Accept_v\gets\{\text{an empty map of ($\id$ $\rightarrow$ $\id$ set) }\}$.
\State $EchoCnt_v\gets\{\text{an empty map of (($\id,\id$) $\rightarrow$ $\id$ set) }\}$,  $Buf_v\gets \emptyset$, $out_v\gets 0^{N}$.

\For {(each phase $j$ \textbf{from} 1 \textbf{to} $\hat{b}+1$)}
	\LineComment{Part I: Broadcast $\langle BC\rangle$ messages. This part lasts for $d$ rounds.}
	\If {($j==1$)}
		\For{($k$ \textbf{from} $1$ \textbf{to} $N$)}
			\If{($in[k]==1$)}
				\State Send $\langle BC, \id(v), k\rangle$ to all nodes in $S_v$.
				\State $Broad_v[k]\gets true$.
			\EndIf
		\EndFor
	\Else 
		\For{($k=1$ \textbf{to} $N$)}
			\If{($Broad_v[k] == false$  \textbf{and} $|Accept_v[k]|\geq \hat{b}+j-1$)}
				\State Send $\langle BC, \id(v), k\rangle$ to all nodes in $S_v$.
				\State $Broad_v[k]\gets true$.
			\EndIf
		\EndFor
	\EndIf

	\LineComment{Part II: Broadcast $\langle ECHO\rangle$ messages. This part lasts for $j\cdot \hat{c}^3d$ rounds.}
	\For{(each $\langle BC, k_2, k_1\rangle$ received from $u\in S_v$ with $\id(u)==k_2$)}
		\If{($\id(v)\notin EchoCnt_v[k_2][k_1]$)}
			\State Send  $\langle ECHO, \id(v), k_2, k_1\rangle$ to all nodes in $S_v$.
			\State $ EchoCnt_v[k_2][k_1] \gets EchoCnt_v[k_2][k_1] \cup \{\id(v)\}$.
		\EndIf
	\EndFor
	\For{(each message $m\in Buf_v$)}
		Send $m$ to all nodes in $S_v$.
	\EndFor
	\State $Buf_v\gets \emptyset$.

	\LineComment{Part III: Update buffer. No communication required. }
	\For{(each $\langle ECHO, k_3, k_2, k_1\rangle$ received  from $u\in C_v$ with $\id(u)==k_3$ such that $\exists w\in C_v$ with $\id(w)==k_2$)}
		\State $EchoCnt_v[k_2][k_1] \gets EchoCnt_v[k_2][k_1] \cup \{k_3\}$.
		\If{($\id(v)\notin EchoCnt_v[k_2][k_1]$ \textbf{and} $|EchoCnt_v[k_2][k_1]|\geq \hat{b}+1$)}\label{line:alg-vector-consensus-add-echo-condition}
			\State $Buf_v\gets Buf_v\cup \{\langle ECHO, \id(v), k_2, k_1\rangle\}$.
		\EndIf
		\If{($|EchoCnt_v[k_2][k_1]|\geq 2\hat{b}+1$)}\label{line:alg-vector-consensus-add-ac-condition}
			\State $Accept_v[k_1]\gets Accept_v[k_1]\cup \{k_2\}$.
		\EndIf
	\EndFor
\EndFor

\For{($k$ \textbf{from} $1$ to \textbf{to} $N$)}
	\If{($|Accept_v[k]|\geq 2\hat{b}+1$)}
		$out_v[k]\gets1$.
	\EndIf
\EndFor

\State \textbf{return} $out_v$.
\end{algorithmic}
\end{footnotesize}
\hrule
\vspace{1ex}
\caption{Pseudocode of the vector consensus algorithm.}\label{fig:alg-vector-consensus}
\vspace{-3ex}
\end{figure}

In essence, this algorithm is a parallel variant of classical binary consensus based on reliable broadcast.
See \Cref{fig:alg-vector-consensus} for the pseudocode of the vector consensus algorithm, which operates over $\hat{b}+1$ phases.
In each phase, the algorithm proceeds through three distinct parts. In Part I, if currently in the first phase, nodes broadcast their non-zero vector entry positions via $BC$ messages; otherwise, in subsequent phases, nodes broadcast positions (again via $BC$ messages) that have received sufficient acceptances but have not been broadcast yet.
In Part II, nodes broadcast echo messages for the received $BC$ messages. Nodes also flush their outgoing message buffers during this phase, ensuring all pending echo messages get delivered.
Part III processes received echo messages to update acceptance counts. When a node observes $\hat{b}+1$ distinct echoes for a $BC$ message, it adds an echo message for this $BC$ message to the buffer (which will be sent during the next phase). Moreover, when $2\hat{b}+1$ echoes are collected for a $BC$ message, this $BC$ message is accepted.
After completing all phases, each node outputs a bit vector $out_v$ where position $k$ is set to 1 if at least $2\hat{b}+1$ accepted $BC$ messages contain this position.

We conclude this section with a proof of \Cref{lem:consensus-vector}.

\begin{proof}[Proof of \Cref{lem:consensus-vector}]
As mentioned above, our vector consensus algorithm implements reliable broadcast internally.
Hence, we begin by proving the properties that reliable broadcast can enforce.

\begin{claiminline}\label{cla:vector-con-reliable-pro}
\begin{enumerate}[nosep,itemsep=1pt]
	\item \label{cla:vector-con-reliable-pro-1} If a correct committee node $v\in G$ broadcasts $\langle BC, \id(v), k\rangle$ at the beginning of some phase $j$, then for any  $u\in G$, it holds that $\id(v)\in Accept_u[k]$ at the end of phase $j$.

	\item \label{cla:vector-con-reliable-pro-2} If a correct committee node $v\in G$ does not broadcast $\langle BC, \id(v), k\rangle$, then for any $u\in G$, it holds that $\id(v)\notin Accept_u[k]$.

	\item \label{cla:vector-con-reliable-pro-3} If a correct committee node $v\in G$ has $k_2\in Accept_v[k_1]$ at the end of phase $j$, then for any correct committee node $u\in G$, it holds that $k_2\in Accept_u[k_1]$ at the end of phase $j+1$.

	\item \label{cla:vector-con-reliable-pro-4} Let $D$ denote $\{i: i\in [N], \bigvee_{v\in G} in_v[i]=1\}$. For any correct committee node $v\in G$, for any $1\leq j\leq \hat{b}+1$, node $v$ has at most $d$ messages to send in Part I of phase $j$ each of the form $\langle BC, \id(v), k\rangle$ with $k\in D$, and it has at most $j\hat{c}^3d$ messages to broadcast in Part II of phase $j$.
\end{enumerate}
\end{claiminline}

\begin{proof}
First, we prove \Cref{cla:vector-con-reliable-pro-2} of \Cref{cla:vector-con-reliable-pro}. If a correct committee $v$ does not broadcast $\langle BC, \id(v), k\rangle$, then no correct committee node $u$ will broadcast the echo message $\langle ECHO, \id(u), \id(v), k\rangle$ in Part II.
Since there are at most $\hat{b}$ Byzantine committee nodes, for any correct committee node $u\in G$, we know $|Echo_u[\id(v)][k]|$ is at most $\hat{b}$. Thus, node $u$ never adds $\langle ECHO, \id(u), \id(v), k\rangle$ to $Buf_u$, due to Line~\ref{line:alg-vector-consensus-add-echo-condition} of \Cref{fig:alg-vector-consensus}. Moreover, node $u$ never adds $\id(v)$ to $Accept_u[k]$, due to Line~\ref{line:alg-vector-consensus-add-ac-condition} of \Cref{fig:alg-vector-consensus}.

\smallskip
Then, we prove \Cref{cla:vector-con-reliable-pro-4} of \Cref{cla:vector-con-reliable-pro}.
Intuitively, this item states that all messages will be send in time. By lemma assumption, we have $|D|\leq d$. More specifically, we show that, in any phase $1\leq j \leq \hat{b}+1$, for any correct committee node $v\in G$: (1) in Part I, it has at most $d$ messages to broadcast, each of which is the form $\langle BC, \id(v), k\rangle$ with $k\in D$; (2) in Part II, it has at most $j\hat{c}^3d$ messages to broadcast; and (3) in Part III, it has $|Buf_v|\leq j\hat{c}^3d$.

We prove above properties by induction. First, consider the base case $j=1$.

Consider Part I. By lemma assumption, we have $in_v[k]=0$ for any $k\in [N]\setminus D$. Thus, node $v$ has at most $d$ messages of the form $\langle BC, \id(v), k\rangle$ with $k\in D$ to broadcast in Part I in the first phase.

Consider Part II. In the first phase, $v$ receives at most $\hat{c}d$ valid BC messages in Part I. Thus, it broadcasts at most $\hat{c}d$ echos for these BC messages in Part II. Moreover, $Buf_v$ is empty at the beginning of the first phase. Therefore, node $v$ has at most $\hat{c}d$ messages to broadcast in Part II of the first phase.

Consider Part III. According to \Cref{fig:alg-vector-consensus}, $v$ adds an echo message of the form $\langle ECHO, \id(v), k_2, k_1\rangle$ to $Buf_v$ if and only if it receives valid echo messages of the form $\langle ECHO, *, k_2, k_1\rangle$ from at least $\hat{b}+1$ different committee nodes. We bound the size of $Buf_v$ by counting the possible values of $k_2$ and $k_1$. The message $\langle ECHO, *, k_2, k_1\rangle$ is valid only if there exists a node $w\in C_v$ with $\id(w)=k_2$. Thus, $k_2$ has at most $\hat{c}$ different values.
Moreover, among all the nodes that send $\langle ECHO, *, k_2, k_1\rangle$ to $v$, there must exist a correct committee node.
Note that if no message of the form $\langle BC, *, k_1\rangle$ is ever sent to a correct committee, then no correct committee node $u\in G$ will send $\langle ECHO, \id(u), k_2, k_1\rangle$ in Part I or add $\langle ECHO, \id(u), k_2, k_1\rangle$ to $Buf_u$ in Part III. Thus, for each echo message of the form $\langle ECHO, \id(v), k_2, k_1\rangle$ that adds to $Buf_v$, there must exist some correct committee that has received $\langle BC, *, k_1\rangle$. Recall that at most $\hat{c}^2d$ valid messages of the form $\langle BC, *, k_1\rangle$ are received by any correct committee in Part I in the first phase, we know there could be at most $\hat{c}^2d$ different values for $k_1$. Thus, the size of $Buf_v$ is at most $\hat{c}^3d$.

Suppose that the properties hold up to phase $j$, now we prove that they hold for phase $j+1$.

Consider Part I.
By algorithm description, node $v$ considers the message $\langle ECHO, *, k_2, k_1\rangle$ as valid only if there exist a node $w\in C_v$ with $\id(w)=k_2$. Thus, for any $k_2\in [N]\setminus \{\id(u):u\in C_v\}$, any $k_1\in [N]$, we have $EchoCnt_v[k_2][k_1]=\emptyset$ and $k_2\notin Accept_v[k_1]$.
On the other hand, recall that by the induction hypothesis, no correct committee broadcasts messages of the form $\langle BC, *, k_1\rangle$ with $k_1\in [N]\setminus D$ in Part I of previous phases.
Hence, by \Cref{cla:vector-con-reliable-pro-2} of \Cref{cla:vector-con-reliable-pro}, for any $u\in G$, any $k_1\in [N]\setminus D$, we have $\id(u)\notin Accept_v[k_1]$.
Therefore, for any $k_1\in [N]\setminus D$, we have $|Accept_v[k_1]|\leq \hat{b}$.
As a result, $v$ does not broadcast messages of the form $\langle BC, \id(v), k\rangle$ with $k\in [N]\setminus D$ in Part I of phase $j+1$.
This implies, $v$ has at most $d$ messages of the form $\langle BC, \id(v), k\rangle$ with $k\in D$ to broadcast in Part I of phase $j+1$.

Consider Part II.
By above analysis for Part I, in phase $j+1$, each correct committee node receives at most $\hat{c}d$ valid BC messages in Part I. Thus, any correct committee $v$ sends at most $\hat{c}d$ echos to response these messages in Part II. By the induction hypothesis, the size of $Buf_v$ is at most $j\hat{c}^3d$ at the end of phase $j$. Thus, node $v$ sends at most $(j+1)\hat{c}^3d$ messages in Part II of phase $j+1$.

Consider Part III.
Similar to the analysis of the base case, we bound the size of $Buf_v$ by counting the possible value of $k_1$ and $k_2$. Since $k_2\in C_v$, $k_2$ still has at most $\hat{c}$ different values. Recall that at most $\hat{c}^2d$ valid messages of the form $\langle BC, *, k_1\rangle$ are received by any correct committee in Part I in each phase, there could be at most $(j+1)\hat{c}^2d$ different values for $k_1$. Therefore, the size of $Buf_v$ is at most $(j+1)\hat{c}^3 d$.

By now, \Cref{cla:vector-con-reliable-pro-4} of \Cref{cla:vector-con-reliable-pro} is proved.

\smallskip
Next, we prove \Cref{cla:vector-con-reliable-pro-1} of \Cref{cla:vector-con-reliable-pro}. If a correct committee node $v$ broadcasts $\langle BC, \id(v), k\rangle$ at the beginning of phase $j$, then every correct committee node $u\in G$  will broadcast the echo message $\langle ECHO, \id(u), \id(v), k\rangle$ in phase $j$.
Recall that there are at least $2\hat{b}+1$ correct committee nodes.
Therefore, for any correct committee $u\in G$, it will add at least $2\hat{b}+1$ different $\id$s of correct committees to $EchoCnt_u[\id(v)][k]$ (see Line~\ref{line:alg-vector-consensus-add-echo-condition} of \Cref{fig:alg-vector-consensus}), hence add $\id(v)$ to $Accept_u[k]$ (see Line~\ref{line:alg-vector-consensus-add-ac-condition} of \Cref{fig:alg-vector-consensus}) by the end of phase $j$.

\smallskip
Lastly, we prove \Cref{cla:vector-con-reliable-pro-3} of \Cref{cla:vector-con-reliable-pro}. Suppose $v$ adds $k_2$ to $Accept_v[k_1]$ in phase $j'\leq j$. This implies $v$ has received valid echo messages of the form $\langle ECHO, *, k_2, k_1\rangle$ from at least $2\hat{b}+1$ nodes in $C_v$. Among these nodes, there are at least $\hat{b}+1$ correct committee nodes. Since correct committee nodes perform broadcast honestly, every correct committee receives at least $\hat{b}+1$ different valid echo messages of the form $\langle ECHO, *, k_2, k_1\rangle$ by phase $j'$. This guarantees that every correct committee node $u\in G$ sends $\langle ECHO, \id(u), k_2, k_1\rangle$ by phase $j'+1$. Thus, every correct committee node will receive at least $2\hat{b}+1$ valid echo messages of the form $\langle ECHO, *, k_2, k_1\rangle$ by the end of phase $j'+1$. As a result, for any correct committee node $u\in G$, it holds that $k_2\in Accept_u[k_1]$ at the end of phase $j'+1\leq j+1$.
\end{proof}

With above claim, we resume lemma proof.

\smallskip
Consider the complexity of the algorithm. There are $\hat{b}+1$ phases, and each phase $1\leq j\leq \hat{b}+1$ contains $\hat{c}+j\hat{c}^3d$ rounds.
Thus, runtime of the algorithm is $O(\hat{c}^5d)$ rounds.
In each round, each correct committee node sends at most $|C_v|\leq \hat{c}$ messages. Thus, the message complexity is $O(\hat{c}^7d)$. Note that each message consists of $O(\log N)$ bits. Thus, the bit complexity is $O(\hat{c}^7d\log N)$.

\smallskip
Next, consider the validity property.
Fix an $i\in [N]$, there are two cases. First, if every $v\in G$ has $in_v[i] =1$, then every $v\in G$ broadcasts $\langle BC, \id(v), i\rangle$ in the first phase. By \Cref{cla:vector-con-reliable-pro-1} of \Cref{cla:vector-con-reliable-pro}, for any two correct committee nodes $u,v\in G$, we have $\id(u)\in Accept_v[i]$. Thus, for any $v\in G$, we have $|Accept_v|\geq 2\hat{b}+1$ and $out_v[i]$ will be set to $1$ at the end of the algorithm.
On the other hand, if  $in_v[i] =0$ for every $v\in G$, then $i\notin D$. According to the algorithm, a correct node $v\in G$ considers the message $\langle ECHO, *, k_2, k_1\rangle$ as valid only if there exist a node $w\in C_v$ with $\id(w)=k_2$.
Thus, for any $k_2\in [N]\setminus \{\id(u):u\in C_v\}$, any $k_1\in [N]$, we have $EchoCnt_v[k_2][k_1]=\emptyset$ and $k_2\notin Accept_v[k_1]$.
Moreover, by \Cref{cla:vector-con-reliable-pro-4} of \Cref{cla:vector-con-reliable-pro}, no correct committee will broadcast messages of form $\langle BC, \id(v), k_1\rangle$ with $k_1\in [N]\setminus D$. Hence, for any $v, u\in G$, any $k_1\in [N]\setminus D$, $\id(u)\notin Accept_v[k_1]$.
Therefore, for any $v\in G$, any $k\in [N]\setminus D$, $|Accept_v[k]|\leq \hat{b}$. As a result, $out_v[i]$ is set to $0$ at the end of the algorithm.

\smallskip
We conclude this proof with the agreement property.
Specifically, we show that for any $i\in [N]$, any $v_1\in G$, if $out_{v_1}[i]=1$, then $out_u=1$ for all $u\in G$.
Fix a $v_1\in G$, fix some $i\in [N]$ with $out_{v_1}[i]=1$. Since $out_{v_1}[i]=1$, node $v_1$ must have $|Accept_{v_1}[i]|\geq 2\hat{b}+1$ by the end of phase $\hat{b}+1$. Let $Accept_{v_1}'[i]$ denote the $\id$s of correct committee nodes in $Accept_{v_1}[i]$. Recall that $v_1$ treats the message $\langle ECHO, *, k_2, k_1\rangle$ as valid only if there exists a node $w\in C_{v_1}$ with $\id(w)=k_2$. Thus, for any $i\in [N]$, we have $Accept_{v_1}[i]\subseteq C_{v_1}$. As a result, there are at most $\hat{b}$ Byzantine nodes in $Accept_{v_1}[i]$. This implies $|Accept'_{v_1}[i]|\geq \hat{b}+1$.

If $in_u[i]=1$ for any $u\in Accept'_{v_1}[i]$, then they broadcast $\langle BC, \id(u), i \rangle$ in Part I of phase $1$. By \Cref{cla:vector-con-reliable-pro-1} of \Cref{cla:vector-con-reliable-pro}, for any $u\in Accept'_{v_1}[i]$, every $w\in G$ has $\id(u)\in Accept_w[i]$ at the end of the phase $1$. Thus, $|Accept_w[i]|\geq \hat{b}+1$ at the end of phase $1$. This implies $w$ will broadcast $\langle BC, \id(w), i\rangle$ in phase $2$. By \Cref{cla:vector-con-reliable-pro-1} of \Cref{cla:vector-con-reliable-pro}, for any two correct committees $u, w\in G$, we have $\id(w)\in Accept_{u}[i]$ by the end of phase $2$. Thus, for every $u\in G$, it holds that $|Accept_u[i]|\geq 2\hat{b}+1$ at the end of phase $2$. As a result, for every $u\in G$, we have $out_u[i]=1$ at the end of the algorithm.

On the other hand, suppose that one of the nodes in $Accept'_{v_1}[i]$, say $u_1$ does not have $in_{u_1}[i]=1$. Then, by \Cref{cla:vector-con-reliable-pro-2} of \Cref{cla:vector-con-reliable-pro}, node $u_1$ must have broadcast $\langle BC, \id(u_1), i\rangle$ in Part I of some phase $2\leq j\leq \hat{b}+1$, which means that $u_1$ has $|Accept_{u_1}[i]|\geq \hat{b}+j-1$ at the end of phase $j-1$.
Denote the set $Accept_{u_1}[i]$ at the end of phase $j-1$ as $AC$.
By \Cref{cla:vector-con-reliable-pro-3} of \Cref{cla:vector-con-reliable-pro}, we have $AC\subseteq Accept_u[i]$ for every $u\in G$ at the end of phase $j$.
Since $u_1$ does not broadcast $\langle BC, \id(u_1), i\rangle$ until phase $j$, we have $\id(u_1)\notin AC$ by \Cref{cla:vector-con-reliable-pro-2} of \Cref{cla:vector-con-reliable-pro}.
Since $u_1$ broadcasts $\langle BC, \id(u_1), i\rangle$ in phase $j$, we have $\id(u_1)\in Accept_u[i]$ for every $u\in G$ at the end of phase $j$ by \Cref{cla:vector-con-reliable-pro-1} of \Cref{cla:vector-con-reliable-pro}.

Now, there are two cases. If $j=\hat{b}+1$. Then every $u\in G$ has $|Accept_u[i]|\geq |AC\cup \{\id(u_1)\}|\geq 2\hat{b}+1$ at the end of phase $\hat{b}+1$. Thus, every $u\in G$ will set $out_u[i]$ to $1$ at the end of the algorithm.
Otherwise, if $j<\hat{b}+1$, then every $u\in G$ has $|Accept_u[i]|\geq |AC\cup \{\id(u_1)\}|\geq \hat{b}+j$ at the end of phase $j$. Thus, every node $u\in G$ broadcasts $\langle BC, \id(u), i\rangle$ in phase $j+1$. By \Cref{cla:vector-con-reliable-pro-1} of \Cref{cla:vector-con-reliable-pro}, for any $u, v\in G$, we have $\id(v)\in Accept_u[i]$ at the end of phase $j+1$. As a result, for every $u\in G$, node $u$ has $|Accept_u[i]|\geq 2\hat{b}+1$ at the end of phase $j+1$, hence $out_u[i]=1$ at the end of the algorithm.
\end{proof}

\section{Alternative Analysis for the Committee Election Algorithm}\label{app-sec:committee-elect-alt-analysis}

\begin{corollary}
Assume we use the same model same as \cite{augustine20}.
For the \committee{} algorithm shown in \Cref{fig:alg-nsrb-comelect}, \highlightblue{after its execution}, the following properties hold with probability at least \highlight{$1-n^{-2c}$}:
\begin{itemize}[topsep=1ex,itemsep=0ex]
	\item $G \subseteq \bigcap_{v \in G'} S_v'$;
	\item $|\bigcup_{v \in G'}S_v'|<com\_all$, {$|G| > (1+\epsilon)com\_g >(1+\epsilon)(\frac{2}{3}+\epsilon')com\_all$}, $com\_b = com\_all-com\_g$;
	\item $com\_all=O(\min \{ f\log n, T/n \}+\log n)$;
	\item it has time complexity \highlight[cyan]{$O(\text{poly-log}(n))$} and message complexity $\tilde O(n + \min \{ nf, T \})$.
\end{itemize}
Here, $0\leq f<(1/3-\delta)n$ is the actual number of Byzantine nodes, and $T\geq 0$ is the number of messages Byzantine nodes sent.
\end{corollary}

\begin{proof}
Recall the discussion on model in \Cref{sec:preliminary}, note that \cite{augustine20} assumes ``a Byzantine node's identity is presented to all nodes which receive messages from that node'', whereas we further assume transitive verifiability throughout the main body of the paper. (For instance, assume $u$ receives a message $m_w$ which claims to be sent by $w$, further assume the content of $m_w$ includes another message $m_v$ that $w$ claims to be received from $v$. Then, $u$ can verify whether $m_w$ originates from $w$. However, without transitive verifiability, $u$ but cannot verify whether $m_v$ originates from $v$.) Here, we provide an alternative analysis for our committee election algorithm \committee, showing that it can provide similar guarantees without assuming transitive verifiability.

Throughout the execution of \committee, $SERVE$ messages exchanged during the execution of \clst{} (see \Cref{fig:alg-clst}) is the only place where we utilize transitive verifiability.
{(\largelead{} is from \cite{king06soda}, which is also used in \cite{augustine20}, therefore suits the modified model.)}

Now we discuss the impact of $SERVE$ messages.
Due to the limitation imposed by $Cnt$, for any correct node, each Byzantine node sends at most $(1+\epsilon)C\log n$ $SERVE$ messages to that node.
Note that $SERVE$ messages only affects the constitution of $Long$, while committee lists (i.e., $S$ and $S'$) are constructed based on direct identity exchanges (thereby only containing nodes present in the system).

Based on above discussion, to show the correctness of \committee{} in the modified model, we only need to prove that the prerequisites to call $\comcon{}$ (see \Cref{fig:alg-comcon}) are still satisfied, since this is the only place where we use $Long$.
The proof is identical with that of \Cref{lem:without-largecore}.

To bound the final committee's size in the modified model, the analysis is similar to that of \Cref{lem:without-s}.

\begin{claiminline}[Analogue of \Cref{clm:without-s-1}]\label{clm:without-s-1-ana}
With probability at least $1-n^{-5c}$, during \clst{} (see \Cref{fig:alg-clst}), there exists a constant $\delta'>0$ such that, if there exists some $v \in G$ with $com\_out_v=1$, then $f\ge  \delta'\cdot  com\_all/\log n$ and $T \ge \delta' \cdot com\_all \cdot n$.
\end{claiminline}

\begin{proof}
By \Cref{lem:comcon} and  \Cref{lem:without-largecore}, with probability at least $1-n^{-6c}$, all nodes in $G$ have the same $com\_out$. Moreover, if there are at least $com\_g$ nodes in $G$ with $com\_in=value$, then $com\_out=value$.

Assume there exists a node in $G$ with $com\_out=1$, then there are at least $\epsilon \cdot com\_g$ nodes in $G$ with $com\_in=1$. Consider two scenarios.

Scenario one, there are less than $\frac{\epsilon}{2}n$ nodes in $G'$ with $|Long| \ge com\_all$. Since there are at least $\epsilon \cdot com\_g$ nodes in $G$ with $com\_in=1$, we know each of them receives at least $\frac{\epsilon}{2} n$ $LONG$ messages from Byzantine nodes,
So $f \ge \frac{\epsilon}{2} n \ge \frac{\epsilon}{2(1+2\epsilon)}com\_all$ and $T \ge \frac{\epsilon^2}{2}n \cdot com\_g \ge \frac{\epsilon^2}{3} \cdot com\_all \cdot n$.

Scenario two, there are at least $\frac{\epsilon}{2}n$ nodes in $G'$ with $|Long| \ge com\_all$. Let $Q$ denote these nodes.
By a Chernoff bound, with probability at least $1-n^{-10c}$, every node in $Q$ receives at least $\epsilon pn$ Byzantine identities from $ELECT$ messages or $SERVE$ messages.
If there are at least $\frac{\epsilon}{4}n$ nodes in $Q$ that each receives at least $\frac{\epsilon}{2} pn$ Byzantine identities from $ELECT$ messages, then $f \ge \frac{\epsilon}{2} pn \ge \frac{\epsilon}{2(1+2\epsilon)}com\_all$ and $T \ge \frac{\epsilon}{4}n \cdot \frac{\epsilon}{2} pn \ge \frac{\epsilon^2}{8(1+2\epsilon)} \cdot com\_all \cdot n$, as required.

In the following, assume there are at least $\frac{\epsilon}{4}n$ nodes in $Q$ that each receives at least $\frac{\epsilon}{2} pn$ Byzantine identities from $SERVE$ messages. Let these nodes in $Q$ be node set $Q'$.
Without loss of generality, assume a set $W$ of Byzantine identities are received by every node in $Q'$ from $SERVE$ messages, and these nodes add these Byzantine identities to $Long$, where $|W| \ge \frac{\epsilon}{2}pn$.
We further consider two cases.

First case, there are at least $\frac{\epsilon}{4}pn$ Byzantine identities in $W$ that each appears in less than $n/40$ correct nodes' list $S$.
By a Chernoff bound, with probability at least $1-n^{-10c}$, for any of these Byzantine identities, at most $(C\log n)/20$ correct nodes will send $SERVE$ messages for it.
So, spreading each of these at least $\frac{\epsilon}{4}pn$ Byzantine identities to $Q'$ needs at least $\frac{\epsilon}{4}n \cdot (C\log n)/20$ $SERVE$ messages from Byzantine nodes.
Recall that by algorithm description, each node accepts at most $(1+\epsilon)C\log n$ $SERVE$ messages from a particular node, so we have
$f \ge \frac{\epsilon}{4}pn  / ((1+\epsilon)C\log n) \ge \frac{\epsilon}{4(1+\epsilon)(1+2\epsilon)C} com\_all/\log n $, and
$T \ge  (C\log n)/20 \cdot \frac{\epsilon}{4}n \cdot \frac{\epsilon}{4}pn \ge \frac{\epsilon^2C}{320(1+2\epsilon)}com\_all \cdot n\log n$.

Second case, there are at least $\frac{\epsilon}{4}pn$ Byzantine identities in $W$ that each appears in at least $n/40$ correct nodes' list $S$.
Note that correct nodes build $S$ based on $ELECT$ messages.
Therefore, $f \ge \frac{\epsilon}{4}pn \ge \frac{\epsilon}{4(1+2\epsilon)}com\_all$ and $T \ge \frac{\epsilon}{4}pn \cdot \frac{n}{40} \ge \frac{\epsilon}{160(1+2\epsilon)}com\_all \cdot n$.
\end{proof}

By \Cref{clm:without-s-1-ana}, with probability at least $1-n^{-5c}$, we have $com\_all=O(\min \{ f\log n, T/n \}+\log n)$ for any phase.
Then by the same analysis as in the proof of \Cref{lem:without-s}, $com\_all=O(\min \{ f\log n, T/n \}+\log n)$ is satisfied after \committee{}.

Lastly, the complexity analysis is similar to that of \Cref{lem:without-complex}.
With probability at least \highlight{$1-n^{-3c}$}, during any phase of \committee:
\comcon{} has time complexity \highlight[cyan]{$O(\text{poly-log}(n))$} and message complexity $\tilde O((com\_all+\min \{ f, T/n \})^2)$;
\redo{} has time complexity \highlight[cyan]{$O(\text{poly-log}(n))$} and message complexity $\tilde O(com\_all \cdot n + \min  \{ nf, T \})$;
\clst{} has time complexity \highlight[cyan]{$O(\text{poly-log}(n))$} and message complexity $\tilde O(com\_all \cdot n + \min  \{ nf, T \})$.
There are $O(\log n)$ phases in total, and $com\_all$ doubles after each phase.
Since $com\_all=O(\min \{ f\log n, T/n \}+\log n)$ when \committee{} finishes, the total time complexity is \highlight[cyan]{$O(\text{poly-log}(n))$} and the total message complexity is $\tilde O( n + \min \{ nf, T \})$.
\end{proof}

\addcontentsline{toc}{section}{Bibliography}
\bibliographystyle{plain}
\bibliography{podc26-arxiv-v1}

\end{document}